\documentclass[a4paper,twocolumn,11pt,accepted=2025-12-29]{quantumarticle}
\pdfoutput=1
\usepackage[utf8]{inputenc}
\usepackage[english]{babel}
\usepackage[T1]{fontenc}
\usepackage{amsthm}

\usepackage{quantikz}
\usepackage{tikz}
\usepackage{lipsum}
\usepackage{xcolor}
\usepackage{makecell}
\usepackage{placeins}
\usepackage{booktabs}
\usepackage{comment}
\usepackage{float}

\newtheorem{theorem}{Theorem}
\newtheorem{example}{Example}
\newtheorem{assumption}{Assumption}

\newtheorem{lemma}{Lemma}

\newtheorem{manualtheoreminner}{Theorem}
\newenvironment{manualtheorem}[1]{%
  \IfBlankTF{#1}
    {\renewcommand{\themanualtheoreminner}{\unskip}}
    {\renewcommand\themanualtheoreminner{#1}}%
  \manualtheoreminner
}{\endmanualtheoreminner}

\newtheorem{manualassumptioninner}{Assumption}
\newenvironment{manualassumption}[1]{%
  \IfBlankTF{#1}
    {\renewcommand{\themanualassumptioninner}{\unskip}}
    {\renewcommand\themanualassumptioninner{#1}}%
  \manualassumptioninner
}{%
  \endmanualassumptioninner
}

\usepackage[numbers,compress]{natbib}
\usepackage{braket}
\usepackage{mathtools}
\usepackage{algorithm}
\usepackage{algpseudocode}
\usepackage{amssymb}

\usepackage{xcolor}

\usepackage{graphicx}
\usepackage{subcaption}
\newcommand\blfootnote[1]{%
  \begingroup
  \renewcommand\thefootnote{}\footnote{#1}%
  \addtocounter{footnote}{-1}%
  \endgroup
}

\usepackage{xurl}
\usepackage{hyperref}

\usepackage{multirow}

\usepackage{enumitem}

\begin{document}

\title{Harnessing Bayesian Statistics to Accelerate Iterative Quantum Amplitude Estimation
}

\author{Qilin Li\textsuperscript{*}}
\affiliation{Department of Statistics,  University of Wisconsin-Madison, 
1205 University Ave.,
Madison, WI 53706, USA}

\author{Atharva Vidwans\textsuperscript{*}}
\affiliation{Department of Chemistry, University of Wisconsin-Madison, 
1101 University Ave.,
Madison, WI 53706, USA}

\author{Yazhen Wang}
\affiliation{Department of Statistics,  University of Wisconsin-Madison, 
1205 University Ave.,
Madison, WI 53706, USA}

\author{Micheline B. Soley\textsuperscript{\dag}}
\affiliation{Department of Chemistry, University of Wisconsin-Madison, 1101 University Ave., Madison, WI 53706, USA}
\affiliation{Department of Physics, University of Wisconsin-Madison, 1150 University Ave., Madison, WI 53706, USA}
\affiliation{Data Science Institute, University of Wisconsin-Madison, 1205 University Ave., Madison, WI 53706, USA}

\maketitle

\begin{abstract}
     {We establish a unified statistical framework that underscores the crucial role  statistical inference plays in Quantum Amplitude Estimation (QAE), a task essential to fields ranging from chemistry to finance and machine learning. We use this framework to harness Bayesian statistics for improved measurement efficiency with rigorous interval estimates at all iterations of Iterative Quantum Amplitude Estimation. We demonstrate  the resulting method,  Bayesian Iterative Quantum Amplitude Estimation (BIQAE), accurately and efficiently estimates both quantum amplitudes and molecular ground-state energies to high accuracy, and show in analytic and numerical sample complexity analyses that BIQAE outperforms all other QAE approaches considered. Both rigorous mathematical proofs and numerical simulations conclusively indicate Bayesian statistics is the source of this advantage, a finding that invites further inquiry into the power of statistics to expedite the search for quantum utility.}
\end{abstract}


\section{Introduction}

\blfootnote{\textsuperscript{*} These authors contributed equally to this work.}
\blfootnote{\textsuperscript{\dag} Corresponding Author: \href{mailto:msoley@wisc.edu}{msoley@wisc.edu}}

Reduction of the measurement cost of quantum amplitude estimation is essential to accelerate the utility of quantum computers. The inherently probabilistic nature of quantum computing frequently entails a large number of quantum circuit measurements to definitively determine the proportion of a quantum state in a specific target state \cite{brassard2000quantum,grover1998framework,abrams1999fast,wie2019simpler,aaronson2020quantum,nakaji2020faster,suzuki2020amplitude,harrow2020adaptive,arunachalam2022simpler,rall2023amplitude,plekhanov2022variational,giurgica2022low,callison2023improved}, and finding such a proportion is inherent to fundamental mathematical tasks on quantum computers, from integration to optimization \cite{montanaro2015quantum,yu2020practical,rao2020quantum}. These tasks in turn play a foundational role in fields ranging from chemistry \cite{kassal2008polynomial,gunther2024more,baker2020density,johnson2022reducingcostenergyestimation,kunitsa2024experimental} to particle physics \cite{agliardi2022quantum,de2023quantum,miyamoto2024quantum,de2025quantum,martinez2024loop,lee2025quantum, Williams:2025hza}, condensed matter \cite{johri2017entanglement,rall2020quantum,dong2022ground,piroli2024approximating,agrawal2024quantifying} to fluid mechanics \cite{gaitan2020finding,bharadwaj2023hybrid,penuel2024feasibility,gaitan2024circuit}, and finance \cite{stamatopoulos2020option,bouland2020prospects,stamatopoulos2022towards,alcazar2022quantum,herman2023quantum,wang2024option,woerner2019quantum,gomez2022survey,rebentrost2018quantum,miyamoto2022bermudan,egger2020credit,orus2019quantum,braun2021quantum} to machine learning \cite{wiebe2016quantum,wiedemann2023quantumpolicyiteration}. It is therefore expected that advancements to QAE would have broad impact. 


A Bayesian approach that harnesses past information to predict future outcomes is emerging as a force in quantum algorithm development. Bayesian variants of standard frequentist approaches to quantum computing have been proposed in a wide array of quantum computing algorithms, including QAE \cite{PRXQuantum.2.010346,koh2022foundations,2024bayesianquantumamplitudeestimation}, the quantum approximate optimization algorithm  \cite{wang2021quantum,stein2022quantum,tibaldi2023bayesian,kim2023quantum,song2023trainability,he2025self}, phase estimation \cite{wiebe2016efficient,paesani2017experimental,li2018frequentist,o2019quantum,martinez2019adaptive,qiu2021bayesian,van2021efficient,gebhart2021bayesian,qiu2022efficient,smith2024adaptive,yamamoto2024demonstrating,hurant2024few,zhou2024bayesian,direkci2024heisenberg,yamamoto2024demonstrating}, phase difference estimation \cite{sugisaki2021bayesian,sugisaki2022quantum,sugisaki2023projective}, the variational quantum eigensolver \cite{mcclean2016theory,wang2019accelerated,wakaura2021evaluation,iannelli2021noisy,self2021variational,johnson2022reducingcostenergyestimation,nicoli2023physics,sorourifar2024towards,rohrs2024bayesian,jiang2024error,huynh2024variational,pedrielli2025bayesian,he2025self}, and quantum annealing \cite{pelofske2020advanced,finvzgar2024designing}. In QAE specifically, Bayesian statistics has been found to be associated with key developments in driving down computational costs \cite{PRXQuantum.2.010346,koh2022foundations,2024bayesianquantumamplitudeestimation}. The Bayesian approach of Robust Amplitude Estimation (RAE), which employs fixed-length quantum circuits with Bayesian updated parameters, has been found to effectively reduce the impact of so-called estimation dead spots \cite{PRXQuantum.2.010346}. Likewise, the Bayesian approach of Bayesian Amplitude Estimation (BAE), which employs adaptive circuits with a Bayesian-informed scheduler, has successfully yielded the lowest number of oracle calls to reach a given target accuracy (quantum sample complexity) to date \cite{2024bayesianquantumamplitudeestimation}. 

However, the impact of Bayesian statistics itself in QAE measurement cost reduction has historically been difficult to isolate due to changes in the underlying QAE implementation required to inject Bayesian statistics in a noisy quantum environment. Implementation of Bayesian enhanced sampling in RAE coincides with abbreviation of the quantum circuit to avoid quantum decoherence effects \cite{PRXQuantum.2.010346}, and similarly BAE fuses a Bayesian approach with a dynamic grid expansion scheduling strategy \cite{2024bayesianquantumamplitudeestimation}. There is therefore a pressing need for a controlled experiment based on a clear-cut QAE algorithm capable of remaining unchanged except for the injection of Bayesian statistics in order to showcase the true power of what Bayesian inference offers QAE.

We aim to clearly demonstrate the direct impact that Bayesian inference offers to reduce measurement cost in QAE to provide a proof of the quantum speedup offered by Bayesian statistics and to provide conclusive numerical evidence of such a speedup. To do so, we establish a statistical framework for three key QAE methods---Classical Quantum Amplitude Estimation (Classical QAE) \cite{brassard2000quantum},  Amplified Amplitude Estimation (AAE)  \citep{simon2024amplified}, and Iterative Quantum Amplitude Estimation (IQAE) \cite{grinko2021iterative}---in the statistical perspective of asymptotic mean squared error (AMSE) and the normal approximation, and we use the resulting framework to introduce ``Bayesian Iterative Quantum Amplitude Estimation (BIQAE)'' based on the injection of Bayesian statistics into IQAE.  We show the resulting statistical framework quantifies a double-digit percentage speedup of BIQAE relative to IQAE that can be completely attributed to Bayesian inference with rigorous mathematical proofs, and that BIQAE outperforms all other QAE methods considered---including benchmark IQAE and state-of-the-art BAE---over many orders of magnitude in the target accuracy for estimation of arbitrary quantum amplitudes. BIQAE also lowers the quantum sample complexity of ground-state energy approximation relative to benchmark IQAE for all molecular test cases considered over a broad range of equilibrium and nonequilibrium bond distances. These findings set the stage for development of further schemes to capitalize on Bayesian inference in QAE and to develop new approaches to reduce measurement cost in quantum mechanics with a Bayesian rather than a frequentist perspective.

\section{Statistical Framework for QAE}
\label{sec:bg_rw}

\subsection{Preliminaries}
\label{sec:CAI}

We define the QAE problem in terms of the action of a quantum oracle $\mathcal{A}$ on an ($n+1$)-qubit zero state $\ket{0}_n\ket{0}$. The resulting output state may be expressed in terms of two normalized states $\ket{\psi_0}_n$ and $\ket{\psi_1}_n$, which need not be orthogonal, as
\begin{align*}
\mathcal{A}\ket{0}_n\ket{0}=\sqrt{1-a}\ket{\psi_0}_n\ket{0}+\sqrt{a}\ket{\psi_1}_n\ket{1},
\end{align*}
where the oracle $\mathcal{A}$ may be chosen freely without loss of generality with the states $\ket{\psi_0}_n$ and $\ket{\psi_1}_n$ defined accordingly. The objective of QAE is then to estimate the amplitude $a \in[0,1]$ of the quantum state $\ket{\psi_1}_n\ket{1}$.

As a benchmark, Classical QAE (or the ``standard-sampling'' approach) approximates the quantum amplitude $a$ as the sample mean of $n$ measurement outcomes of the quantum circuit associated with oracle \(\mathcal{A}\), where the ancilla states \(|0\rangle\) and \(|1\rangle\) are mapped to the outcomes \(0\) and \(1\), respectively. Where the \textit{quantum sample complexity} is defined by the total number of accesses to the oracle, the quantum sample complexity required to achieve a target accuracy of $\varepsilon$ with Classical QAE follows as on the order of $\mathcal{O}(1/\varepsilon^2)$~\cite{brassard2000quantum}. Restatement of the sample complexity in terms of asymptotic mean-square error (AMSE), as derived 
in Appendix~\ref{appendix:CAI}, then yields a sample complexity required to achieve a mean squared error (MSE) $\varepsilon^2$ of
\begin{align}
\label{eq:complexity_ClassicalQAE}
    N_{\text{oracle}} = \frac{1}{\varepsilon^2} a (1 - a).
\end{align}

The relatively high measurement cost of Classical QAE is addressed via amplified estimation (or the ``enhanced-sampling'' approach), which yields quadratic improvement in the sample complexity \cite{brassard2000quantum}. Let the amplitude $a$ be parametrized as 
\begin{align*}
    a = \sin^2{\theta},
\end{align*}
where $\theta \in \left[0, \frac{\pi}{2}\right]$, and let the Grover operator $\mathcal{Q}$ be defined as
\begin{align*}
    \mathcal{Q}\coloneqq \mathcal{A} \mathcal{S}_0 \mathcal{A}^\dagger \mathcal{S}_{\psi_0},
\end{align*}
where $\mathcal{S}_0 \coloneqq \mathbb{I}_{n+1} - 2\ket{0}_{n+1}\bra{0}_{n+1}$ and $\mathcal{S}_{\psi_0} \coloneqq \mathbb{I}_{n+1} -2\ket{\psi_0}_n\bra{\psi_0}_n \otimes \ket{0}\bra{0}$. Repeated application of the Grover operator $\mathcal{Q}$ $k$ times then results in the amplified state~\cite{brassard2000quantum}
\begin{align}
    \label{eq:fundamental}
    \mathcal{Q}^k \mathcal{A}\ket{0}_n\ket{0}&=\cos{\left((2k+1)\theta\right)}\ket{\psi_0}_n\ket{0}\nonumber\\
    &+ \sin{\left((2k+1)\theta\right)}\ket{\psi_1}_n\ket{1},
\end{align}
such that the \textit{amplified target probability} is
\begin{align}
    \label{eq:pk}
    p_k\coloneqq \sin^2{\left((2k+1)\theta\right)}.
\end{align}
Given that the probability $p_k$ is amplified relative to the original probability $a$, where enhanced sampling estimates $p_k$ with the same accuracy as $a$ and subsequently rescales based on inference of $\theta$, $p_k$ results in a higher-accuracy estimate of $a$ than direct measurement of $a$ itself (see Appendix~\ref{appendix:CAICoin} Example~\ref{appendix_example: coin}). 

Abstraction of the quantum amplitude $a$ in amplified estimation from the amplified state therefore follows an analogous procedure to abstraction of the quantum amplitude $a$ from the original state in Classical QAE---with  an added identifiability challenge. For a fixed number of Grover operators $k$, the circuit $\mathcal{Q}^k \mathcal{A}$ applied to the $(n+1)$-qubit zero state is measured $N_k$ times. Let $X_{i,k}$ represent the outcome of the $i$-th application. The sample mean is calculated as $\overline{X}_k=\frac{1}{N_k}\sum_{i=1}^{N_k} X_{i,k}$. This mean approximates $p_k$, which for known $k$ yields $\theta$, which is in turn used to calculate $a$ according to the pipeline schematic in Fig.~\ref{fig:detour}.
\begin{figure}[t] 
    \centering
    \includegraphics[width=1\linewidth]{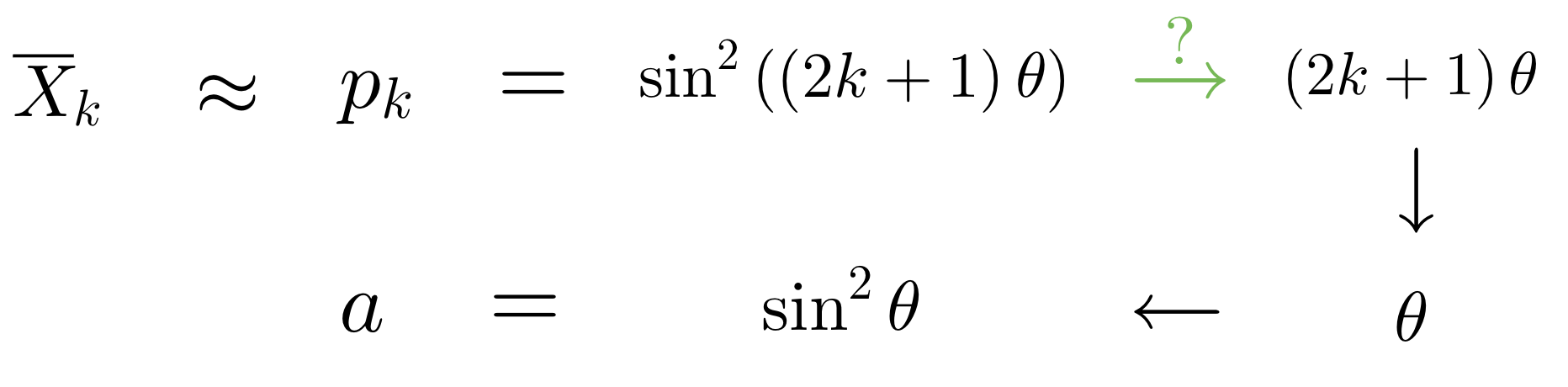}
    \caption{Pipeline of enhanced-sampling amplified estimation with identifiability challenge highlighted in green.}
    \label{fig:detour}
\end{figure}
For each value of $p_k$ there exist $2k+1$ valid solutions of Eq.~\eqref{eq:pk} 
\begin{align}
    \label{eq:mul_sol}
    (2k+1)\theta = 
    \begin{cases} 
        \arcsin{(\sqrt{p_k})} + l\cdot\frac{\pi}{2} & \text{if $l$ is even}, \\
        \arccos{(\sqrt{p_k})} + l\cdot\frac{\pi}{2} & \text{if $l$ is odd},
    \end{cases} 
\end{align} 
where $l = 0, 1, 2, \ldots, 2k$ (see Appendix~\ref{appendix:CAIIdentifiability}). Each choice of the $l$ value and its corresponding possible solution in Eq.~\eqref{eq:mul_sol} leads to a distinct estimate for $a$, each of which is well-reasoned; yet, based solely on information about a single $p_k$ $(k>0)$, it is impossible to determine which possible solution is the true solution. As a result, without further information, $a$ would remain unidentifiable, and additional information is required to resolve this ambiguity, thereby identifying the unique {\em amplified angle} \( (2k+1)\theta \) that corresponds to the true $l$ value for a given $k$. This \textit{quadrant index} \( l(k) \) may be readily interpreted geometrically as the quadrant index of the amplified angle
\begin{align}
    \label{eq:quadrant index}
    l(k) = \left\lfloor \frac{(2k+1)\theta}{(\pi/2)} \right\rfloor,
\end{align}
where \( \lfloor \cdot \rfloor \) denotes the floor function. 

\subsection{Amplified Amplitude Estimation}

The recently developed Amplified Amplitude Estimation (AAE) method \citep{simon2024amplified} resolves the identifiability problem in amplified estimation via the small-angle assumption and can be readily formulated in the statistical framework of AMSE (thus providing a conceptual stepping stone to IQAE and BIQAE), as follows:

To begin, we specify that the AAE approach 
considers cases in which the amplitude $a$ is sufficiently small that there exists only one solution to the identifiability challenge Eq.~\eqref{eq:mul_sol}.
\begin{assumption}
    \label{assump:smallprob}
    For a given $k$, $a$ is sufficiently small that $(2k+1)\theta$ does not exceed $\frac{\pi}{2}$.
\end{assumption}
With Assumption~\ref{assump:smallprob}, the quadrant index $l(k)$ in Eq.~\eqref{eq:mul_sol} must be zero, such that there exists only one valid solution. 

According to the proof outlined in Appendix~\ref{appendix:proof_thm123}, the quantum sample complexity of the approach can then be stated in the statistical framework of AMSE according to Theorem~\ref{thm:AAE}:
\begin{theorem}
    \label{thm:AAE} 
    Suppose Assumption~\ref{assump:smallprob} holds with a certain $k$. Then, 
    \begin{align}
        \label{eq:a0}
        \hat{a}_0\coloneqq\sin^2{\left(\frac{1}{2k+1}\arcsin{\sqrt{\bar{X}_k}}\right)}
    \end{align}
    is the Maximum Likelihood Estimator (MLE) for $a$. The corresponding asymptotic MSE is
    \begin{align*}
        \text{AMSE}_{\hat{a}_0}(a)=\frac{1}{(2k+1)^2}\frac{1}{N_k}a(1-a),
    \end{align*}
    and to achieve a target accuracy $\varepsilon$, the quantum sample complexity is
    \begin{align}
    \label{eq:complexity_AAE}
        N_\text{oracle} =(2k+1) N_k = \frac{1}{2k+1}\frac{1}{\varepsilon^2}a(1-a).
    \end{align}
\end{theorem}

Comparison of the quantum sample complexity of AAE Eq.~\eqref{eq:complexity_AAE} to that of Classical QAE Eq.~\eqref{eq:complexity_ClassicalQAE} then confirms AAE decreases the quantum sample complexity by a factor $\frac{1}{2k+1}$, thereby demonstrating its advantage over Classical QAE and emphasizing that the greater the number of applied $\mathcal{Q}$ operators, the greater the efficiency of the amplified estimator, provided Assumption~\ref{assump:smallprob} holds. Furthermore, in the case where $a=O(\varepsilon)$, the quantum sample complexity becomes $N_\text{oracle}=O\left(\frac{1}{\sqrt{\varepsilon}}\right)$, a conclusion that can be directly derived from Theorem~\ref{thm:AAE} where the order of $k$ is chosen to be $\Theta\left(\sqrt{\varepsilon}\right)$. The complexity $N_\text{oracle}=O\left(\frac{1}{\sqrt{\varepsilon}}\right)$ exceeds Heisenberg-limited scaling in this context but does not refute the Heisenberg limit, as it relies on prior knowledge.

\subsection{IQAE}
\label{sec:IQAE}

We proceed by formulating IQAE in the same statistical framework of AMSE with the additional introduction of the normal approximation, which transparently demonstrates IQAE's estimator for $a$ achieves a quadratic advantage, up to a $\log\log(1/\varepsilon)$ factor, without the need for the small-angle assumption. Note the key insight exploited by IQAE here is that one can obtain the same AMSE and quantum sample complexity as AAE, provided that the quadrant index is known (see Appendix~\ref{appendix:known_quad_idx} for a formal statement), thereby replacing the need for the small-angle assumption.

To begin the statistical derivation of IQAE, we consider that, since $\theta\in\left[ 0,\frac{\pi}{2}\right]$ implies the initial quadrant index must be zero ($l(0)=0$) per the definition of the quadrant index Eq.~\eqref{eq:quadrant index}, we may initiate the AAE process with zero Grover operators. 

Progressively increasing the number $k$ of Grover operators while carefully tracking the evolution of the quadrant index then provides the means to increase the efficiency of the estimator while preserving identifiability, as follows: Formally, let \( 0 = k_0 < k_1 < k_2 < \dots < k_t < \dots < k_{T-1} \) be the sequence of $T$ $k$-values processed in pursuit of the desired accuracy. Beginning with $t=0$, for each $t$, the amplitude is amplified with \( k_t \) Grover operators and \( N_t \) measurements are performed on the corresponding quantum circuit, which constitutes a $k_t$-dependent amplification and measurement process we term {\em stage} $t$.  With the definition \( K_t \coloneqq 2k_t+1 \), which represents the number of accesses to the oracle $\mathcal{A}$ in a single execution of the circuit with $k_t$ applications of the Grover operator, we organize these stages into a sequence \( \{(K_t, N_t):t=0,1,\ldots,T-1\} \) that we term the \textit{schedule} that is designed to ensure the feasibility of tracking the quadrant index at all stages.

To develop a schedule that is both efficient and facilitates statistical analysis, we consider the distinct advantages of schedules previously employed in Maximum-Likelihood QAE (MLQAE) \citep{suzuki2020amplitude} and Iterative QAE (IQAE) \citep{grinko2021iterative}. MLQAE puts forward both linear and exponential \( K \)-schedules with preset values of \( N_t \) for all $t$. Prior numerical results indicate such an exponential schedule with fixed $N_t$ values  achieves a quadratic advantage at the expense of the lack of a theoretical guarantee of its feasibility. IQAE, in contrast, adaptively selects both $K_t$ and $N_t$ such that the quadrant index of each stage $l_t$ is estimated based on measurements carried out in the preceding stage, which continuously refines a confidence interval at confidence level $\alpha$ for $a$ and guarantees that its radius eventually becomes smaller than a prescribed accuracy threshold $\varepsilon$. A rigorous theoretical analysis indicates IQAE achieves optimal complexity up to a multiplicative $\log\log(1/\varepsilon)$ factor; however, derivation of this complexity upper bound is relatively intricate due to IQAE's fully adaptive algorithmic structure, which requires accounting for many possible algorithmic behaviors. 

Based on these points, we develop two hybrid approaches that combine the analytical advantages of MLQAE's exponential $K$-schedule with IQAE's beneficial confidence intervals. 

\subsubsection{Base-3 Hybrid-Scheduled IQAE\label{sec:Base3}}

We first introduce a hybrid $(K,N)$-schedule for IQAE to combine the benefits of exponential $K$-schedules and confidence intervals: a base-3 exponential $K$-schedule 
\begin{align*}
    K_t = 3^t,\quad t=0,1,\ldots,T-1
\end{align*}
with an adaptive $N$-schedule in which the algorithm incrementally performs batches of $N_{\rm incre}$ measurements (termed \textit{iterations}) until the confidence interval for $\theta$ is narrow enough that \( l(k_{t+1}) \) is determined with a desired degree of high confidence $\alpha_t$. We set $\alpha_t=\alpha/T$ such that, by the union bound, the choice of $\alpha_t$ ensures that the final confidence interval for $a$ achieves the overall confidence level $\alpha$. As derived in Appendix~\ref{appendix:proof_T}, to obtain the desired confidence interval for $a$ with a radius no more than $\varepsilon$, the proposed schedule requires the number of stages to be just
\begin{align}
    \label{eq:T}
    T = \left\lceil \log_3 \left( \frac{\pi}{2} \frac{1}{\varepsilon} \sqrt{a(1-a)}\right) \right\rceil.
\end{align}
To aid analysis, the resulting hybrid schedule admits both ternary-expansion and geometric interpretations: Combined with the formula for the quadrant index Eq.~\eqref{eq:quadrant index}, imposition of the hybrid $(K,N)$-schedule ensures \( l(k_t) \) represents the \( t \)-th digit in the ternary expansion of \( \theta/(\pi/2) \). At stage \( t \), the task of identifying the quadrant index is then equivalent to determining the \( (t+1) \)-th digit in this ternary expansion given its first \( t \) digits. Alternatively, from a geometric perspective, the task of identifying the quadrant index is equivalent to determining the \textit{reference interval} to which the amplified angle \( 3^t\theta \) belongs, where the \textit{reference interval} at stage \( t \) is defined as one of the three equal-angle divisions of the quadrant containing \( 3^t\theta \)
\begin{align*}
    l_t\cdot\frac{\pi}{2} + j \cdot \frac{\pi}{6} + \left[0, \frac{\pi}{6}\right], \quad j=0,1,2,
\end{align*}
where $l_t$ is used as an abbreviation of $l(k_t)$, the quadrant index at stage $t$.

To ascertain the quantum sample complexity, we consider the simplest case in which the confidence interval for the amplified angle is relatively broad at each stage relative to the reference intervals:

\begin{assumption}
    \label{assump:schedule3}
    Let $\varepsilon_t$ be the radius of the confidence interval for $\theta$ obtained at stage $t$. Assume that 
    $$\varepsilon_t \geq \frac{\gamma}{2} \frac{1}{3^{t+1}} \frac{\pi}{2}$$ for some constant $\gamma \in ( 0, 1]$, and for all $t = 0, 1, \ldots, T-2$.
\end{assumption}
The hybrid schedule and normal approximation then lead to the following upper bound for the quantum sample complexity (see detailed proof in Appendix~\ref{appendix:proof_schedule3}):
\begin{theorem}
    \label{thm:schedule3}
    Under Assumption~\ref{assump:schedule3}, the accumulated quantum sample complexity of the proposed base-3 hybrid schedule to obtain a confidence interval for $a$ of radius \(\varepsilon\) with an overall confidence level of \(\alpha\) \((\leq 0.1)\) is upper bounded by
    \begin{align*}
    N_{\text{oracle}} 
    &\leq \frac{6}{\pi} \left( \frac{3}{\gamma^2} + 2 \right) \frac{1}{\varepsilon} \log\left( \frac{2}{\alpha} 
    \left\lceil \log_3 \left( \frac{\pi}{4\varepsilon} \right) \right\rceil \right) \\
    &\quad \times 
    \sqrt{a(1-a)}.
\end{align*}
\end{theorem}
The proposed base-3 upper bound thereby achieves the same order quantum sample complexity $\mathcal{O}\left( \frac{1}{\varepsilon}\log{\log{\frac{1}{\varepsilon}}} \right)$ as that of standard-scheduled IQAE \cite{grinko2021iterative}, and thus achieves the optimal order up to a double logarithmic factor. Note the use of the normal approximation in this derivation provides the conceptual and mathematical basis for the creation of Normal-BIQAE, introduced in Sec.~\ref{sec:biqae_normal}.

\subsubsection{Base-3/Base-5 Hybrid-Scheduled IQAE\label{sec:Base3Base5}}

To observe the same $\varepsilon$-dependence of the upper bound of the sample complexity with hybrid-scheduled IQAE in extreme cases in which Assumption~\ref{assump:schedule3}  does not hold, we introduce a second hybrid $(K,N)$-schedule for IQAE that includes another base in its $K$-schedule, as follows: 

In base 5, the reference intervals are
\begin{align*}
    l_t\cdot\frac{\pi}{2} + j \cdot \frac{\pi}{10} + \left[0, \frac{\pi}{10}\right], \quad j=0,1,\ldots,4,
\end{align*}
Simultaneous monitoring of both sets of reference intervals for base 3 and base 5 to determine if either set
completely contains the confidence interval therefore guarantees that
\begin{align}
 \label{eq:epsilont35}
 \varepsilon_t \geq \frac{1}{2}\frac{1}{K_{t}}\frac{1}{15}\frac{\pi}{2}, \quad\text{for $t = 0, 1, \ldots, T-2$},
\end{align}
which reflects the fact that, where the amplified angle falls close to the boundaries of the reference intervals defined for base 3 and thus violates Assumption~\ref{assump:schedule3}, the amplified angle nonetheless rests sufficiently far from the reference interval boundaries in base 5 to avoid entrapment. 

Replacement of Assumption~\ref{assump:schedule3} with Eq.~(\ref{eq:epsilont35}) then yields the accumulated complexity bound for the base-3/base-5 schedule of
\begin{align}
    \label{bound:schedule35}
    N_{\text{oracle}}\leq \frac{227}{\pi} \frac{1}{\varepsilon} 
    \log\left( \frac{2}{\alpha} 
    \left\lceil \log_3 \left( \frac{\pi}{4\varepsilon} \right) \right\rceil \right) \sqrt{a(1-a)},
\end{align}
as formally shown in Appendix~\ref{appendix:35_IQAE}.

The resulting bound thus retains the same order upper bound as base-3 hybrid-scheduled IQAE in Theorem~\ref{thm:schedule3}, and thus the same order as in  standard-scheduled IQAE \cite{grinko2021iterative}, with a difference only in the constant factor.

\subsubsection{IQAE Workflow}

We execute IQAE  according to the workflow depicted in Fig.~\ref{fig:schematic_iqae}, which forms the structural scaffolding for BIQAE. Repeated sets of measurements of the quantum circuit are fed into a classical \texttt{ComputeCI} module to compute the confidence interval for $\theta$ based on all measurements at the current stage, where confidence intervals are computed as normal approximated intervals for hybrid-scheduled IQAE or as Chernoff-Hoeffding or Clopper-Pearson confidence intervals for standard-scheduled IQAE (see Appendix~\ref{appendix:eq_intv}). The \texttt{FindNextK} module then determines the following $K\coloneq2k+1$, as outlined in the aforementioned strategy for hybrid-scheduled IQAE or via a linear search over a candidate set of $K$ values for standard-scheduled IQAE. The updated $K$-value is then employed both to update the form of the quantum circuit and to inform the classical \texttt{ComputeCI} module, and the process continues such that the confidence interval for $a$ continuously shrinks until its radius is smaller than the target accuracy $\varepsilon$. 
Note this workflow only permits IQAE to retain memory of the confidence interval determined in the most recent stage, without preference within said confidence interval.

\begin{figure}[htbp]
    \centering
    \begin{subfigure}{\linewidth}
        \centering
        \includegraphics[width=1.1\linewidth]{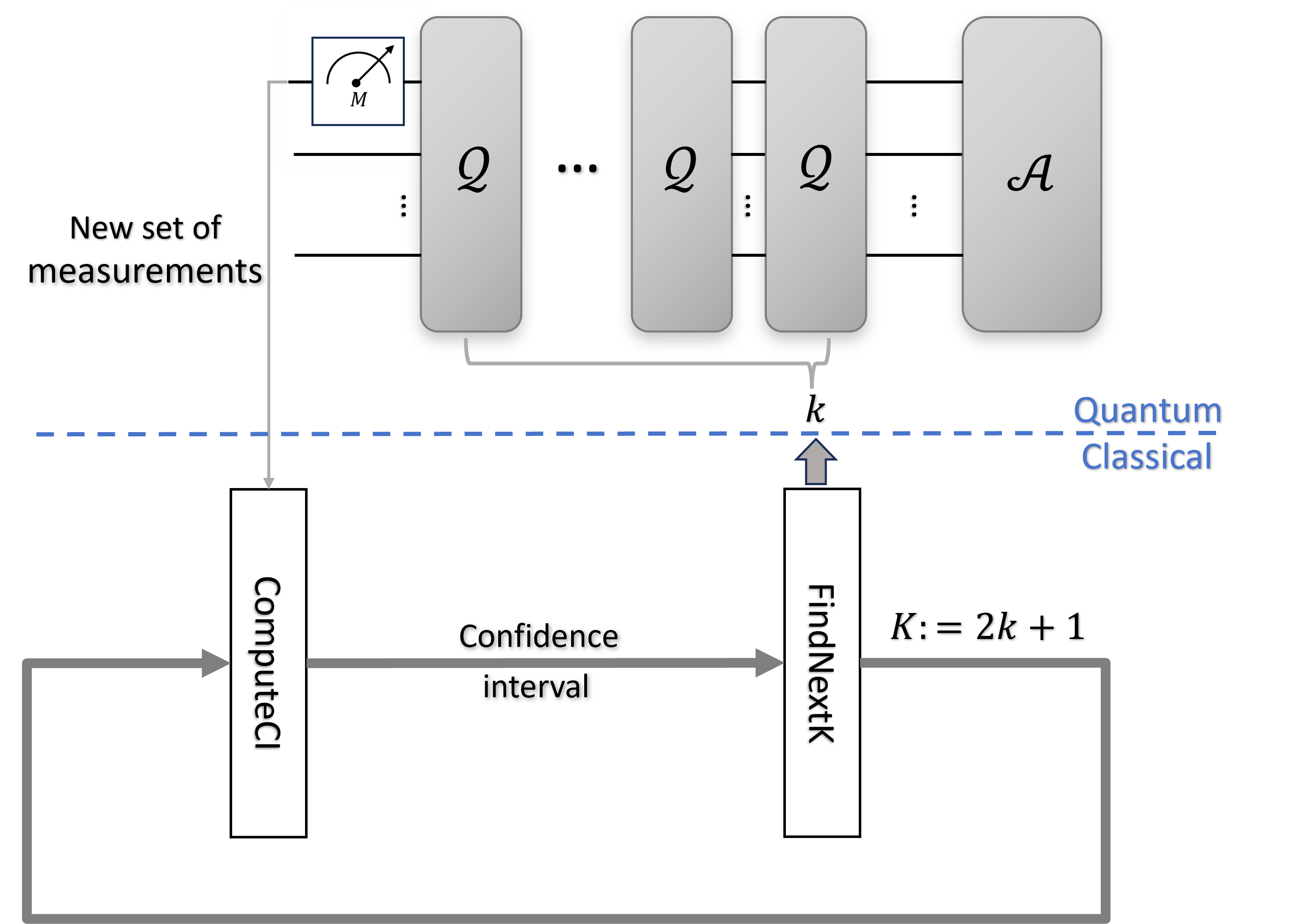}
        \caption{IQAE}
        \label{fig:schematic_iqae}
    \end{subfigure}

    \vspace{3em}

    \begin{subfigure}{\linewidth}
        \centering
        \includegraphics[width=1.1\linewidth]{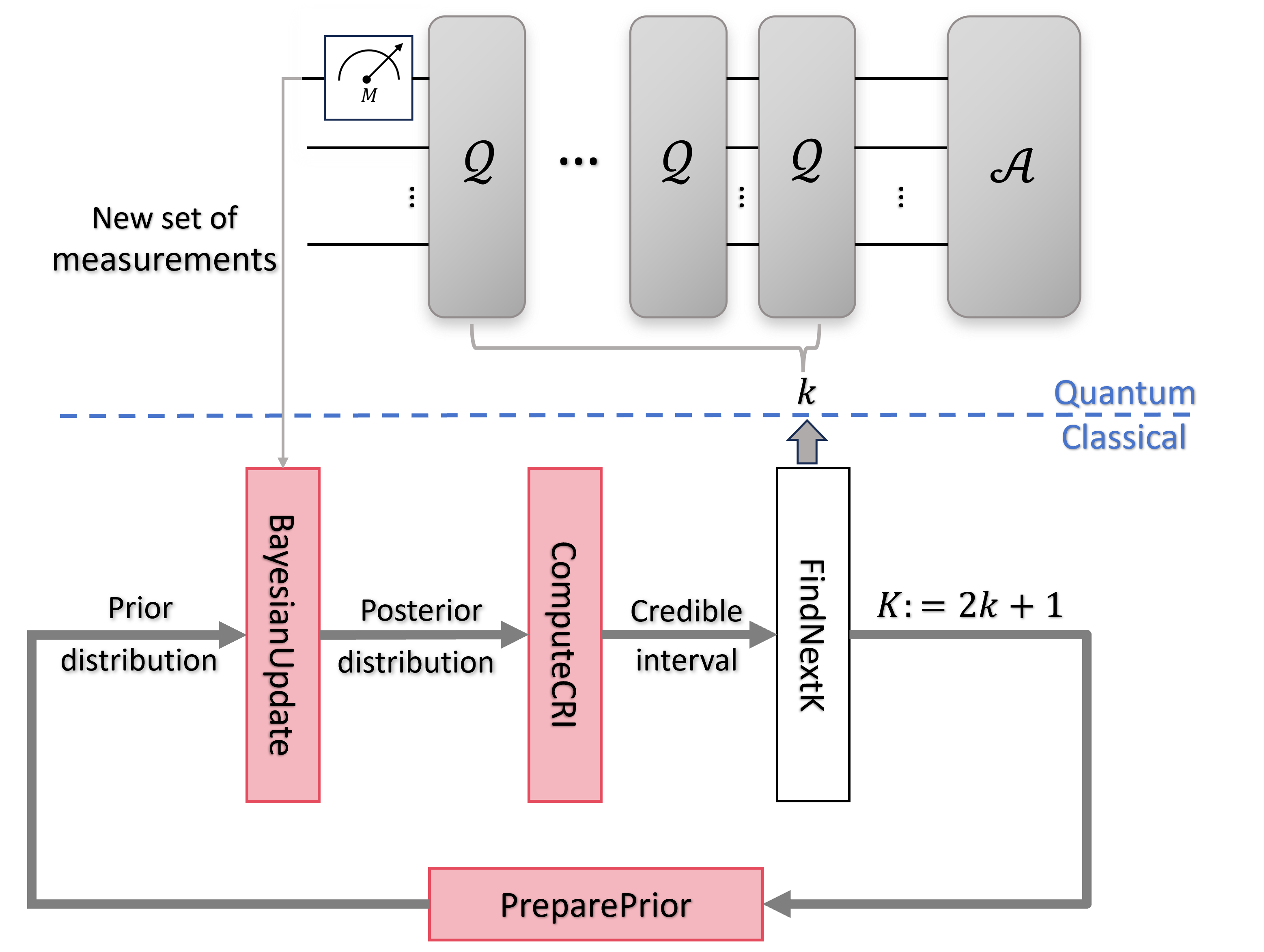}
        \caption{BIQAE}
        \label{fig:schematic_biqae}
    \end{subfigure}

    \caption{Schematic illustrations of IQAE and BIQAE algorithms. (a) In IQAE, measurements of amplified-amplitude quantum circuits (top) are fed to a classical computer (bottom), which both iteratively improves the estimate of the confidence interval of $\theta$ (\texttt{ComputeCI}) and increments the number of Grover operators $k$ of the quantum circuit (\texttt{FindNextK}) until the quantum amplitude is estimated with the desired accuracy. (b) BIQAE introduces three new modules to improve IQAE's efficiency via Bayesian inference (highlighted in red):  \texttt{BayesianUpdate} and \texttt{PreparePrior}, which maintain and refine the prior and posterior distribution of the amplified target probability; and \texttt{ComputeCRI}, the Bayesian substitute to  \texttt{ComputeCI}, which calculates a credible interval for $\theta$ based on the posterior distribution in place of confidence interval.}
    \label{fig:iqae_vs_biqae}
\end{figure}

\section{BIQAE Method}
\label{sec:method}

\subsection{BIQAE Framework}\label{sec:framework}

To harness and refine not only the interval estimate but also the {\em posterior distribution} of the amplitude,
we incorporate Bayesian inference into IQAE as follows:

For each stage~$t$, we consider a prior distribution over $p_{k_t}$ that models understanding of $p_{k_t}$ thus far as
\begin{align}
    \label{eq:prior}
    p_{k_t} \sim \pi_t^{\text{prior}}(p_{k_t}),
\end{align}
where $\pi_t^{\text{prior}}(\cdot)$ denotes a prior distribution whose form remains to be specified. We hold this prior distribution fixed from the conclusion of stage~$t-1$ throughout the remainder stage~$t$.

In each iteration within stage~$t$, we increment $N_t$ by $N_{\rm incre}$ to enhance the accuracy of the estimate. Let $N_t$ denote the total number of measurements gathered at stage~$t$. We map the measurement outcomes to 0 and 1. Then, the distribution of the sample average $\bar{X}_t$ is given by
\begin{align}
    \label{eq:likelihood}
    \bar{X}_t|p_{k_t} \sim \frac{1}{N_{t}}{\rm Bin}\left(N_{t}, p_{k_t}\right).
\end{align}    
We refer to this distribution of the sample average $\bar{X}_t$ given $p_{k_t}$ as the likelihood model and denote the density function by $\mathcal{L}_t\left(\bar{X}_t \middle| p_{k_t}\right)$.

A Bayesian update then integrates the $N_t$ measurements collected at stage~$t$ with the prior distribution to compute a posterior distribution that reflects an updated understanding about $p_{k_t}$. This computation is carried out using Bayes’ rule as
\begin{align}
    \label{eq:posterior}
    \pi_t^{\text{post}}\left(p_{k_t} \middle| \bar{X}_t\right) \propto \pi_t^{\text{prior}}(p_{k_t}) \cdot \mathcal{L}_t\left(\bar{X}_t \middle| p_{k_t}\right),
\end{align}
where $\pi_t^{\text{post}}\left(p_{k_t} \middle| \bar{X}_t\right)$ denotes the posterior distribution.

An interval estimate for the angle $\theta$ is then obtained by first computing an interval estimate for $p_{k_t}$ and then translating this interval to $\theta$. The interval estimate for $p_{k_t}$ is derived from the posterior distribution of $p_{k_t}$ by taking its upper and lower $\alpha_t/2$-quantiles
\begin{align}
\label{eq:ci_pkt}
\left[p_{k_t}^l, p_{k_t}^u\right] 
&= \left[ \left[\Pi_t^{\text{post}}\right]^{-1}\left( \frac{\alpha_t}{2} \right), \right. \nonumber\\
&\left.  \quad\quad\left[\Pi_t^{\text{post}}\right]^{-1}\left( 1 - \frac{\alpha_t}{2} \right) \right],
\end{align}
where $\Pi_t^{\text{post}}$ represents the cumulative distribution function of the posterior distribution. This approach, known as \textit{credible intervals} in Bayesian statistics, replaces the confidence intervals in IQAE, as the latter is incompatible with the posterior distribution.

To obtain the $\theta$ location information required to determine the next $k$, the resulting credible interval for $p_{k_t}$ is then transformed by
\begin{align}
    \label{eq:intv_theta}
    \left[\theta_t^l, \theta_t^u\right]=f_t\left([p_{k_t}^l, p_{k_t}^u]\right),
\end{align}
where $f_t(\cdot)$ denotes the deterministic functional relationship between $\theta$ and $p_{k_t}$
\begin{align*}
    f_t(x) \coloneqq
    \begin{cases} 
        \frac{1}{K_t}\left(\arcsin{(\sqrt{x})} + l_t \frac{\pi}{2}\right) & \text{if \(l_t\) is even}, \\
        \frac{1}{K_t}\left(\arccos{(\sqrt{x})} + l_t  \frac{\pi}{2}\right) & \text{if \(l_t\) is odd}.
    \end{cases}
\end{align*}
To track the evolution of the bases such that $l_t$ is known, let $\left[\theta_{t-1}^l, \theta_{t-1}^u\right]$ denote the interval estimate for $\theta$ from the previous stage. Then, $l_t$ is readily computed as
\begin{align*}
    l_t = \left\lfloor\frac{K_t\theta_{t-1}^l}{\pi/2}\right\rfloor.
\end{align*}

Based on the interval estimate obtained from Eq.~\eqref{eq:intv_theta}, a linear search is performed to find the largest feasible $k$-value as in standard-scheduled IQAE. Denote the $k$-value output as $k_{\rm next}$.  If $k_{\rm next} = k_t$, the $k$-value remains identical, and the algorithm continues with a new iteration. If instead $k_{\rm next} > k_t$, the number of $\mathcal{Q}$ operators in the associated quantum circuit is increased to $k_{t+1}\coloneqq k_{\rm next}$ to improve measurement efficiency. A new stage $t+1$ then begins, which requires preparation of a new prior distribution for $p_{k_{t+1}}$ before resumption of the iteration process.

To prepare the new prior upon identification of a greater feasible $k$-value, 
the posterior distribution for $p_{k_t}$ is transformed into a prior distribution of the target probability in the upcoming stage as
\begin{align}
    \label{eq:p2p}
    p_{k_{t+1}} = \sin^2\left(K_{t+1}\cdot f_t(p_{k_{t}})\right),
\end{align}
namely
\begin{align}
    \label{eq:PreparePrior}
    p_{k_t}\sim \pi_t^{\rm post} 
    \xrightarrow{\sin^2\left(K_{t+1}\cdot f_t(\cdot)\right)} 
    p_{k_{t+1}}\sim \pi_{t+1}^{\rm prior}.
\end{align}
In practice, this transformation is implemented approximately based on the chosen form of prior distribution.

\subsection{BIQAE Workflow\label{sec:biqaeworkflow}}

We execute BIQAE according to the workflow diagram provided in Fig.~\ref{fig:schematic_biqae}, which differs from the workflow of IQAE in (i) the addition of a \texttt{BayesianUpdate} module, which updates the posterior distribution based on previous measurements and the prior distribution according to Bayes' rule, (ii) the insertion of a \texttt{PreparePrior} module, which prepares the prior based on the posterior distribution of the preceding stage, and (iii) the substitution of IQAE's \texttt{ComputeCI} module, which computes frequentist confidence intervals, with BIQAE's \texttt{ComputeCRI} module, which computes  Bayesian credible intervals in BIQAE. Note that whereas IQAE feeds  measurements into the \texttt{ComputeCI} module to compute the confidence interval, BIQAE inputs new measurements into the \texttt{BayesianUpdate} module to refine the posterior distribution of the amplitude $a$ (thereby integrating prior knowledge with new measurements) and subsequently passes the posterior distribution on to \texttt{ComputeCRI} to compute the credible interval. This flow of information ensures the algorithm learns from both the previous interval estimate and the shape of the distribution within this interval to further increase the measurement efficiency of the quantum amplitude estimate.

\subsection{Implementation of BIQAE}
\label{sec:biqaeimpl}

We introduce two approaches to implement the aforementioned Bayesian framework: (i) BIQAE with a normal prior (Normal-BIQAE) for theoretical insights into quantification of the improvement brought about by injection of Bayesian inference into IQAE and (ii) BIQAE with a beta prior (Beta-BIQAE) for superior practical performance. 
    
\subsubsection{Normal-BIQAE}
\label{sec:biqae_normal}

\paragraph{Prior.} To model the probability distributions in BIQAE according to a normal prior distribution, we adopt as the prior Eq.~\eqref{eq:prior} a distribution of the form
\begin{align}
    \label{eq:prior_normal}
    p_{k_t} \sim \mathcal{N}(\mu_{0, t}, \Delta_t^2)
\end{align}
for normal distribution $\mathcal{N}$, prior mean $\mu_{0, t}$, and prior variance $\Delta_t^2$ at stage $t$.

\paragraph{\texttt{BayesianUpdate.}} 
To leverage the associated conjugate normal prior, we approximate the likelihood model in Eq.~\eqref{eq:likelihood} in \texttt{BayesianUpdate} by the normal distribution
\begin{align}
    \label{eq:likelihood_normal}
    \bar{X}_t|p_{k_t} \sim \mathcal{N}\left(p_{k_t}, \frac{1}{N_{t}}p_{k_t}(1-p_{k_t})\right).
\end{align}
Combining Eqs.~\eqref{eq:posterior}, \eqref{eq:prior_normal}, and \eqref{eq:likelihood_normal}, the posterior distribution  follows as 
\begin{align}
    p_{k_t}|\bar{X}_t &\sim \mathcal{N}\left(\mu_{\text{post},t}, \sigma^2_{\text{post},t} \right), \nonumber \\
    \mu_{\text{post},t} &= \frac{\frac{\mu_{0,t}}{\Delta_t^2}+\frac{\bar{X}_t}{\hat{\sigma}_t^2/N_{t}}}{\frac{1}{\Delta_t^2}+\frac{1}{\hat{\sigma}_t^2/N_t}}, \nonumber \\
    \sigma^2_{\text{post},t} &= \left(\frac{1}{\Delta_t^2}+\frac{1}{\hat{\sigma}_t^2/N_{t}}\right)^{-1},
    \label{eq:posterior_normal}
\end{align}
where $\hat{\sigma}_t^2\coloneqq \bar{X}_t(1 - \bar{X}_t)$ is an asymptotic unbiased estimator for the variance.

\paragraph{\texttt{ComputeCRI.}} Since the resulting posterior distribution is normal, the credible interval for $p_{k_t}$ Eq.~\eqref{eq:ci_pkt} in the \texttt{ComputeCRI} module becomes
\begin{align}
    \label{eq:ci_pkt_normal}
    \left[p_{k_t}^l, p_{k_t}^u\right]=\mathrm{trunc}_{[0,1]}(\mu_{\text{post},t} \pm z_{\alpha_t/2} \cdot \sigma_{\text{post},t}),
\end{align}
where \(z_{\alpha_t/2}\) is the \((1 - \alpha_t/2)\)-quantile of the standard normal distribution. 

\paragraph{\texttt{PreparePrior.}}
An approximation is then required to prepare a normal prior from the normal posterior distribution  Eq.~\eqref{eq:posterior_normal}  since, although the posterior distribution is normal and the posterior-to-prior transformation Eq.~(\ref{eq:PreparePrior}) is smooth, the exact prior distribution after transformation is no longer a normal distribution. Assuming $N_t$ is sufficiently large that the normality in the central limit theorem applies, we thus approximate the prior distribution for $p_{k_{t+1}}$ in the \texttt{PreparePrior} module by the normal distribution according to the delta method (see Appendix~\ref{appendix:biqae_normal}), which yields
\begin{align*}
    p_{k_{t+1}} & \sim \mathcal{N}\left( \mu_{0, t+1}, \Delta_{t+1}^2 \right), \\
    \mu_{0, t+1} & = \sin^2\left(K_{t+1} \cdot f_t(\mu_{\text{post},t})\right) ,\\
    \Delta_{t+1}^2 & = \left(\frac{K_{t+1}}{K_t}\right)^2 \frac{\mu_{0, t+1}(1 - \mu_{0, t+1})}{\mu_{\text{post},t}(1 - \mu_{\text{post},t})} \cdot \sigma^2_{\text{post},t}.
\end{align*}

\paragraph{Quantum Sample Complexity.} Where $N_t$ is sufficiently large that the normality in central limit theorem applies, we derive for the resulting Normal-BIQAE implementation an upper bound on the quantum sample complexity at stage $t$ required to obtain an interval estimate for $\theta$ with radius $\varepsilon_t$ of
\begin{align}
    \label{complexity_bound}
    N_{\text{oracle},t} \leq \frac{z^2_{\alpha_t/2}}{4K_t}\left(\frac{1}{\varepsilon_t^2} - \frac{1}{\varepsilon_{t-1}^2} \right),
\end{align}
where $\varepsilon_{t-1}$ is the radius of the interval estimate for $\theta$ obtained at stage $t-1$ (see Appendix~\ref{appendix:complexity_normal}). Compared to the upper bound of Normal-IQAE (see Appendix~\ref{appendix:proof_schedule3})
\begin{align}
    \label{complexity_bound_iqae}
    N_{\text{oracle},t} = \frac{z^2_{\alpha_t/2}}{4K_t}\frac{1}{\varepsilon_t^2},
\end{align}
Normal-BIQAE therefore reduces the complexity by a factor of \( \varepsilon_t^2 / \varepsilon_{t-1}^2 \). 

\subsubsection{Beta-BIQAE}
\label{sec:biqae_beta}

\paragraph{Prior.} 
In Beta-BIQAE, we assume a prior Eq.~\eqref{eq:prior} in the form of a beta distribution
\begin{align}
    \label{eq:prior_beta}
    p_{k_t} \sim {\mathrm {Beta}}(a_{0, t}, b_{0,t}),
\end{align} 
where $a_{0, t}, b_{0,t}$ are the shape parameters. 

\paragraph{\texttt{BayesianUpdate.}} A beta-distributed prior allows us to directly use the binomial distribution Eq.~\eqref{eq:likelihood} to derive a beta-distributed posterior, which avoids the need for normal approximation. Using Eqs.~\eqref{eq:likelihood}, \eqref{eq:posterior}, and \eqref{eq:prior_beta}; the posterior distribution is 
\begin{align}
    p_{k_t}|\bar{X}_t & \sim {\mathrm {Beta}}(a_{\text{post}, t}, b_{\text{post}, t}), \nonumber\\
    a_{\text{post}, t} & = a_{0, t} + N_t\bar{X}_t, \nonumber\\
    b_{\text{post}, t} &= b_{0, t} + N_t\left(1-\bar{X}_{t}\right).
    \label{eq:posterior_beta}
\end{align}

\paragraph{\texttt{ComputeCRI.}} The interval estimate Eq.~\eqref{eq:ci_pkt} in the \texttt{ComputeCRI} module then follows as 
\begin{align*}
\left[p_{k_t}^l, p_{k_t}^u\right] 
&= \left[ \text{BetaCDF}^{-1}\left(\frac{\alpha_t}{2}, a_{\text{post}, t}, b_{\text{post}, t}\right), \right. \\
&\quad \left. \text{BetaCDF}^{-1}\left(1 - \frac{\alpha_t}{2}, a_{\text{post}, t}, b_{\text{post}, t}\right) \right]
\end{align*}
to accommodate the beta-distributed posterior, where $\text{BetaCDF}^{-1}(\cdot,a, b)$ is the inverse of the cumulative density function of the beta distribution with shape parameters $a$ and $b$. 

\paragraph{\texttt{PreparePrior.}}Analogously to Normal-BIQAE, given that the exact prior distribution of $p_{k_{t+1}}$ after the transformation Eq.~\eqref{eq:PreparePrior} is no longer a beta distribution, an appropriate beta distribution is required to approximate the true distribution. However, unlike Normal-BIQAE, Beta-BIQAE does not admit use of the delta method to enable the derivation of a closed-form approximation. We therefore employ a sampling procedure in \texttt{PreparePrior} to obtain an asymptotically optimal beta approximation of the exact prior distribution, where optimality is defined as minimization of the distance (Kullback–Leibler divergence) between the chosen approximation and the exact prior (see Appendix~\ref{appendix:biqae_beta} for a rigorous establishment of this optimality result). 

Specifically, given the posterior distribution derived in Eq.~\eqref{eq:posterior_beta}, we prepare the prior as follows:
\begin{enumerate}
    \item Draw a sufficiently large sample size $R$ (for example, $R=1000$) of independent and identically distributed (i.i.d.) values \(\{y_1, y_2, \ldots, y_R\}\) from the beta posterior in Eq.~\eqref{eq:posterior_beta} to accurately approximate the underlying distribution.
    \item Map each value \(y_i\) to the parameter space for \(p_{k_{t+1}}\) through the transformation described in Eq.~\eqref{eq:PreparePrior} to obtain
    \[
    \tilde{y}_i \coloneqq \sin^2\left(K_{t+1} \cdot f_t(y_i)\right)
    \]
    where \(\{\tilde{y}_1, \tilde{y}_2, \ldots, \tilde{y}_R\}\) is identical to a sample drawn from the exact prior distribution of $p_{k_{t+1}}$.
    \item Treat \(\{\tilde{y}_1, \tilde{y}_2, \ldots, \tilde{y}_R\}\) as if they were drawn from a beta distribution and use the Maximum Likelihood Estimator (MLE) to infer its shape parameters. \item Use the maximum likelihood estimates as the shape parameters $(a_{0, t+1}, b_{0,t+1})$ of the beta distribution Eq.~\eqref{eq:prior_beta}.
\end{enumerate}

The chosen sampling procedure is performed only at transitions between stages, namely, a total of $T-1$ times where $T=\mathcal{O}\left(\log(1/\varepsilon)\right)$ 
and requires less than a second of classical computing time per application, thereby circumventing the high classical cost associated with Monte Carlo techniques that are required where the form of the prior is not specified to allow conjugacy \cite{2024bayesianquantumamplitudeestimation}.

In practice, the per-stage sample size $N_t$ employed in implementations of IQAE and BIQAE is typically quite small (often $N_t<100$). Whereas Normal-BIQAE depends on a normal approximation, such that as $N_t$ decreases  the performance of Normal-BIQAE deteriorates (namely, both the quantum sample complexity and the loss of nominal coverage at level $1-\alpha$ increase, see Appendix~\ref{appendix:beta_vs_normal}); Beta-BIQAE makes no such approximation, such that Beta-BIQAE is capable of handling a broader range of sample sizes, including the practical small sample regime (see Section~\ref{sec:math_simul} and Appendix~\ref{appendix:beta_vs_normal}).

\subsubsection{BIQAE with a Noninformative Prior\label{sec:connections}}

We further recognize that, although not originally derived according to a Bayesian perspective, IQAE constitutes implementation of BIQAE with a noninformative prior, as derived in detail in Appendix~\ref{appendix:eq_intv}. Precisely, IQAE with Chernoff-Hoeffding confidence intervals (termed Normal-IQAE or IQAE-CH) corresponds to BIQAE with a noninformative normal prior, and IQAE with Clopper-Pearson confidence intervals (termed Beta-IQAE or IQAE-CP)  corresponds to BIQAE with a noninformative beta prior.

\begin{figure*} 
    \centering
    \includegraphics[width=\textwidth]{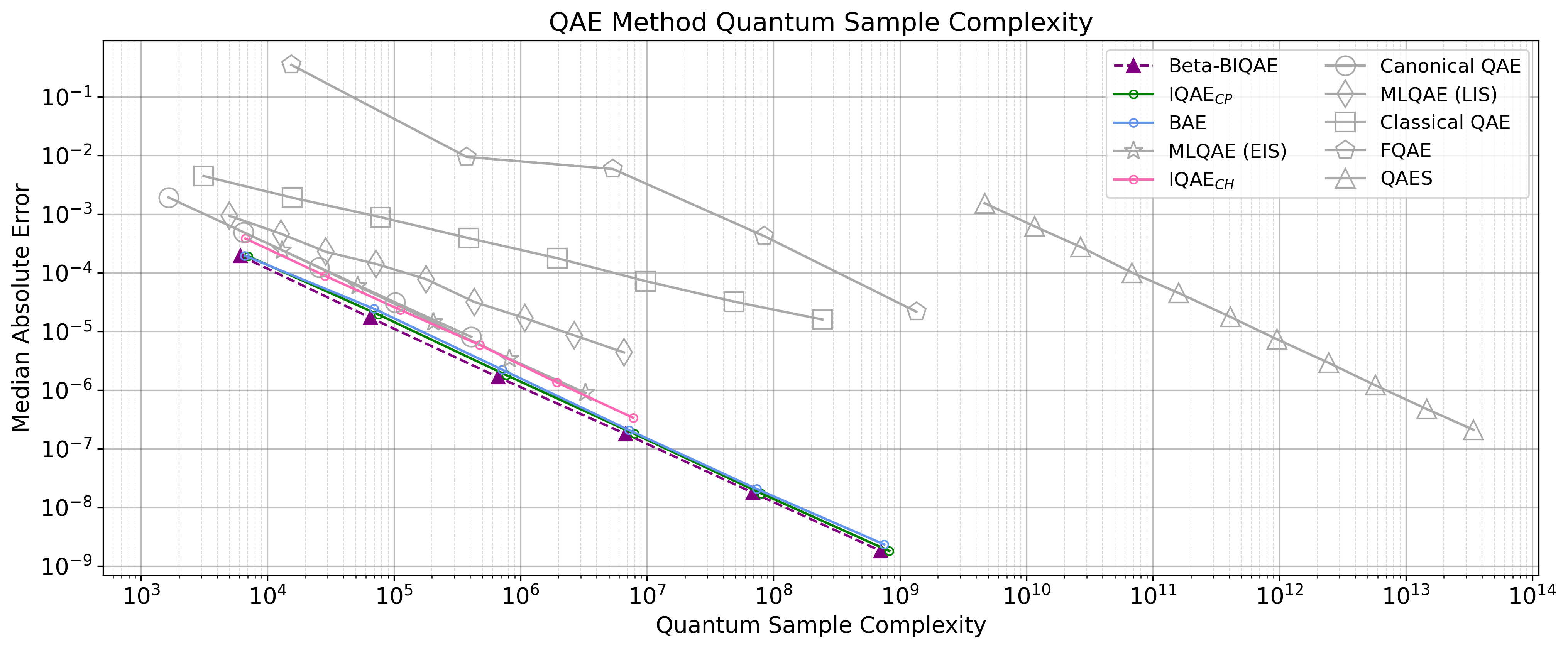}
    \caption{Quantum sample complexity comparison of BIQAE to benchmark QAE approaches for quantum amplitude $a=0.5$ over 200 repetitions. Legend is ordered from highest to lowest performance, with Bayesian and IQAE methods emphasized in color (BIQAE, purple triangles; IQAE-CP, green circles; BAE, blue circles; and IQAE-CH, pink circles).}
    \label{fig:compare_all}
\end{figure*}

\subsection{Computational Methods}

To evaluate the broad applicability and measurement efficiency of BIQAE, we conduct numerical experiments on the ability of BIQAE to estimate both amplitudes of specified quantum states and ground-state energies of diatomic and triatomic molecules. We have made code for BIQAE and all computational tests readily accessible via Github at \url{http://www.github.com/Kirin0570/BIQAE} \cite{qilin2025biqae}. For benchmark methods, we employ the open-source code \cite{bae_repo} that accompanies ref.~\citep{2024bayesianquantumamplitudeestimation} with  minor modifications for output formatting and storage. 

To examine the performance of BIQAE for quantum state amplitude estimation relative to existing QAE techniques, we compare the quantum sample complexity of Beta-BIQAE, Canonical QAE \cite{nielsen2010quantum}, Classical QAE \cite{brassard2000quantum}, IQAE \cite{grinko2021iterative}, BAE \cite{2024bayesianquantumamplitudeestimation}, MLQAE \cite{suzuki2020amplitude}, FQAE (Faster Quantum Amplitude Estimation) \cite{nakaji2020faster}, and QAES (QAE, Simplified) \cite{aaronson2020quantum}. Experiments are performed with the standardized quantum circuit design presented in Appendix~\ref{appendix:circ_math_simul} at a variety of target accuracies, with each experiment independently repeated with a stated number of \textit{repetitions} to compute summary statistics of the performance metrics. Given that BIQAE, IQAE, FQAE and Classical QAE are able to produce interval estimates whereas BAE, MLQAE, and Canonical QAE provide only point estimates; we simulate all approaches based on confidence intervals with a confidence level of \(\alpha=0.05\) during execution of the algorithm and subsequently compare all QAE methods as a whole from a point estimation perspective in which we reduce the confidence interval to a point estimate by taking its center. To further aid direct comparison with literature results for IQAE \cite{grinko2021iterative}, (i) we adopt as the weight of each shot of the quantum circuit the number of the applications of the Grover operator (\(k_t\))  instead of the number of accesses to the oracle $\mathcal{A}$ (\(K_t =2k_t+1\)) 
and (ii) in cases where the quantum amplitude is prespecified, we consider a value of $a=0.5$---an amplitude that is simultaneously expected to yield the worst-case sample complexity because all derived bounds exhibit dependence on $a(1-a)$ either directly or through its square root. A hyperparameter of incremental shot batch size $N_\text{incre}=10$ is employed for all relevant experiments with IQAE and BIQAE.

To assess the accuracy and efficiency of BIQAE for molecular ground-state energy approximation, we employ BIQAE to approximate the energy of the ground-state wavefunction of four molecules --
H\textsubscript{2}, LiH, HF, and BeH\textsubscript{2}. For each molecule, the electronic Hamiltonian is generated at the Hartree-Fock level of electronic structure theory in the Born-Oppenheimer approximation with frozen core orbitals, Jordan-Wigner mapping, and qubit tapering \cite{jordan1928uber,bravyi2017tapering,ryabinkin2018constrained,sun2018pyscf,mcardle2020quantum,qiskit2024}. The equilibrium bond length is assumed unless otherwise stated. The ground-state wavefunction is then prepared using the variational quantum eigensolver, as detailed in  Appendix~\ref{appendix:vqe_exps}, prior to QAE for expectation value estimation. Where the accuracy is fixed, we compare the quantum sample complexity of Classical QAE, IQAE-CP, and Beta-BIQAE. Where the quantum sample complexity is  fixed, we compare Classical QAE and Beta-BIQAE using the same number of shots in Classical QAE as that needed to reach the desired target accuracy for each Pauli string $\varepsilon=10^{-3}$ with Beta-BIQAE. We further restrict the target accuracy to $\varepsilon=10^{-4}$ for HF given that the error entailed by energy estimation grows as the number and coefficient magnitude of Pauli terms increases. All molecular results represent averages over 200 repetitions.

\section{Numerical Simulations\label{sec:results}}

\subsection{Quantum Sample Complexity Scaling Analysis}\label{sec:math_simul}

Quantum sample complexity analysis indicates Beta-BIQAE outperforms all other QAE methods considered for estimation of the benchmark quantum amplitude $a=0.5$. As shown in Fig.~\ref{fig:compare_all}, Beta-BIQAE requires a lower quantum sample complexity to estimate the quantum amplitude than BAE, IQAE (-CP and -CH), Canonical QAE, MLQAE (EIS [Exponentially Incremental Schedule] and LIS [Linearly Incremental Schedule]), Classical QAE, FQAE, and QAES; an advantage that persists over five orders of magnitude in the specified median absolute error. Linear regression analysis detailed in Appendix~\ref{ext_res:scaling_analysis} further indicates Beta-BIQAE's scaling (slope $-1.0088$, intercept $0.0211$, fit $R^2=0.9999$) robustly offers the expected quadratic improvement over Classical QAE's scaling (slope $-1.9821$, intercept $-1.1693$, fit $R^2=0.996$). Beta-BIQAE's scaling also exhibits an intercept that is a multiplicative factor of five lower than the leading alternative, IQAE-CP.
\footnote{Note that whereas Beta-BIQAE outperforms all other QAE methods considered in terms of quantum sample complexity, Normal-BIQAE outperforms all other QAE methods considered with the exception of IQAE-CP. The more favorable performance of Beta- over Normal-BIQAE is consistent with the fact that the normal approximation (and the accuracy of variance estimation) that underlies Normal-BIQAE may become inaccurate for small sample sizes.}

The superior scaling of Beta-BIQAE relative to IQAE-CP is found to be robust to different target accuracies. As shown in Fig.~\ref{fig:loglog_plot_45}, Beta-BIQAE improves upon the measurement cost of IQAE-CP by a near-constant factor of approximately 14\% across six orders of magnitude in the target accuracy (see also interval radius and coverage rate comparison of Beta-BIQAE, IQAE-CP, and BAE in Appendix~\ref{appendix:biqae_vs_bae}). This constant factor improvement is consistent with the predicted upper bound on the sample complexity Eq.~\eqref{complexity_bound} for Normal-BIQAE (in the absence of an analytical result for Beta-BIQAE), as the average interval radius ratio $\varepsilon_{t-1}/\varepsilon_{t}$ observed for the majority of the Beta-BIQAE process remains approximately $2.5$ in Fig.~\ref{fig:radius_ratio}, which corresponds to a 16\% improvement.\footnote{We note that the observed interval radius ratio remains relatively unchanged until the algorithm's final stages, in which the ratio falls. The change observed in the final stages suggests that initial progression with a linear search strategy fluctuates about an exponential $K$-schedule (in which a base-$b$ hybrid schedule implies an accuracy improvement rate per stage and concomitant radius ratio of $b$) and final stages satisfy the stopping criterion (namely, the final stage of each repetition is terminated early before the interval estimate becomes sufficiently small to allow advancement to the next larger $K$-value, which contributes a smaller ratio value to the average).} 

The double-digit percentage improvement of Beta-BIQAE relative to IQAE-CP is also found to persist across the full range of possible quantum amplitudes $a\in[0,1]$, as evidenced by Fig.~\ref{fig:all_angle_plot_epsilon_1e-08} for a target accuracy of $\varepsilon=10^{-8}$. A consistent improvement in the quantum sample complexity of Beta-BIQAE relative to IQAE-CP of approximately 12\% to 15\% is visible for all $a$ values considered. Additionally, close fitting of the quantum sample complexity as a function of $a$ to the model $N_\text{oracle} = \beta\sqrt{a(1-a)}$ (where $\beta$ is a constant depending on $\varepsilon$) reveals the observed complexities are consistent with the theoretical expectation detailed in Eqs.~(\ref{bound:schedule35}) and~(\ref{complexity_bound}) 
that the complexities of IQAE and BIQAE follow the decomposition
    \begin{align}\label{eq:sharedscaling}
        N_\text{oracle} = C \cdot\frac{1}{\varepsilon} (\log\frac{1}{\alpha} + \log\log\frac{1}{\varepsilon})\cdot\sqrt{a(1-a)},
    \end{align}
where $C$ is a constant independent of $a$ that is lower in BIQAE than IQAE and where the effect of differing input amplitude values is entirely captured in the multiplicative factor $\sqrt{a(1-a)}$. In Appendix~\ref{appendix:biqae_vs_iqae}, we demonstrate that this favorable scaling of Beta-BIQAE relative to IQAE-CP for arbitrary $a$ values holds over six orders of magnitude $\varepsilon\in[10^{-2},10^{-7}]$.

    \begin{figure} 
        \includegraphics[width=0.5\textwidth]{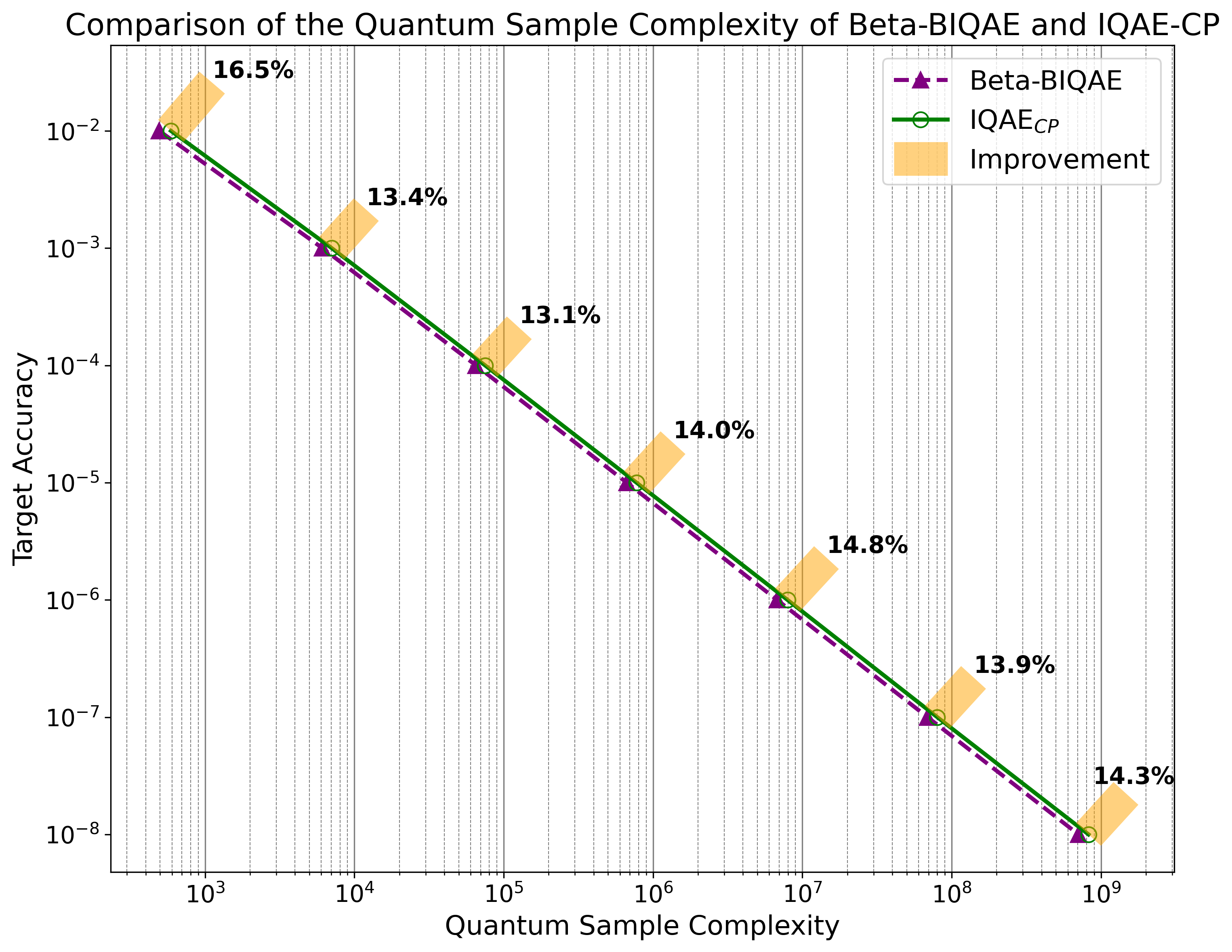}
        \caption{Quantum sample complexity comparison between Beta-BIQAE (dashed purple line with triangles) and IQAE-CP (solid green line with circles) across six orders of magnitude of target accuracy $\varepsilon$ (percentage improvements of Beta-BIQAE over IQAE-CP accentuated in text and as yellow bars).}
        \label{fig:loglog_plot_45}
    \end{figure}
    
    \begin{figure}
        \begin{center}
            \includegraphics[width=0.48\textwidth]{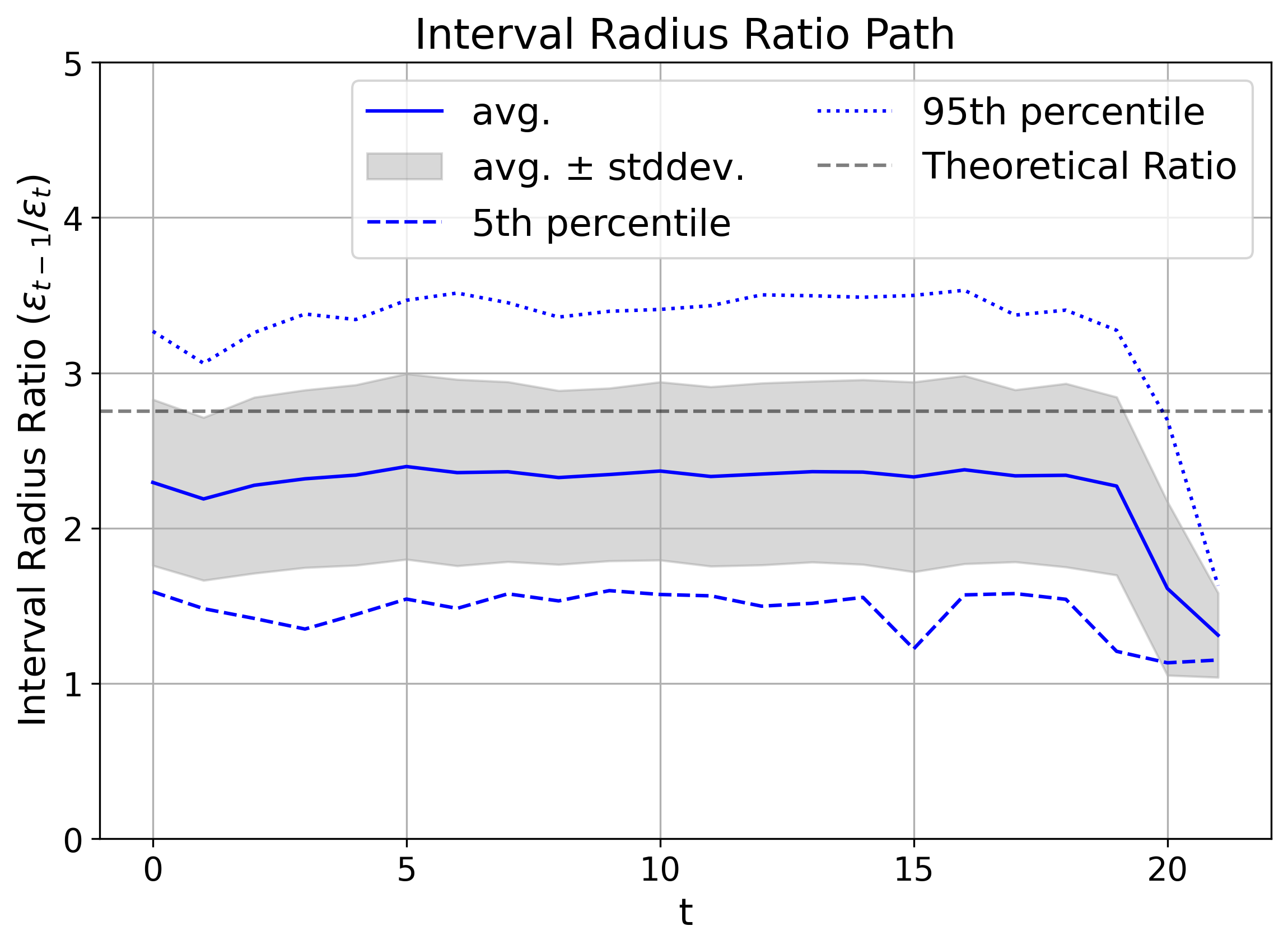}
        \end{center}
        \caption{Radius ratio progression across stages in Beta-BIQAE (average value, solid blue line; values within one standard deviation, shaded gray area; $5^\text{th}$ percentile, dashed blue line; $95^\text{th}$ percentile, dotted blue line). Results are aggregated from 1000 repetitions with a target accuracy of $\varepsilon=10^{-8}$. A reference line (dashed gray line) is included to mark the theoretically expected ratio from Eq.~\eqref{complexity_bound} corresponding to the observed 14.3\% improvement in the quantum sample complexity analysis.}
        \label{fig:radius_ratio}
    \end{figure}
    
    \begin{figure}
        \centering
        \includegraphics[width=0.5\textwidth]{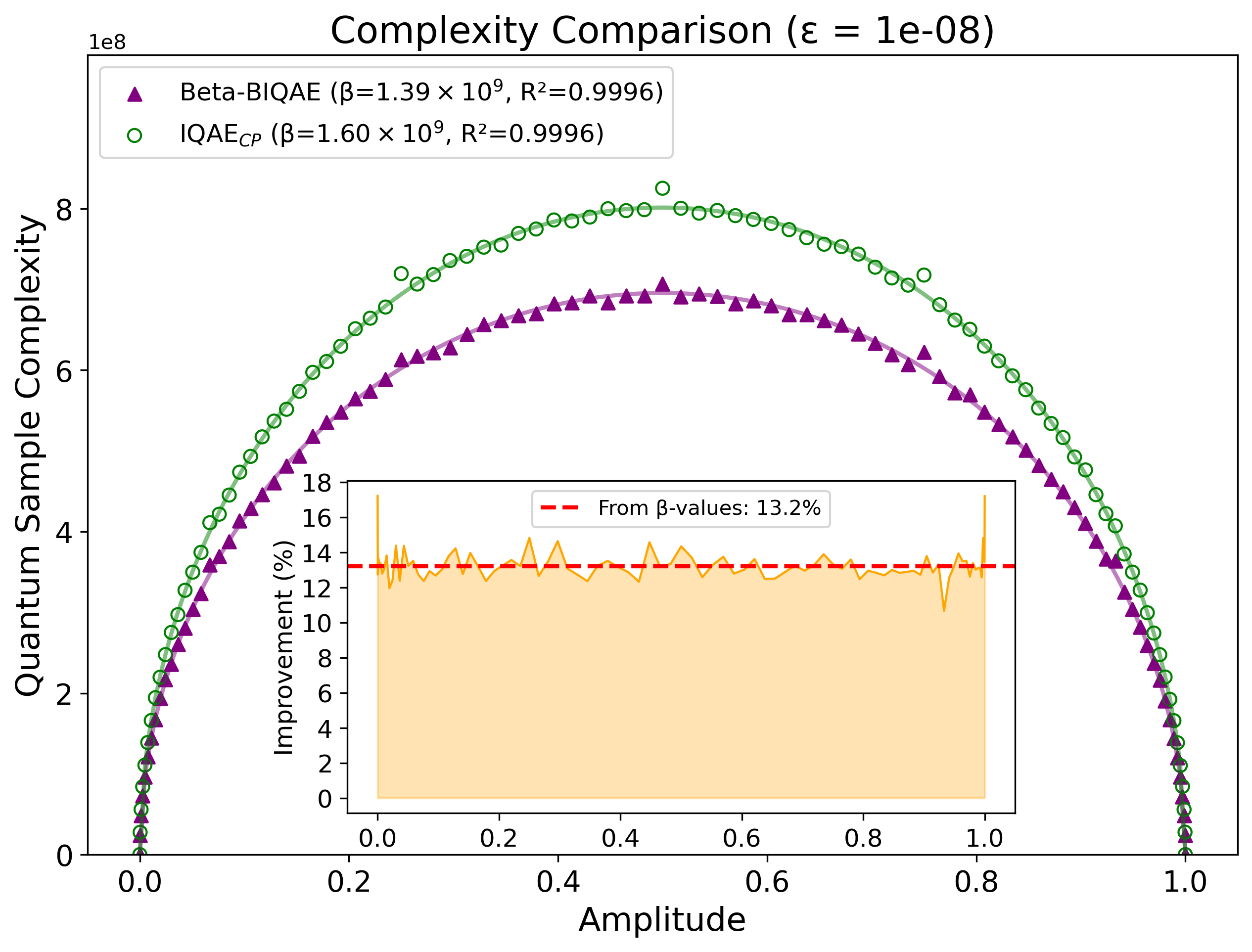}
        \caption{Comparison of the quantum sample complexities of Beta-BIQAE (purple triangles with fit solid line) and IQAE-CP (green circles with fit solid line) across quantum amplitudes $a\in[0,1]$ at a fixed target accuracy $\varepsilon=10^{-8}$, averaged over 1000 repetitions. Insets highlight Beta-BIQAE's consistent percentage improvement over IQAE-CP (yellow solid line values with red dotted line from fitted $\beta$).}
        \label{fig:all_angle_plot_epsilon_1e-08}
    \end{figure}


\color{black}
\begin{figure*}[t]
    \centering
    \begin{tabular}{cc}  
        \includegraphics[width=0.5\textwidth]{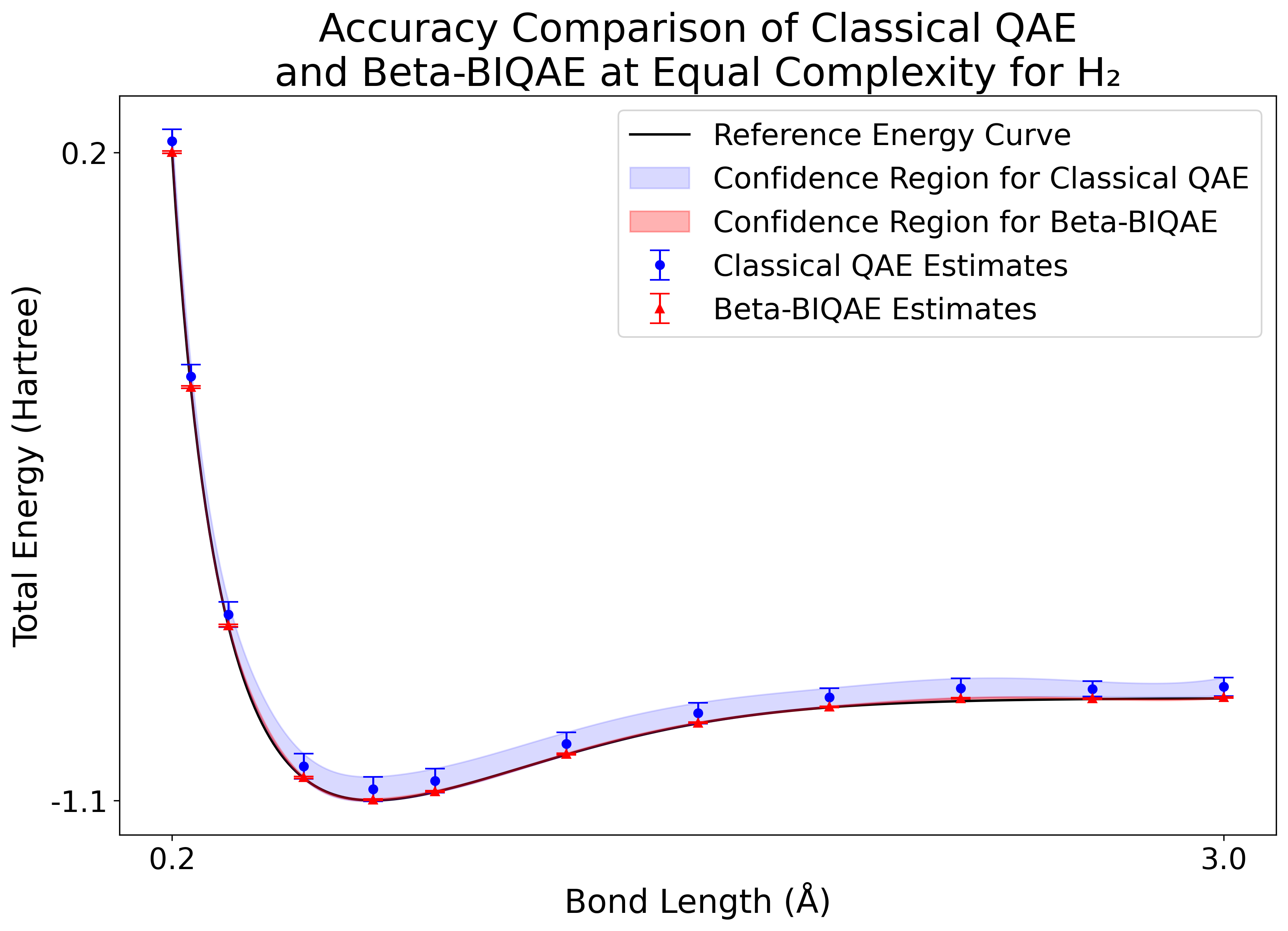} &  
        \includegraphics[width=0.5\textwidth]{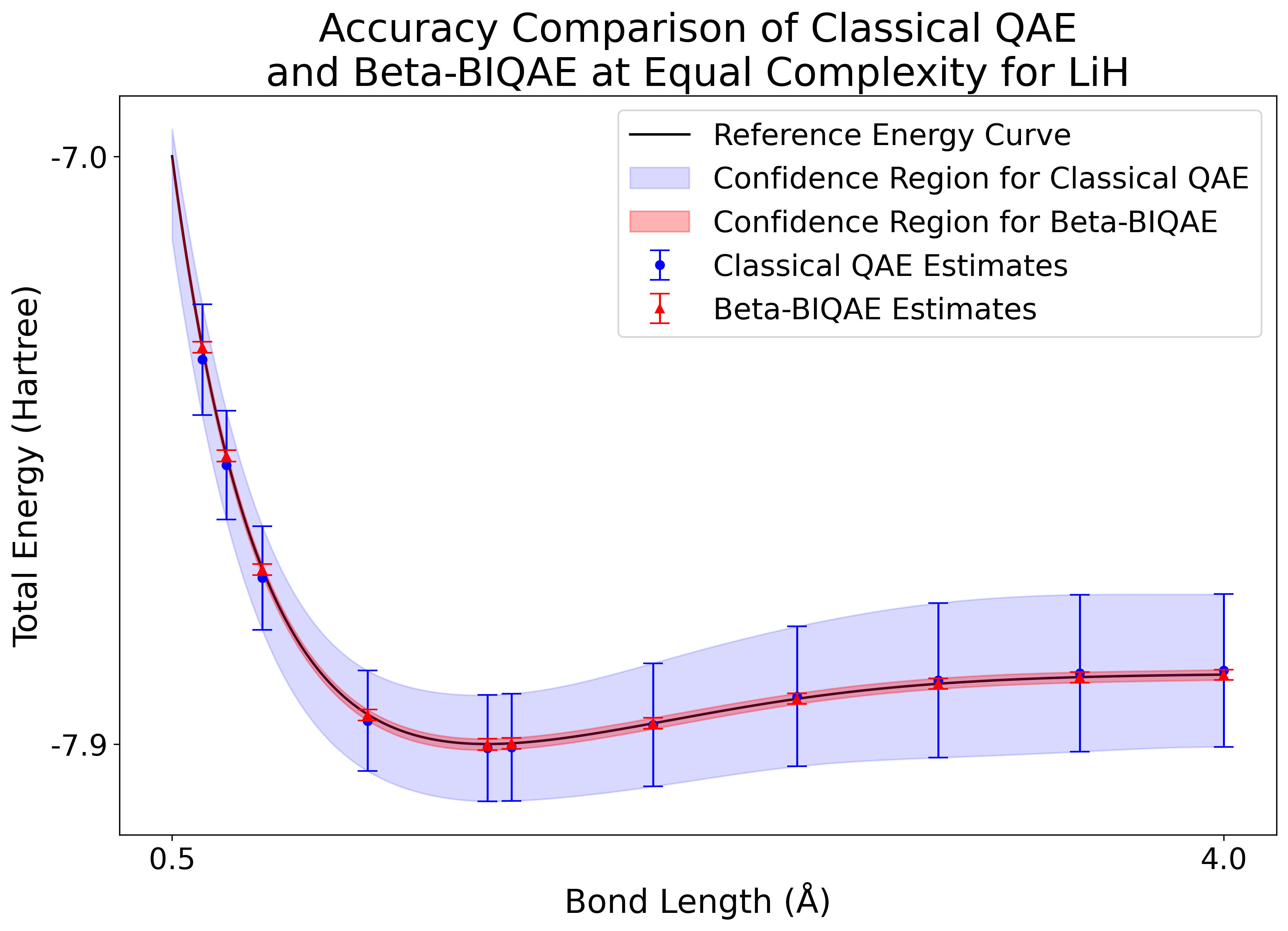} \\

        \includegraphics[width=0.5\textwidth]{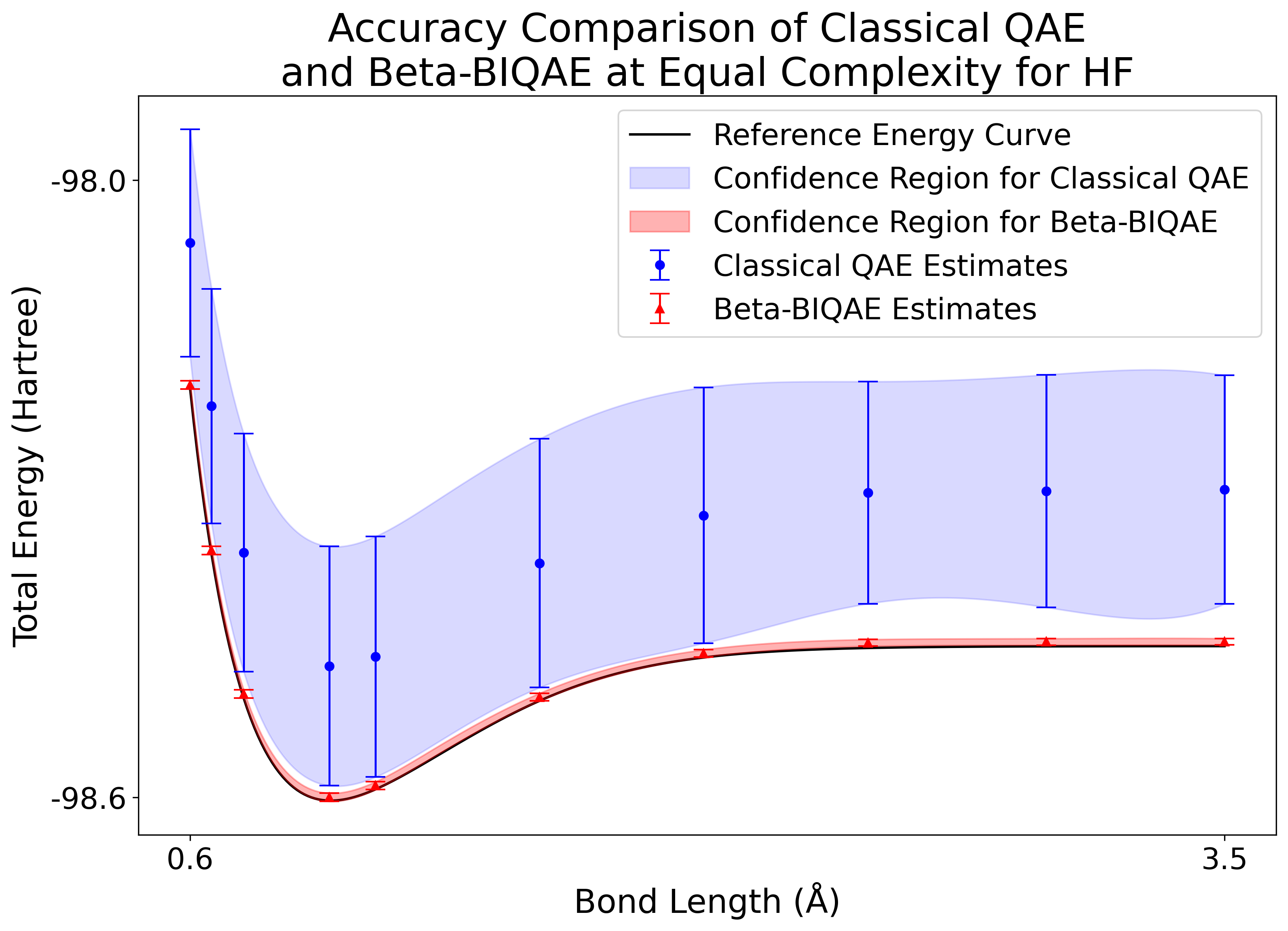} &  
        \includegraphics[width=0.5\textwidth]{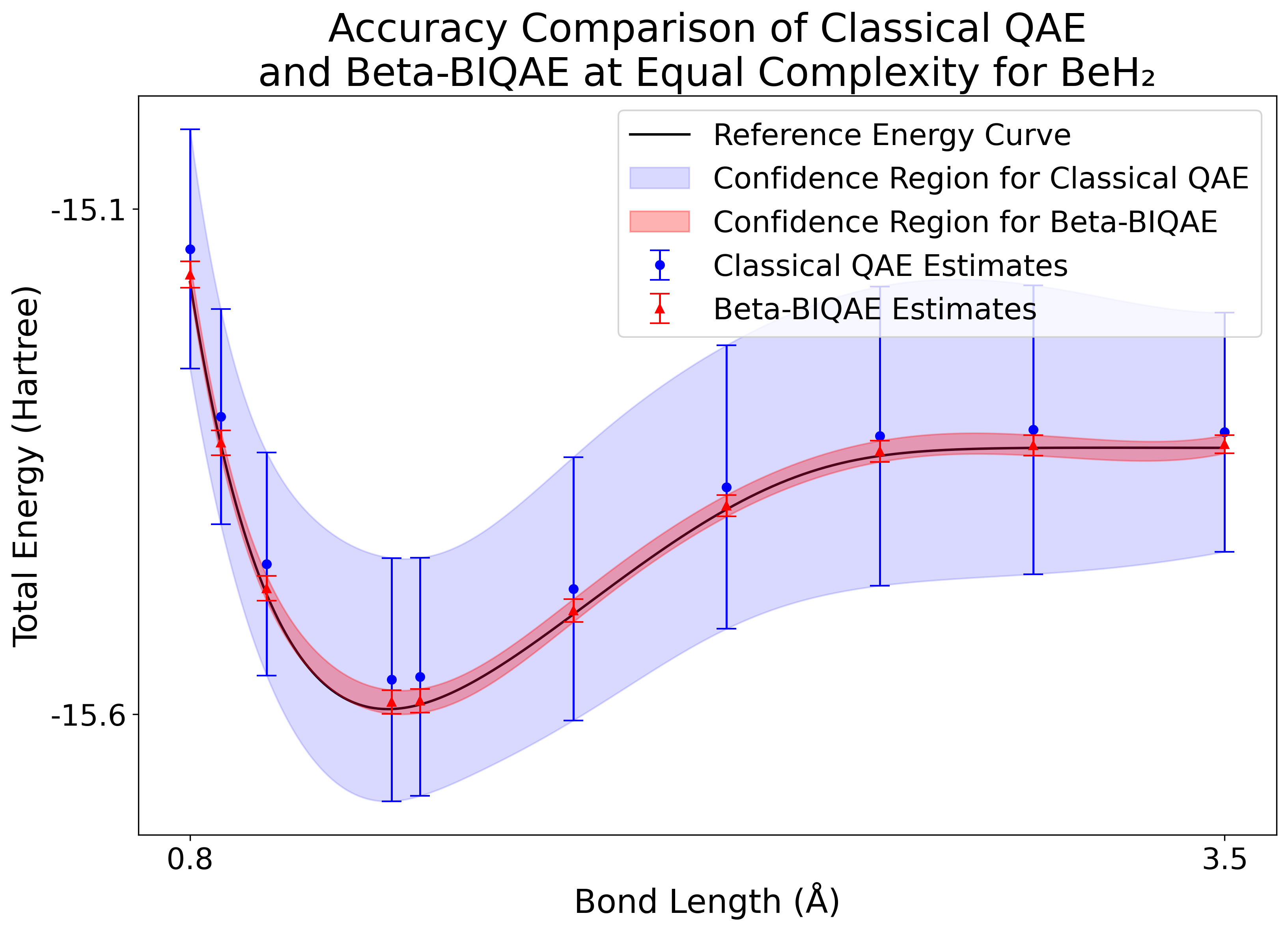}
    \end{tabular}

    \caption{Comparison of Beta-BIQAE (red) and Classical QAE (blue) accuracy for H\textsubscript{2}, LiH, HF, and BeH\textsubscript{2} ground-state energies at equal quantum sample complexity for Classical QAE and Beta-BIQAE. Error bars represent interval estimates, points indicate interval centers, and shaded areas depict the confidence region, with the energy curve from exact diagonalization of the Hamiltonian included for reference (black). Results between sampled bond lengths are interpolated with quartic splines. }
    \label{fig:compare_acc}
\end{figure*}

\begin{figure*}[t]
    \centering
    \begin{tabular}{cc}  
        \includegraphics[width=0.5\textwidth]{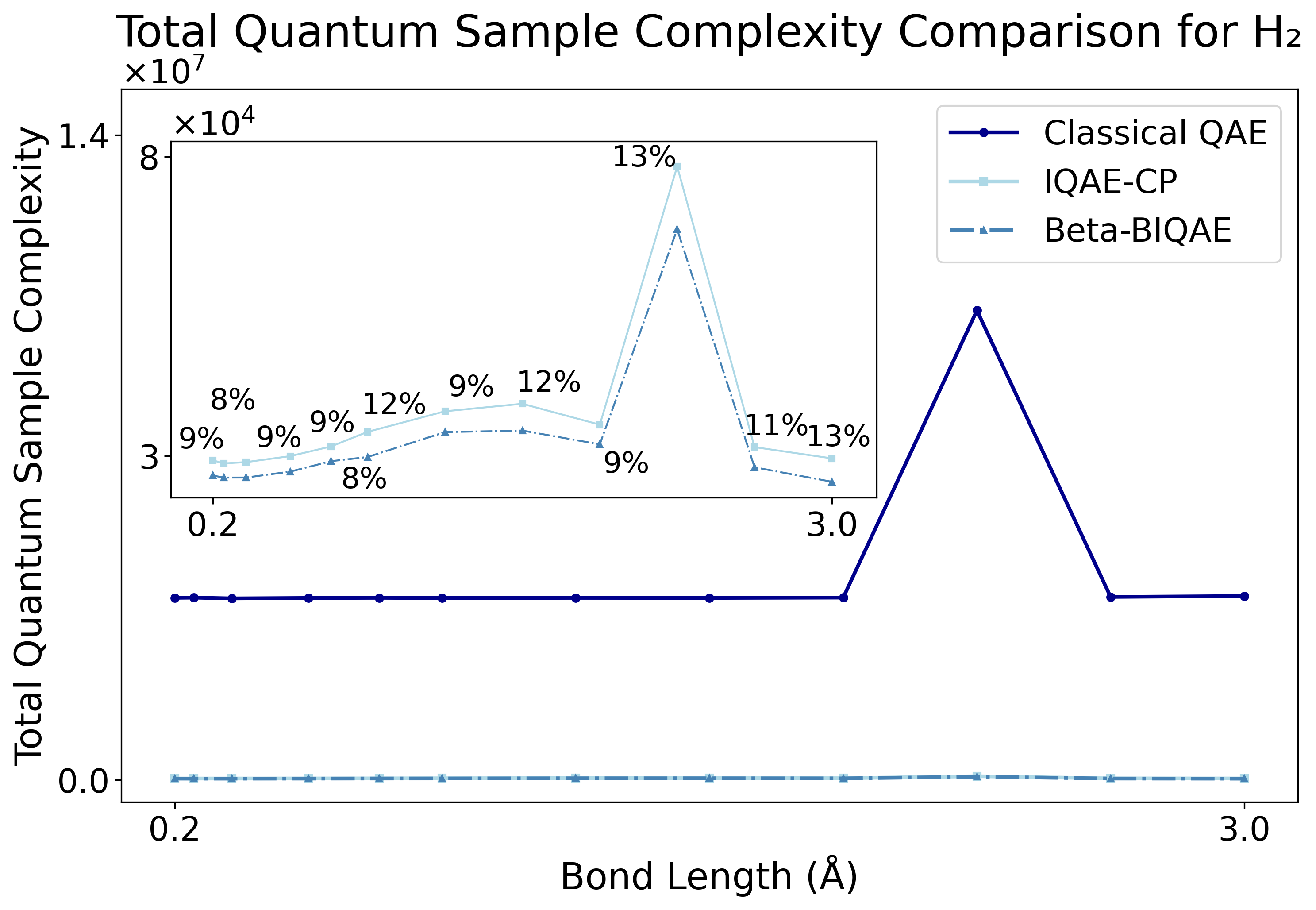} &  
        \includegraphics[width=0.5\textwidth]{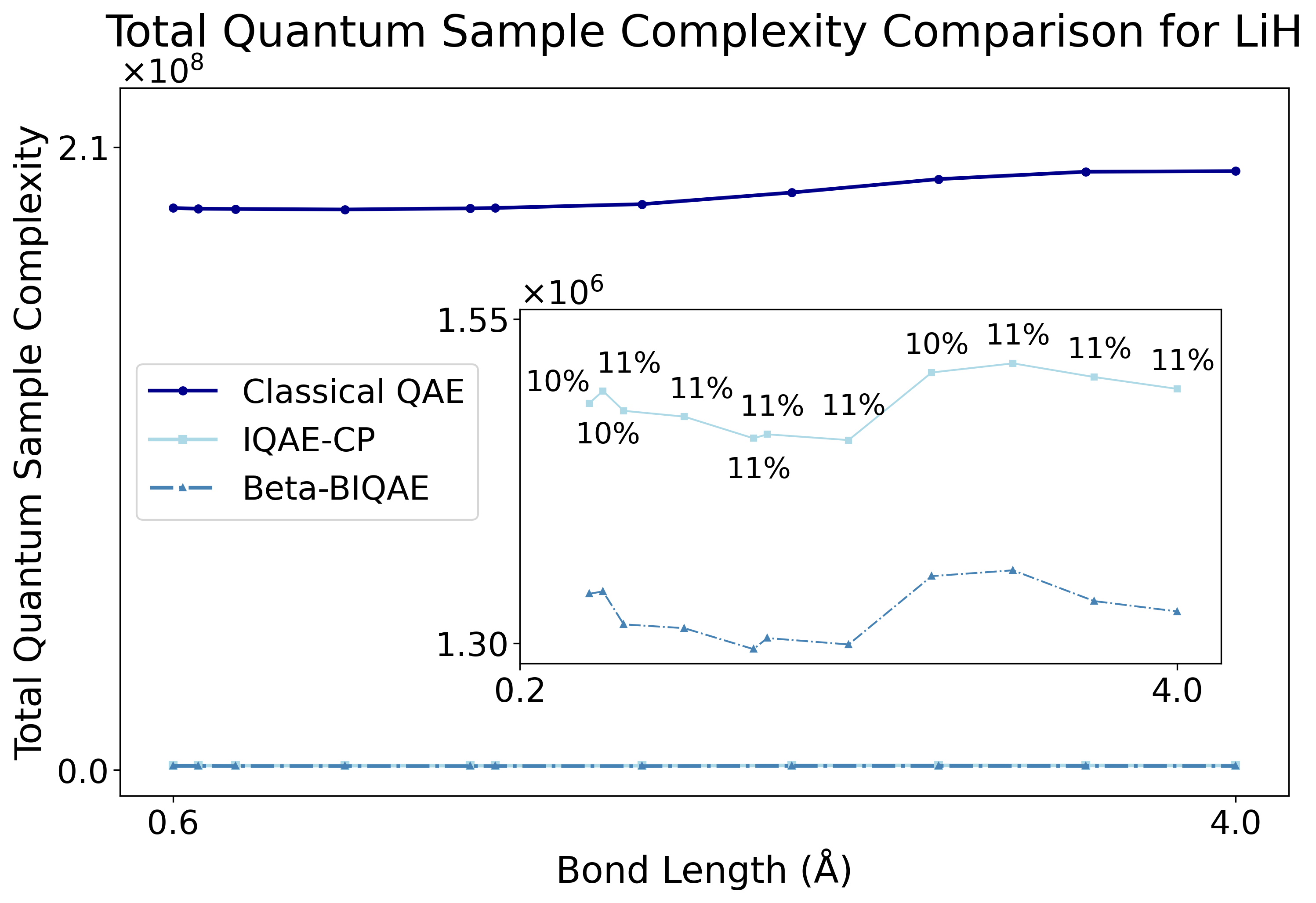} \\

        \includegraphics[width=0.5\textwidth]{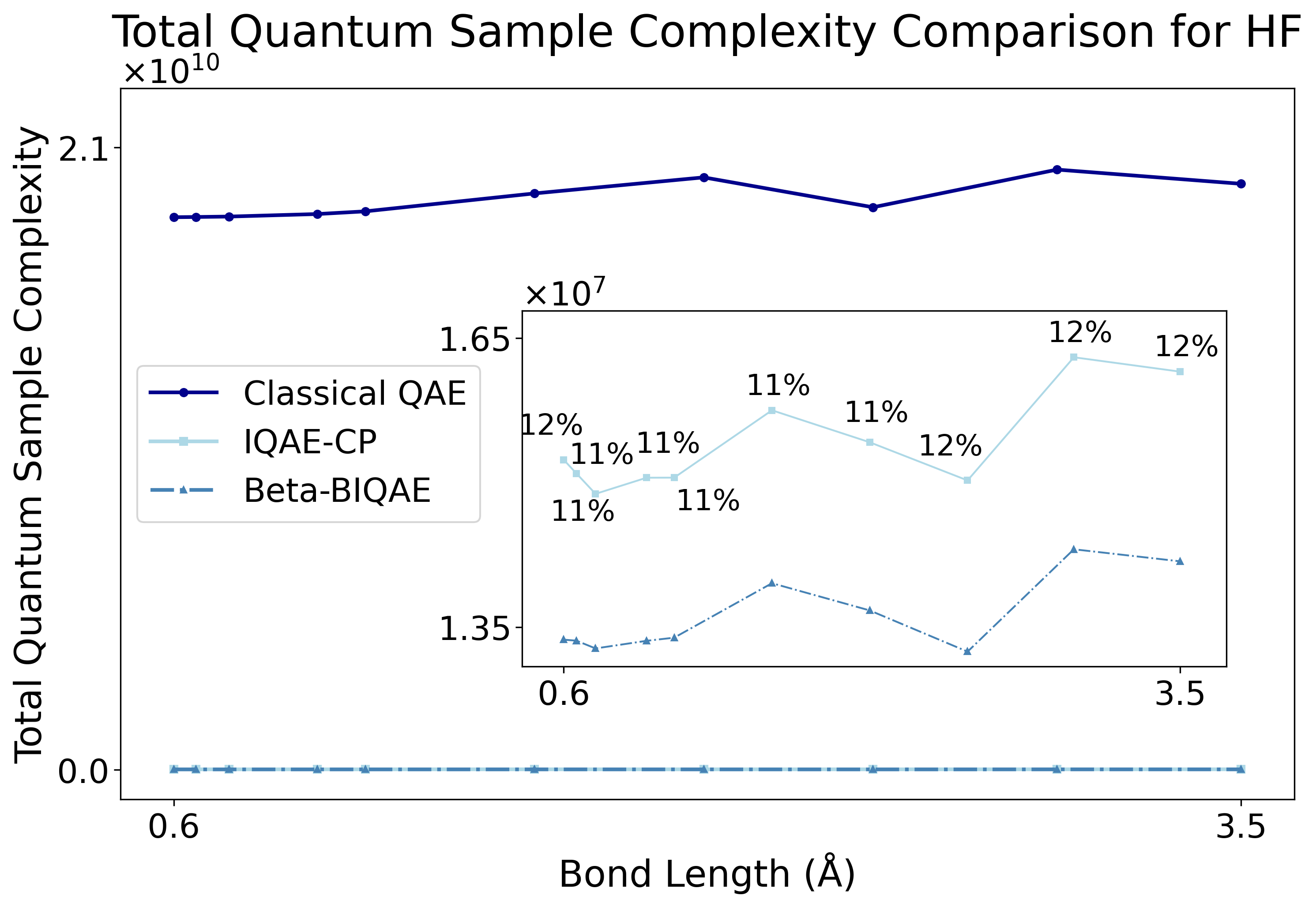} &  
        \includegraphics[width=0.5\textwidth]{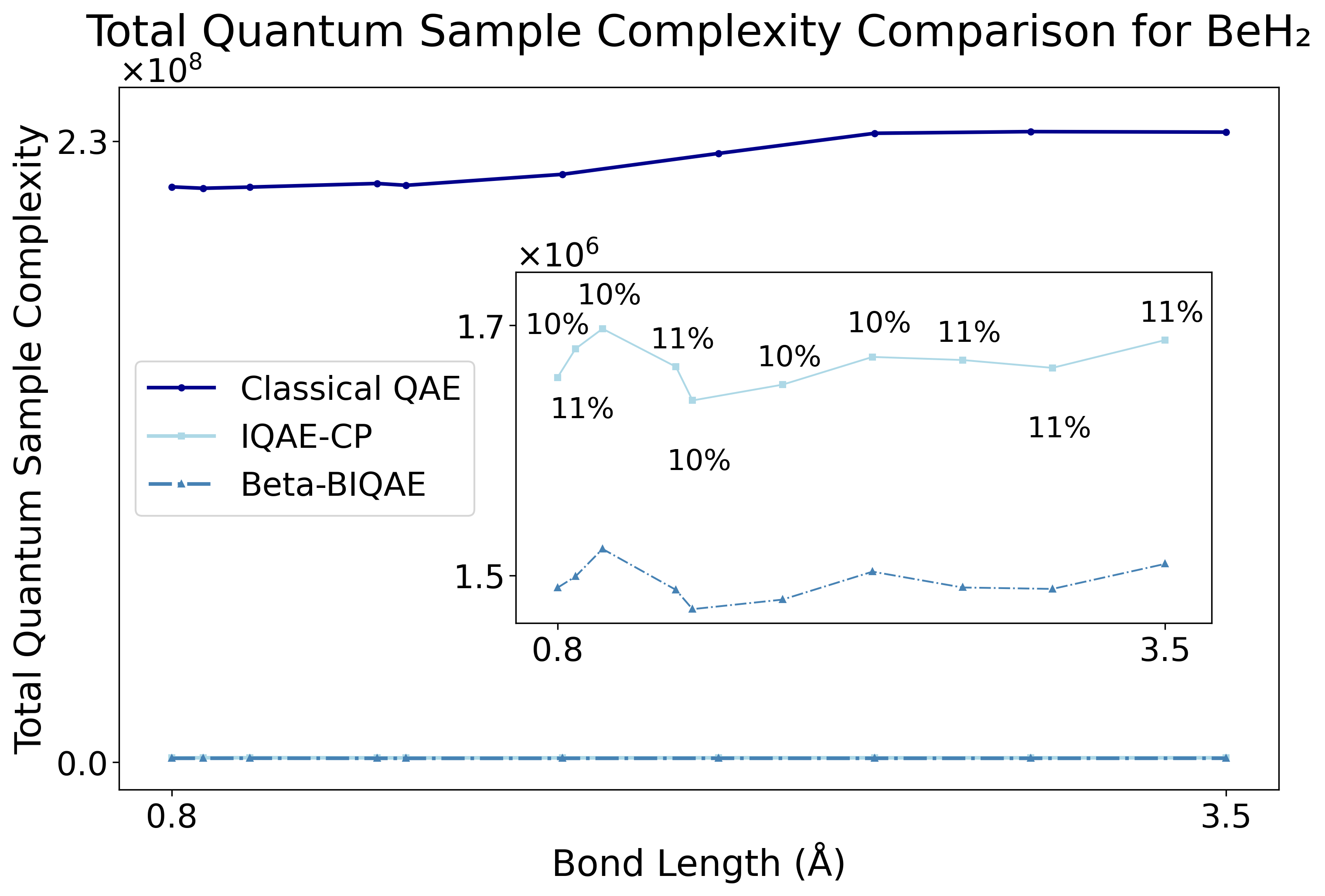}
        
    \end{tabular}

    \caption{Comparison of quantum sample complexities for estimation of molecular ground-state energies as a function of interatomic bond length for Classical QAE (dark blue circles), IQAE-CP (light blue squares), and Beta-BIQAE (medium blue triangles) with inset illustration of percentage improvement achieved by Beta-BIQAE over IQAE-CP.}
    \label{fig:multi_BL}
\end{figure*}

\begin{figure*}[t]
    \centering
    \begin{tabular}{cc}  
        \includegraphics[width=0.5\textwidth]{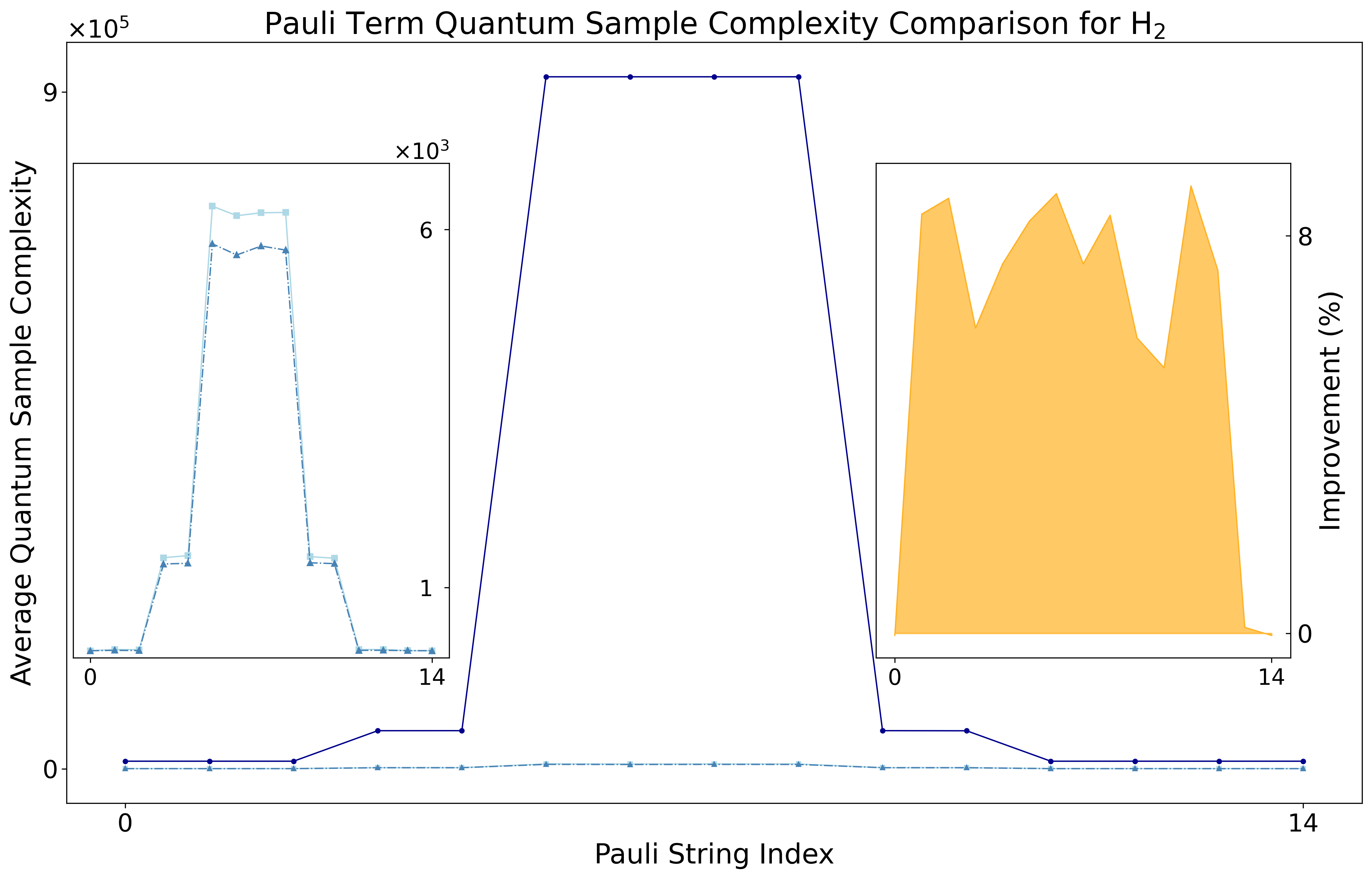} &  
        \includegraphics[width=0.5\textwidth]{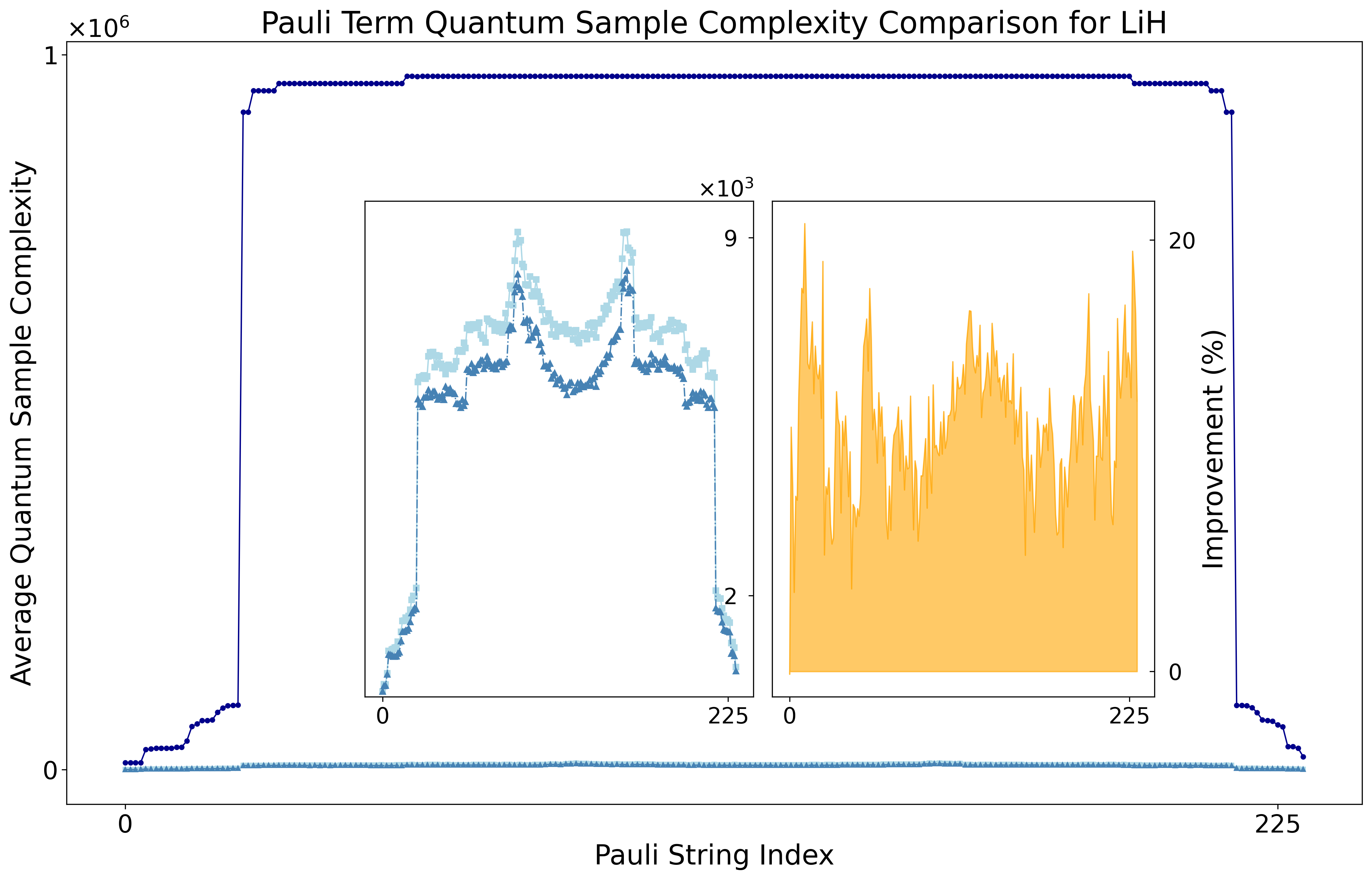} \\
        \includegraphics[width=0.5\textwidth]{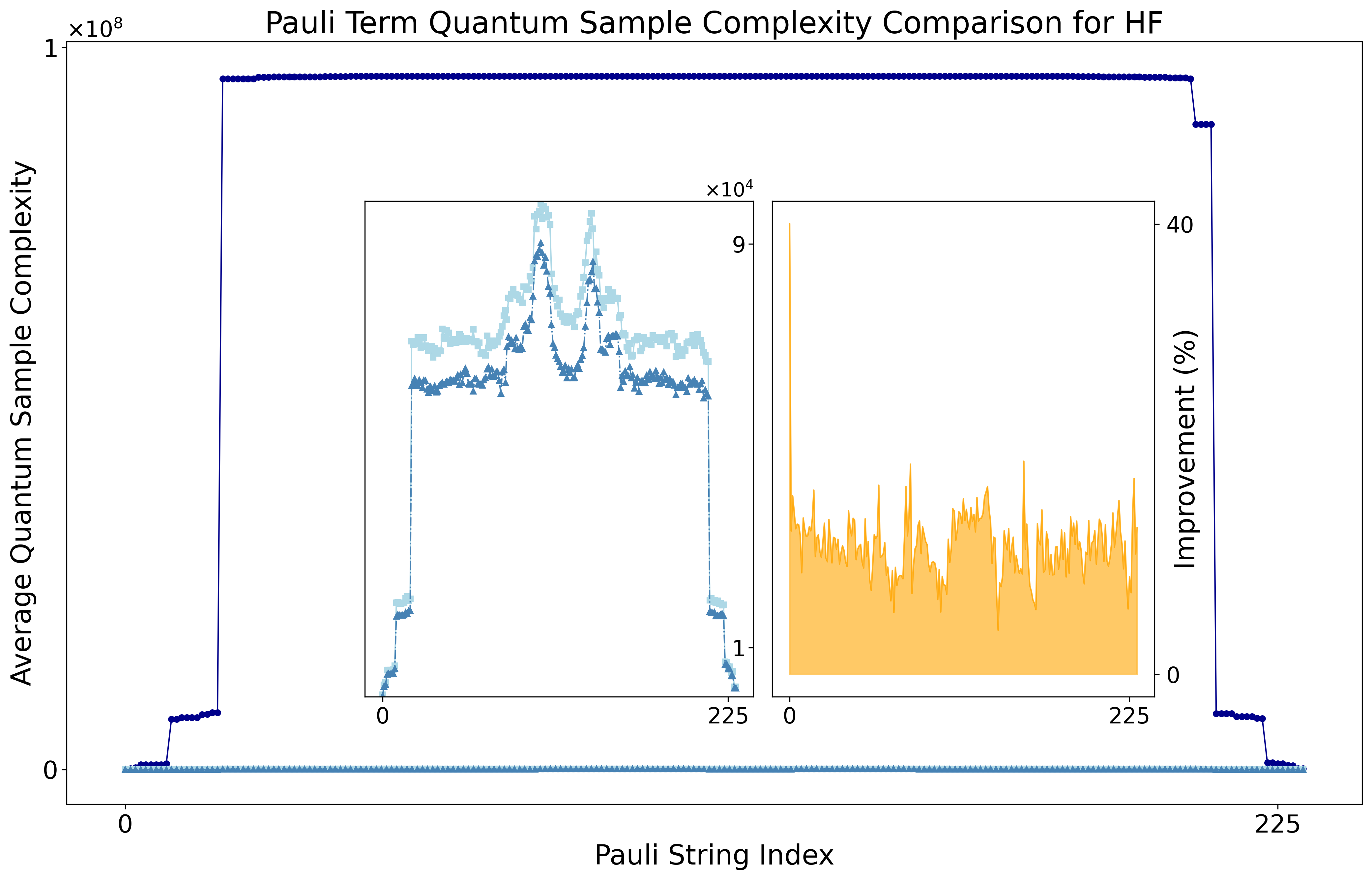} &  
        \includegraphics[width=0.5\textwidth]{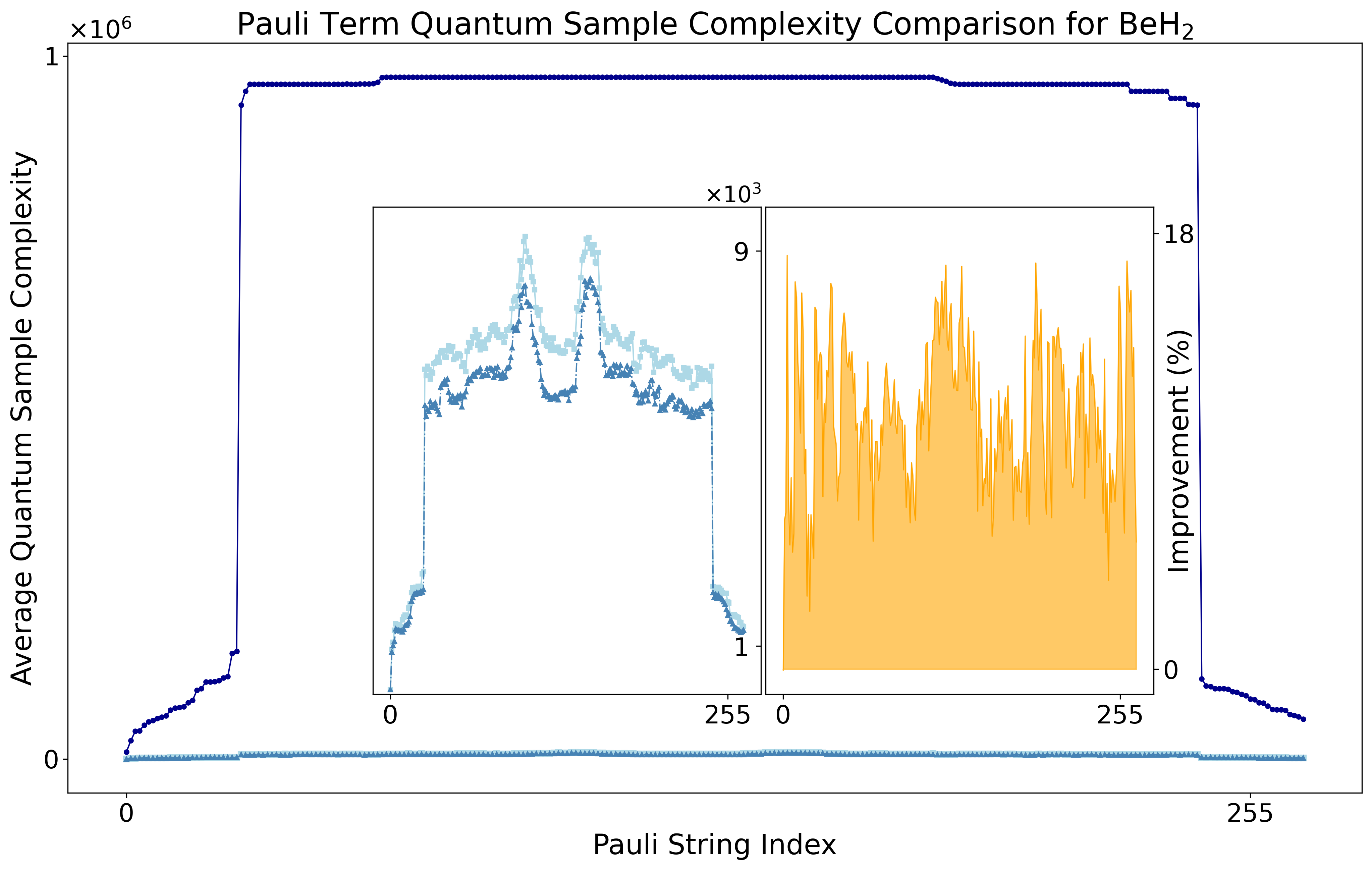} \\
        \multicolumn{2}{l}{\includegraphics[height=0.5cm]{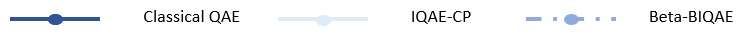}}
    \end{tabular}
    
    \caption{Quantum sample complexity comparison for H\textsubscript{2}, LiH, HF, and BeH\textsubscript{2} ground-state energies with Classical QAE (dark blue circles), IQAE-CP (light blue squares), and Beta-BIQAE (medium blue triangles). Insets indicate Beta-BIQAE outperforms IQAE-CP via magnification of average quantum sample complexity results (left) and percentage improvement of Beta-BIQAE over IQAE-CP (right). Pauli strings are sorted in descending order of magnitude of the corresponding expectation value.}
    \label{fig:sc_all}
\end{figure*}

\subsection{Molecular Ground-State Energy Estimation}
\label{sec:vqe_simul}

As shown in Fig.~\ref{fig:compare_acc}, Beta-BIQAE estimates molecular ground-state energies with significantly greater accuracy than Classical QAE for all molecules at all bond lengths given a fixed quantum sample complexity. 
The interval estimates are significantly smaller for Beta-BIQAE than Classical QAE by an average factor of $\times32.84$ for HF, $\times14.2$ by H$_2$, $\times12.16$ for BeH$_2$, and $\times12.13$ for LiH. Additionally, in the majority of cases considered, the interval centers for Beta-BIQAE lie closer to the exact diagonalization value than Classical QAE, on average by a difference of 0.137 Hartree, 0.020 Hartree, 0.004 Hartree, and 0.002 Hartree for HF, H$_2$, LiH, and BeH$_2$, respectively. Note all of these improvements exceed chemical accuracy.  

Examination of the quantum sample complexity required to yield a given target accuracy likewise indicates Beta-BIQAE significantly reduces measurement cost for energy estimation relative to Classical QAE and IQAE-CP: Fig.~\ref{fig:multi_BL} indicates Beta-BIQAE requires a sample complexity 
two to three orders of magnitude lower than Classical QAE to reach a target accuracy of $\varepsilon=10^{-3}$ (three in the $\varepsilon=10^{-4}$ case of HF) across all bond lengths considered. 
Furthermore, Beta-BIQAE consistently requires a sample complexity 8\% to 13\% lower than IQAE-CP across all molecular tests. 

The quantum sample complexities to compute expectation values of individual Pauli terms to the specified target accuracies similarly demonstrate the lower measurement cost entailed by BIQAE relative to Classical QAE and IQAE-CP (see Fig.~\ref{fig:sc_all}). 
The average quantum sample complexity of Beta-BIQAE is lower than that of Classical QAE for all Pauli terms by up to three orders of magnitude for $\varepsilon=10^{-3}$ and up to four orders of magnitude for $\varepsilon=10^{-4}$. Likewise, the average quantum sample complexity of Beta-BIQAE is lower than that of IQAE-CP by up to a double-digit percentage for all Pauli terms for all molecules considered. 
In particular, this advantage over IQAE-CP averages a percent improvement of 10.16\% for $\text{H}_2$, 10.5\% for $\text{BeH}_2$, 10.72\% for $\text{LiH}$, and 11.4\% for $\text{HF}$, in agreement with the expected scaling result for arbitrary quantum amplitudes of $\sim10-16\%$. Moreover, a trend is clearly visible in which the sample complexity is lowest for extreme Pauli string indices with a plateau at intermediate Pauli string index values. Since the Pauli string indices are ordered in descending order of expectation value magnitude, this plateau is consistent with the prediction that higher quantum sample complexities are needed to estimate intermediate quantum amplitudes per the predicted quantum sample complexity dependence on $a$ of Beta-BIQAE, Classical QAE, and IQAE-CP $N_\text{oracle}\propto\sqrt{a(1-a)}$ in Eqs.~\eqref{eq:complexity_ClassicalQAE} and \eqref{eq:sharedscaling}. Note Beta-BIQAE's measurement cost reduction only requires a number of Grover operators comparable to that of IQAE-CP, and thus is not associated with a significant increase in circuit depth, as demonstrated in the extended results presented in Appendix~\ref{ext_res:vqe_circ_depth}. 

\section{Discussion}
\label{sec:Discussion}



The success of BIQAE for quantum amplitude estimation and molecular ground-state energy approximation shows conclusively the value of Bayesian inference as a way to reduce the cost of QAE and quantum computing algorithms generally. Since QAE is widely used as a subroutine in quantum computing tasks, it is therefore natural to explore BIQAE's implementation to accelerate a variety of calculations currently sought after as targets for quantum utility, including in tandem with supplementary cost reduction techniques that would further accelerate QAE. For example, the successful identification of molecular ground states with BIQAE demonstrated here invites the use of BIQAE to facilitate identification of molecular excited states in conjunction with quantum tomography techniques~\cite{PhysRevA.110.052407, cai2016optimal, wang2013asymptotic}. Consequently, the methods developed in this work are expected to have the potential to impact a wide range of quantum algorithms, thereby underscoring the broader significance of the results. 

Given the measurement cost reduction observed here for implementation of Bayesian inference in IQAE at the level of interval estimates, the next step for BIQAE would be to inject Bayesian inference at the level of hyperparameters or $(K,N)$-scheduling decisions. Bayesian optimization \cite{thompson1933likelihood,kushner1964new,movckus1974bayesian,shahriari2015taking} could be employed to optimize BIQAE's hyperparameters (such as the sample size $N_t$ and incremental shot batch size $N_\text{incre}$), which is expected to improve the constant prefactor associated with  the quantum sample complexity at fixed target accuracy. Likewise, the schedule employed here to facilitate direct comparison between BIQAE and IQAE literature results could be substituted with a Bayesian-informed schedule (such as that of BAE \cite{2024bayesianquantumamplitudeestimation}) to make the most of limited measurement information.

In order to realize BIQAE's gains, a key step would be to modify the algorithm for implementation on current and future quantum computers. The far-term method presented here could be readily modified for today's near-term intermediate scale quantum computers. Importantly, practical implementation of BIQAE with today's quantum computers will require consideration of current realistic noise conditions. Under the noise model considered in BAE \cite{2024bayesianquantumamplitudeestimation}, measurement outcome probabilities decay exponentially toward $p_k=\tfrac{1}{2}$ as the number of Grover operators increases $k\rightarrow\infty$, such that there exists a fundamental tradeoff that accompanies the addition of Grover operators between stochastic noise reduction and biased noise amplification in the absence of error correction. This tradeoff presently prevents IQAE from fully achieving its predicted efficiency gains, and would therefore be likely to similarly impact a naive implementation of BIQAE in the high Grover operator number limit on near-term quantum computers. Noise-aware schedules and restrictions on the number of Grover operators that have increased the robustness of prior QAE variants \cite{PRXQuantum.2.010346,2024bayesianquantumamplitudeestimation} provide a possible route to circumvent such obstacles in BIQAE, a direction that we leave as an open question for future research. More generally, strategies may be imported from QAE methods recently successfully implemented on both superconducting and trapped-ion quantum computers \cite{rao2020quantum1,vazquez2021efficient,tanaka2021amplitude,tanaka2022noisy,certo2022benchmarking,giurgica2022low1,dalal2023noise,herbert2024noise,kunitsa2024experimental}. These strategies may straightforwardly be used in conjunction with the most recent NISQ accommodations of RAE \cite{PRXQuantum.2.010346} and BAE \cite{2024bayesianquantumamplitudeestimation}. 
Ultimately, the power of BIQAE could be demonstrated without modification on far-term quantum computers if sufficiently reliable logical qubits and low-error gates are developed.  Based on the results, it would be fruitful to identify new avenues for Bayesian inference to accelerate the frequentist tasks ubiquitous in quantum approaches to mathematical and scientific applications \cite{montanaro2015quantum,yu2020practical,rao2020quantum,kassal2008polynomial,gunther2024more,baker2020density,johnson2022reducingcostenergyestimation,kunitsa2024experimental,agliardi2022quantum,de2023quantum,miyamoto2024quantum,de2025quantum,martinez2024loop,lee2025quantum,johri2017entanglement,rall2020quantum,dong2022ground,piroli2024approximating,agrawal2024quantifying,gaitan2020finding,bharadwaj2023hybrid,penuel2024feasibility,gaitan2024circuit,stamatopoulos2020option,bouland2020prospects,stamatopoulos2022towards,alcazar2022quantum,herman2023quantum,wang2024option,woerner2019quantum,gomez2022survey,rebentrost2018quantum,miyamoto2022bermudan,egger2020credit,orus2019quantum,braun2021quantum,wiebe2016quantum,wiedemann2023quantumpolicyiteration}.

\section{Acknowledgments}

M.~B.~S.~acknowledges support for this research provided by the Office of the
Vice Chancellor for Research and Graduate Education (OVCRGE) at the University of Wisconsin-Madison with funding from the Wisconsin Alumni Research Foundation (WARF). Y.~W.~acknowledges partial support for this research by OVCRGE from WARF and NSF Grant DMS-2514240.  AI was used to refine early drafts of the manuscript.

\bibliographystyle{quantum}
\bibliography{references}   

\onecolumn\newpage
\appendix

\section{Algorithmic and Computational Details}
\label{appendix:bg_rw}

\subsection{Quantum Sample Complexity of Classical QAE}
\label{appendix:CAI}

As derived in Appendix~\ref{appendix:proof_thm123}, the quantum sample complexity of Classical QAE is formally stated in a statistical framework as follows:
\begin{theorem}
    \label{appendix_thm:ClassicalQAE}
    Consider \(n\) instances of the quantum circuit associated with oracle \(\mathcal{A}\). Let \(X_i\) denote the outcome of the \(i^\text{th}\) instance, where the state \(|0\rangle\) is mapped to the outcome \(0\) and \(|1\rangle\) is mapped to \(1\). Then, the outcomes \(X_1, \ldots, X_n\) are independent and identically distributed (i.i.d.) Bernoulli random variables with success probability \(p = a\). 

    The sample mean \(\bar{X} = \frac{1}{n}\sum_{i=1}^n X_i\) serves as the Uniformly Minimum Variance Unbiased Estimator (UMVUE) for \(a\). The mean squared error (MSE) of \(\bar{X}\) with respect to \(a\) is 
    \[
        \text{MSE}_{\bar{X}}(a) = \frac{1}{n} a (1 - a).
    \]
    To achieve a target accuracy \(\varepsilon\), the required number of shots is
    \[
        N_{\text{shots}} = \frac{1}{\varepsilon^2} a (1 - a).
    \]
    Since each shot requires one query to the oracle \(\mathcal{A}\), the total number of accesses to the oracle is then
    \[
        N_{\text{oracle}} = N_{\text{shots}} = \frac{1}{\varepsilon^2} a (1 - a).
    \]
\end{theorem}

\subsection{Mechanistic Foundations of Amplitude Estimation}
\label{appendix:CAICoin}

To aid readers new to the field, we provide the following brief, simplified example of how amplitude amplification improves measurement efficiency:

\begin{example}
    \label{appendix_example: coin}
    Suppose one seeks to estimate the head probability of a coin while in possession of (i) the ability to estimate the head probability of a biased coin with accuracy $0.02$ and (ii) a technique that amplifies the head probability by a factor of $\times4$. If the unknown head probability is $0.15$ and one tosses the coin repeatedly, one will obtain an interval estimate of $[0.13, 0.17]$ without amplification or $[0.58, 0.62]$ with amplification.

    At first glance, the interval lengths appear to be equal for both approaches. However, the interval in the amplified case must be rescaled by the inverse of the original factor (namely, $\times 1/4$) to reorient to the original head probability. 
    The interval after rescaling is $[0.145, 0.155]$, which is {\em four times} narrower than the interval obtained without amplification. Thus, using the amplified approach improves accuracy in the sample case by a factor of four.
\end{example}

\subsection{Identifiability Challenge of Amplitude Estimation}
\label{appendix:CAIIdentifiability}

Fig.~\ref{appendix_fig:mul_sol} provides a visual demonstration of the need for additional information to identify the unique amplitude to which a measured probability corresponds in amplitude-amplification-based QAE methods. Given the relationship between the amplified angle and the amplified target probability Eq.~\eqref{eq:pk}, there exist multiple  angles that correspond to each observed probability Eq.~\eqref{eq:mul_sol}. However, these individual solutions may be distinguished uniquely by specification of their base or quadrant index Eq.~\eqref{eq:quadrant index}.

\begin{figure}[ht]
    \centering
    \includegraphics[width=0.8\linewidth]{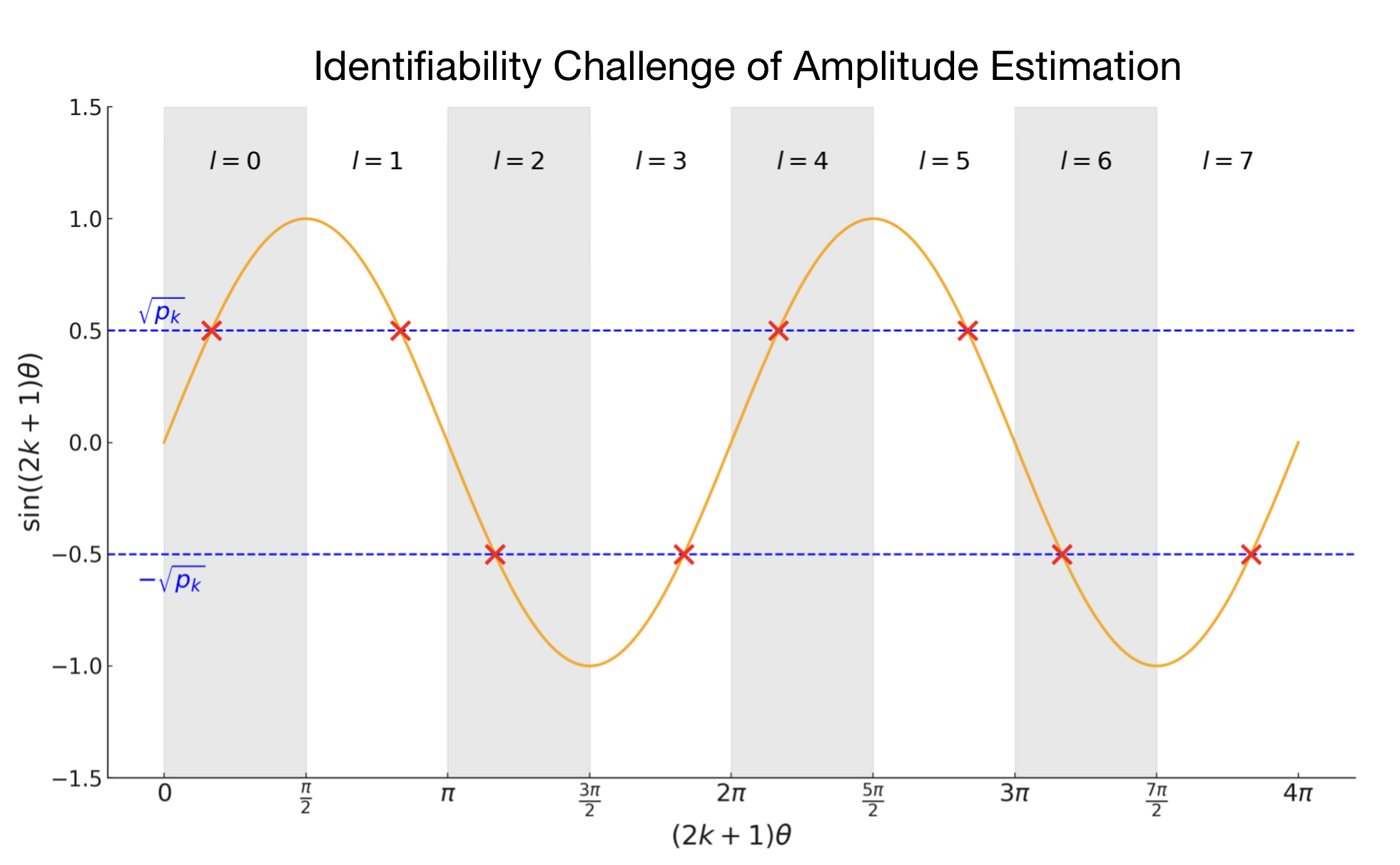}
    \caption{Illustration of the identifiability challenge in amplitude estimation, in which the relationship between the amplified angle and amplified target probability (yellow solid line) entails multiple valid solutions (red points) for each measured amplitude ($p_k=0.5$ shown, blue dashed line), which are distinguished uniquely by their base or quadrant index (from left to right, \(l=0, 1, 2,3,\ldots\), alternating shaded gray and white vertical areas).}
    \label{appendix_fig:mul_sol}
\end{figure}

\subsection{Resolution of the Identifiability Challenge with Arbitrary Quadrant Index}
\label{appendix:known_quad_idx}

To demonstrate that the quadrant index $l(k)$ need not be zero to uniquely estimate the quantum amplitude $a$ (a requirement of Theorem~\ref{thm:AAE}), we relax the small-angle assumption of Assumption~\ref{assump:smallprob} to mere knowledge of the value of $l(k)$, as follows:

\begin{assumption}
    \label{appendix_assump:knownbase}
    For a given $k$, the corresponding quadrant index $l(k)$ is known.
\end{assumption}

The proof provided in Appendix~\ref{appendix:proof_thm123} then indicates that Theorem~\ref{thm:AAE} extends to

\begin{theorem}
    \label{appendix_thm:knownbase}
    Suppose Assumption~\ref{appendix_assump:knownbase} holds with a certain $k$. Then, 
    \begin{align*}
        \hat{a}\coloneqq
        \begin{cases} 
            \sin^2{\left[\frac{1}{2k+1}\left(\arcsin{\left(\sqrt{\bar{X}_k}\right)} + l(k)\cdot\frac{\pi}{2}\right)\right]} & \text{if $l(k)$ is even}, \\
            \sin^2{\left[\frac{1}{2k+1}\left(\arccos{\left(\sqrt{\bar{X}_k}\right)} + l(k)\cdot\frac{\pi}{2}\right)\right]} & \text{if $l(k)$ is odd}
        \end{cases}
    \end{align*}
    is the Maximum Likelihood Estimator (MLE) for $a$. The corresponding asymptotic MSE is
    \begin{align*}
        \text{AMSE}_{\hat{a}}(a)=\frac{1}{(2k+1)^2}\frac{1}{N_k}a(1-a),
    \end{align*}
    and the required number of accesses to the oracle $\mathcal{A}$ to achieve a target accuracy $\varepsilon$ is 
    \begin{align*}
        N_\text{oracle} = \frac{1}{2k+1}\frac{1}{\varepsilon^2}a(1-a).
    \end{align*}
\end{theorem}

Note this theorem not only establishes that amplified estimation is feasible for arbitrary amplitudes given prior knowledge of the quadrant index, but also that the estimator increases in efficiency as $k$ grows---insights implicit in the design of IQAE \citep{grinko2021iterative} that are broadly applicable to the wider class of amplified-estimation-based QAE methods.

\subsection{Pictorial Representation and Accumulated Complexity Bound of Base-3/Base-5 Hybrid-Scheduled IQAE}
\label{appendix:35_IQAE}

To overcome the problem that base-3 hybrid-scheduled IQAE struggles where the amplified angle $K_t\theta$ falls too close to the reference interval boundary (due to the increasing difficulty of constructing an interval estimate fully contained in a single reference interval as the amplified angle nears said boundaries), we introduce a second set of reference intervals in base 5, as depicted in Fig.~\ref{appendix_fig:boundaries}. 
Introduction of the second set of reference intervals ensures that the interval estimate must be fully contained in a single reference interval either in base 3 or 5 whenever its length is less than the smallest separation between reference interval boundaries (specifically, 6 degrees).
Replacement of the Assumption~\ref{assump:schedule3} that underpins base-3 hybrid-scheduled IQAE with the lower bound on the radius of the confidence interval radius Eq.~\eqref{eq:epsilont35} then results in the following accumulated complexity bound, as derived and proven in Appendix~\ref{appendix:proof_schedule35}:

\begin{theorem}
    \label{appendix_thm:schedule35}
    Using the base-3/base-5 \(K\)-schedule, the accumulated sample complexity to achieve a accuracy of \(\varepsilon\) with an overall confidence level of \(\alpha\) \((\leq 0.1)\) is upper bounded by  Eq.~\eqref{bound:schedule35}.
\end{theorem}
Note this strategy requires no additional ancilla, in contrast to the strategy of ref.~\citep{zhao2022adaptive}.

\begin{figure}[ht] 
    \centering
    \includegraphics[width=0.45\linewidth]{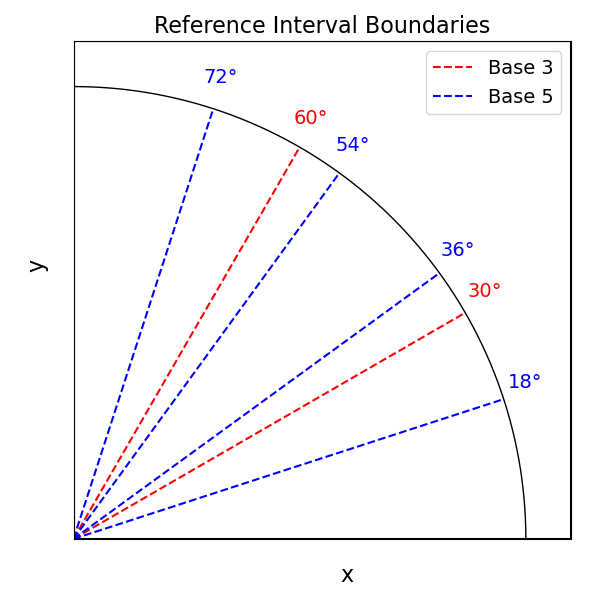}
    \caption{Reference interval boundaries in base 3 (red dashed lines) and base 5 (blue dashed lines) in base-3/base-5 hybrid-scheduled IQAE.}
    \label{appendix_fig:boundaries}
\end{figure}

\subsection{Pseudocode for Normal-BIQAE}
\label{appendix:biqae_normal}

Normal-BIQAE proceeds according to the workflow illustrated in Fig.~\ref{fig:flow_chart_normal} and modules provided in Algorithms \ref{appendix_alg:BayesianUpdate}--\ref{appendix_alg:PreparePrior}, in which normal conjugate priors ensure all prior and posterior distributions remain normal throughout the process and only the mean and variance evolve as the method progresses.

\begin{figure}[ht] 
    \centering
    \includegraphics[width=1\textwidth]{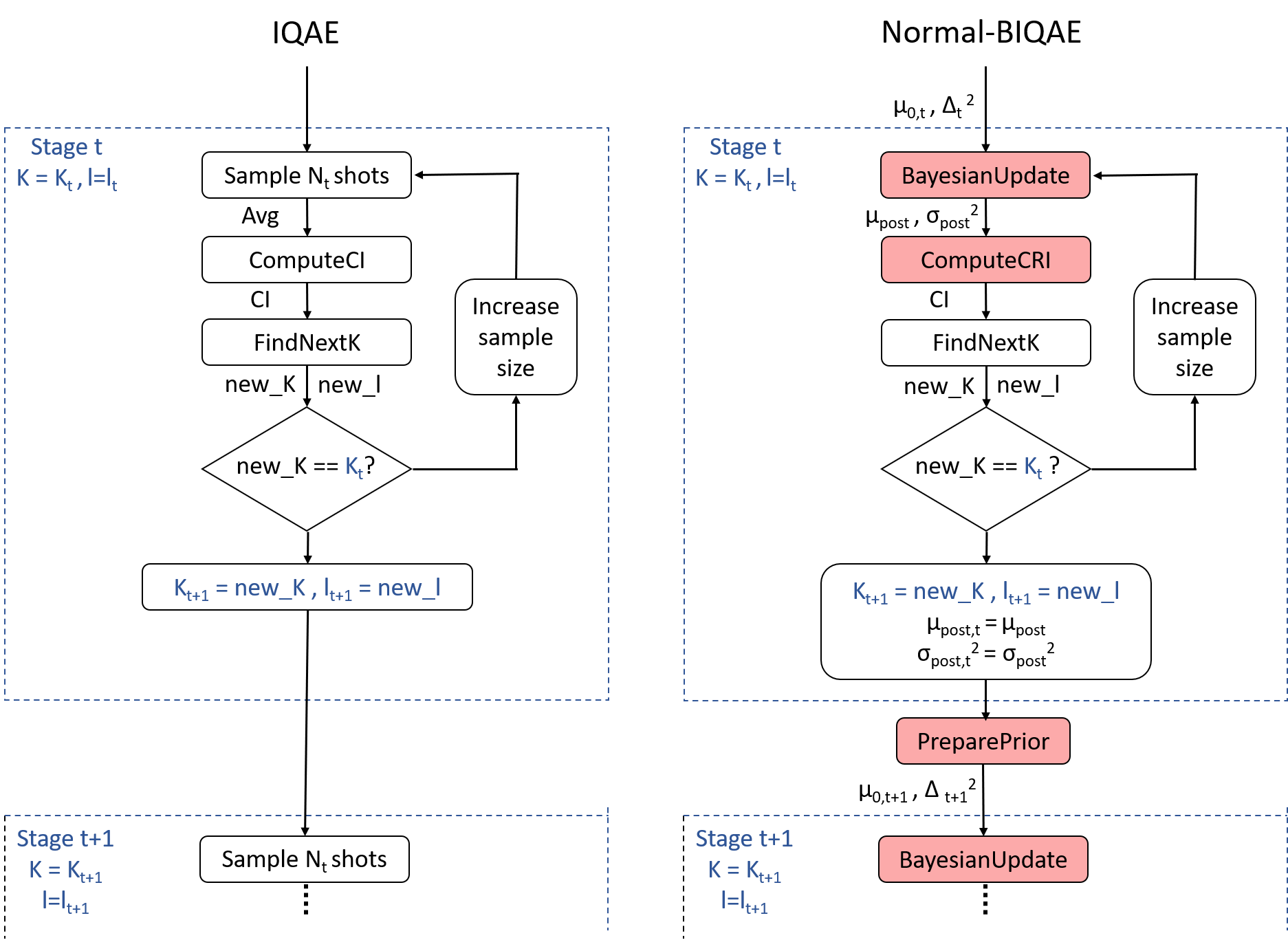}
    \caption{Detailed workflow of Normal-BIQAE (right), juxtaposed with IQAE for comparison (left). The three modules that distinguish BIQAE from IQAE are emphasized in red.}
    \label{fig:flow_chart_normal}
\end{figure}

\begin{algorithm}[H]
\caption{BayesianUpdate}
\label{appendix_alg:BayesianUpdate}
\begin{algorithmic}[1]
\Require Prior mean \(\mu_{0,t}\) and variance \(\Delta_t^2\) for \(p_{k_t}\), number of \(\mathcal{Q}\) operators \(k_t\), sample size \(N_t\)
\Ensure Posterior mean \(\mu_{\text{post},t}\) and variance \(\sigma_{\text{post},t}^2\) for \(p_{k_t}\)

\State Construct \(k_t\) \(\mathcal{Q}\) operators into the circuit \(\mathcal{A}\)

\State Sample shots until a total of \(N_t\) shots (including existing shots at stage \(t\)) is collected. Store the sample average as \(\bar{X}_{t}\)

\State Estimate the population variance:
\[
\hat{\sigma}_t^2 = \bar{X}_t(1 - \bar{X}_t)
\]

\State Compute the posterior mean and variance:
\[
\mu_{\text{post},t} = \frac{\frac{\mu_{0, t}}{\Delta_t^2} + \frac{\bar{X}_t}{\hat{\sigma}_t^2 / N_t}}{\frac{1}{\Delta_t^2} + \frac{1}{\hat{\sigma}_t^2 / N_t}},
\quad
\sigma_{\text{post},t}^2 = \left(\frac{1}{\Delta_t^2} + \frac{1}{\hat{\sigma}_t^2 / N_t}\right)^{-1}
\]

\State \textbf{Return} \((\mu_{\text{post}, t}, \sigma_{\text{post}, t}^2)\)

\end{algorithmic}
\end{algorithm}

\begin{algorithm}[H]
\caption{ComputeCRI}
\label{appendix_alg:ComputeCRI}
\begin{algorithmic}[1]
\Require Posterior \((\mu_{\text{post}, t}, \sigma_{\text{post}, t}^2)\) for \(p_{k_t}\), amplification factor \(K_t\), quadrant index \(l_t\), confidence level \(\alpha_t\)
\Ensure An interval estimate \((\theta_{l,t}, \theta_{u,t})\) for \(\theta\) at confidence level \(\alpha_t\)

\State Compute the interval estimate for \(p_{k_t}\):
\[
\left[p_{k_t}^l, p_{k_t}^u\right]=\mathrm{trunc}_{[0,1]}(\mu_{\text{post},t} \pm z_{\alpha_t/2} \cdot \sigma_{\text{post},t})
\]

\State Map the interval estimate for \(p_{k_t}\) to an interval estimate for \(\theta\) using \(f_t\):
\[
\left[\theta_t^l, \theta_t^u\right]=f_t([p_{k_t}^l, p_{k_t}^u])
\]
where 
\begin{align*}
    f_t(x) \coloneqq
    \begin{cases} 
        \frac{1}{K_t}\left(\arcsin{(\sqrt{x})} + l_t \cdot \frac{\pi}{2}\right) & \text{if \(l_t\) is even}, \\
        \frac{1}{K_t}\left(\arccos{(\sqrt{x})} + l_t \cdot \frac{\pi}{2}\right) & \text{if \(l_t\) is odd}
    \end{cases}
\end{align*}

\State \textbf{Return} \(\left[\theta_t^l, \theta_t^u\right]\)

\end{algorithmic}
\end{algorithm}

\begin{algorithm}[H]
\caption{PreparePrior}
\label{appendix_alg:PreparePrior}
\begin{algorithmic}[1]
\Require Posterior mean and variance \((\mu_{\text{post},t}, \sigma^2_{\text{post},t})\) for \(p_{k_t}\) obtained at the end of stage \(t\), current and next amplification factors \(K_t, K_{t+1}\), current quadrant index \(l_t\)
\Ensure Prior mean \(\mu_{0, t+1}\) and variance \(\Delta_{t+1}^2\) for \(p_{k_{t+1}}\)

\State Compute the prior mean for \(p_{k_{t+1}}\):
\[
\mu_{0, t+1} = \sin^2\left(K_{t+1} \cdot f_t(\mu_{\text{post},t})\right)
\]

\State Compute the prior variance for \(p_{k_{t+1}}\):
\[
\Delta_{t+1}^2 = \left(\frac{K_{t+1}}{K_t}\right)^2 \frac{\mu_{0, t+1}(1 - \mu_{0, t+1})}{\mu_{\text{post},t}(1 - \mu_{\text{post},t})} \cdot \sigma^2_{\text{post},t}
\]

\State \textbf{Return} \((\mu_{0, t+1}, \Delta_{t+1}^2)\)

\end{algorithmic}
\end{algorithm}

Theorem~\ref{appendix_thm:approx_prior}, which is proven in Appendix~\ref{appendix:proof_approx_prior}, provides the theoretical guarantee that underlies Algorithm \ref{appendix_alg:PreparePrior}:

\begin{theorem}
    \label{appendix_thm:approx_prior}
    Let $\mu_{\text{post},t}\in (0,1), \sigma^2_{\text{post},t}$ denote the posterior mean and variance of $p_{k_{t}}$ obtained at the end of stage $t$. Then, the prior distribution of $p_{k_{t+1}}$ can be approximated by a normal distribution with mean
    \begin{align*}
        \mu_{0, t+1} = \sin^2\left(K_{t+1}\cdot f_t(\mu_{\text{post},t})\right)
    \end{align*}
    and variance
    \begin{align*}
        \Delta_{t+1}^2 =\left(\frac{K_{t+1}}{K_t}\right)^2\frac{\mu_{0, t+1}(1-\mu_{0, t+1})}{\mu_{\text{post},t}(1-\mu_{\text{post},t})}\cdot \sigma^2_{\text{post},t}.
    \end{align*}
\end{theorem}

\subsection{Quantum Sample Complexity Analysis for Normal-BIQAE}
\label{appendix:complexity_normal}

To determine the constant-factor reduction of the complexity bound of Normal-BIQAE relative to Normal-IQAE, we derive the sample complexity of a given stage (in place of the aforementioned accumulated sample complexity, see 
proof in Appendix~\ref{appendix:proof_sym2}):

\begin{theorem}
    \label{appendix_thm:sym2}
    At stage $t$, the quantum sample complexity required to obtain an interval estimate for $\theta$ with a radius of $\varepsilon_t$ is upper bounded by Eq.~\eqref{complexity_bound}.
\end{theorem}

Theorem~\ref{appendix_thm:sym2} indicates a reduction in the upper bound of the quantum sample complexity of Normal-IQAE Eq.~\eqref{complexity_bound_iqae} of $\frac{z^2_{\alpha_t/2}}{4K_t}\frac{1}{\varepsilon_{t-1}^2}$, which is consistent with a Bayesian perspective in two respects: (i) The measurement cost reduction depends on the accuracy achieved in the previous stage, as expected given that BIQAE leverages prior information from earlier stages to enhance the efficiency of amplitude estimation at a current stage, and (ii) in the case that the interval estimate in the previous stage is sufficiently accurate that its confidence interval radius exceeds that of the current stage $\varepsilon_{t-1}\leq\varepsilon_{t}$, the upper bound becomes less than or equal to zero, which implies that no additional sampling is required and the algorithm may proceed directly to the next stage.

\subsection{Pseudocode for Beta-BIQAE}
\label{appendix:biqae_beta}

In contrast to Normal-BIQAE, which updates the mean and variance of a normal distribution between iterations; Beta-BIQAE updates the shape parameters $a$ and $b$ of a beta distribution between iterations according to the workflow in
Fig.~\ref{appendix_fig:flow_chart_beta} and modules in Algorithms \ref{appendix_alg:BayesianUpdateBeta}--\ref{appendix_alg:PreparePriorBeta}.

\begin{figure}[h]
    \centering
    \includegraphics[width=0.5\textwidth]{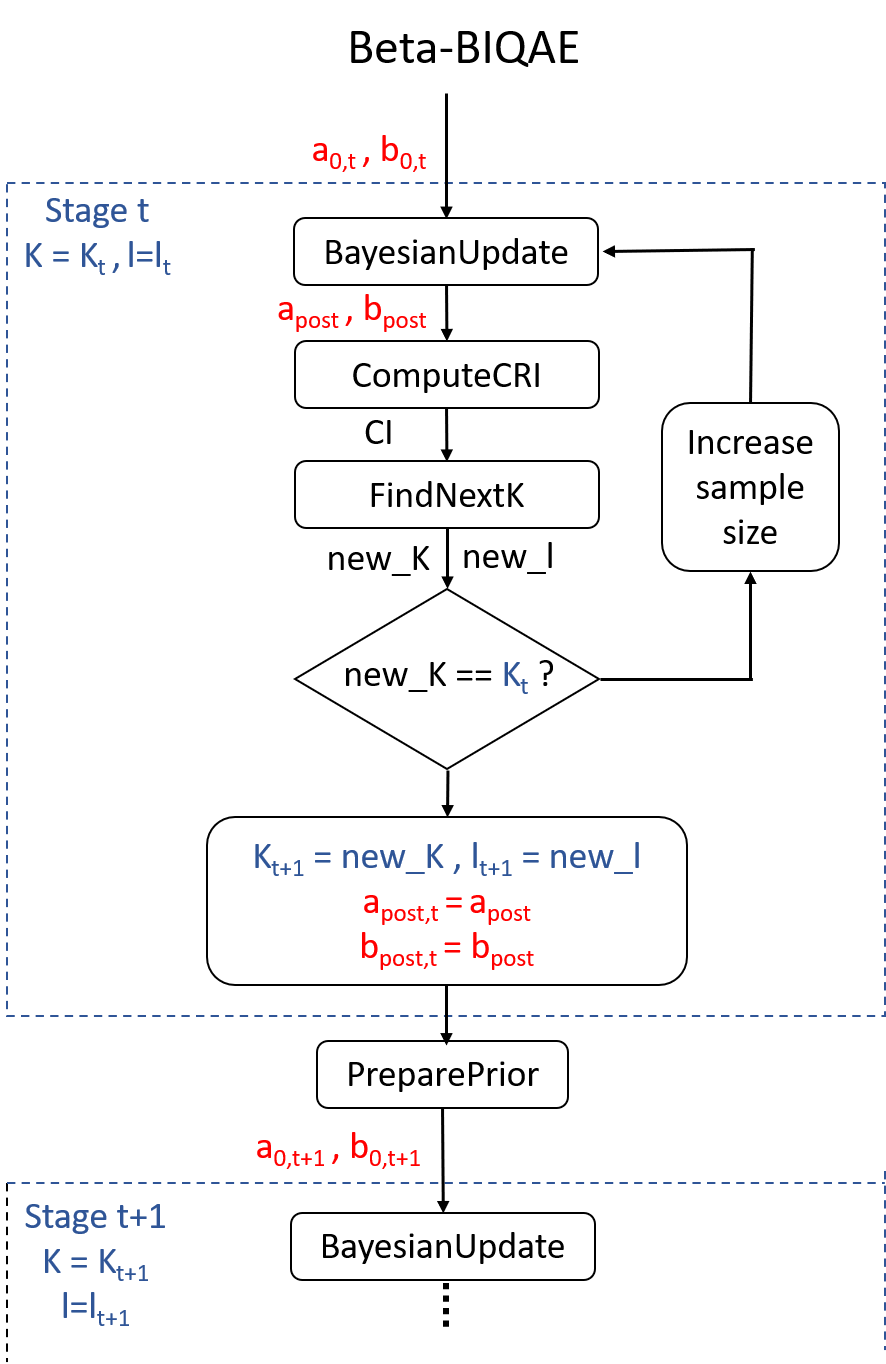}
    \caption{Detailed workflow of Beta-BIQAE, where parameter updates are emphasized in red to contrast with Normal-BIQAE.}
    \label{appendix_fig:flow_chart_beta}
\end{figure}

\begin{algorithm}[H]
\caption{BayesianUpdateBeta}
\label{appendix_alg:BayesianUpdateBeta}
\begin{algorithmic}[1]
\Require Prior shape parameters \(a_{0,t}\) and \(b_{0,t}\) for \(p_{k_t}\), number of \(\mathcal{Q}\) operators \(k_t\), sample size \(N_t\)
\Ensure Posterior shape parameters \(a_{\text{post},t}\) and variance \(b_{\text{post},t}\) for \(p_{k_t}\)

\State Construct \(k_t\) \(\mathcal{Q}\) operators into the circuit \(\mathcal{A}\)

\State Sample shots until a total of \(N_t\) shots (including existing shots at stage \(t\)) is collected. Store the sample mean as \(\bar{X}_{t}\)

\State Compute the posterior shape parameters:
\[
a_{\text{post},t} = a_{0, t} + N_t\bar{X}_{t},
\]
\[
b_{\text{post},t} = b_{0, t} + N_t\left(1-\bar{X}_{t}\right)
\]

\State \textbf{Return} \((a_{\text{post}, t}, b_{\text{post}, t})\)

\end{algorithmic}
\end{algorithm}

\begin{algorithm}[H]
\caption{ComputeCIBeta}
\label{appendix_alg:ComputeCRIBeta}
\begin{algorithmic}[1]
\Require Posterior shape parameters \((a_{\text{post}, t}, b_{\text{post}, t})\) for \(p_{k_t}\), amplification factor \(K_t\), quadrant index \(l_t\), confidence level \(\alpha_t\)
\Ensure An interval estimate \((\theta_{l,t}, \theta_{u,t})\) for \(\theta\) at confidence level \(\alpha_t\)

\State Compute the credible interval for \(p_{k_t}\) using the quantiles of the beta distribution:
\[
\left[p_{k_t}^l, p_{k_t}^u\right] = \left[\text{BetaCDF}^{-1}\left(\frac{\alpha_t}{2}, a_{\text{post}, t}, b_{\text{post}, t}\right), \text{BetaCDF}^{-1}\left(1 - \frac{\alpha_t}{2}, a_{\text{post}, t}, b_{\text{post}, t}\right)\right]
\]
where \(\text{BetaCDF}^{-1}\) denotes the inverse cumulative distribution function of the beta distribution.

\State Map the interval estimate for \(p_{k_t}\) to an interval estimate for \(\theta\) using \(f_t\):
\[
\left[\theta_t^l, \theta_t^u\right]=f_t([p_{k_t}^l, p_{k_t}^u])
\]
where 
\begin{align*}
    f_t(x) \coloneqq
    \begin{cases} 
        \frac{1}{K_t}\left(\arcsin{(\sqrt{x})} + l_t \cdot \frac{\pi}{2}\right) & \text{if \(l_t\) is even}, \\
        \frac{1}{K_t}\left(\arccos{(\sqrt{x})} + l_t \cdot \frac{\pi}{2}\right) & \text{if \(l_t\) is odd}
    \end{cases}
\end{align*}

\State \textbf{Return} \(\left[\theta_{t}^l, \theta_{t}^u\right] \)

\end{algorithmic}
\end{algorithm}

\begin{algorithm}[H]
\caption{PreparePriorBeta}
\label{appendix_alg:PreparePriorBeta}
\begin{algorithmic}[1]
\Require Posterior shape parameters \((a_{\text{post},t}, b_{\text{post},t})\) for \(p_{k_t}\) obtained at the end of stage \(t\), current and next amplification factors \(K_t, K_{t+1}\), current quadrant index \(l_t\)
\Ensure Shape parameters \((a_{0,t+1}, b_{0,t+1})\) for the beta prior of \(p_{k_{t+1}}\)

\State Sample \(R=1000\) i.i.d.~values \(\{y_1, y_2, \ldots, y_R\}\) from the beta posterior \(\text{Beta}(a_{\text{post},t}, b_{\text{post},t})\)

\State Map each sample \(y_i\) to the parameter space for \(p_{k_{t+1}}\) using the transformation:
\[
\tilde{y}_i \coloneqq \sin^2\left(K_{t+1} \cdot f_t(y_i)\right),
\]
where \(f_t\) is the mapping defined in Algorithm~\ref{appendix_alg:ComputeCRIBeta}

\State Fit a beta distribution to the transformed samples \(\{\tilde{y}_1, \tilde{y}_2, \ldots, \tilde{y}_R\}\) using maximum likelihood estimation (MLE) to obtain shape parameters \((a_{0,t+1}, b_{0,t+1})\)

\State \textbf{Return} \((a_{0,t+1}, b_{0,t+1})\)

\end{algorithmic}
\end{algorithm}

Formally, Theorem~\ref{appendix_thm:MLE_convergence_beta_prior} establishes that the beta prior fitted by Algorithm~\ref{appendix_alg:PreparePriorBeta} is the asymptotically optimal beta approximation to the exact prior distribution for \(p_{k_{t+1}}\) that minimizes the Kullback-Leibler (KL) divergence between the exact prior and the fitted beta distribution. The proof is provided in Appendix~\ref{appendix:proof_approx_prior_beta}.

\begin{theorem}
\label{appendix_thm:MLE_convergence_beta_prior}
Let \(g\) denote the density function of the exact prior distribution for \(p_{k_{t+1}}\) derived from transforming the posterior distribution for \(p_{k_t}\), and let \(\{q_{\alpha,\beta}:(\alpha, \beta)\in \Theta\coloneqq [\alpha_{\min}, \alpha_{\max}]\times[\beta_{\min}, \beta_{\max}]\}\) be a parametric family of beta distributions indexed by \((\alpha, \beta)\). Suppose \(R\) i.i.d.~samples \(\tilde{Y}_1, \dots, \tilde{Y}_R \sim g\) are generated by mapping samples from the beta posterior of \(p_{k_t}\) (as described in Algorithm~\ref{appendix_alg:PreparePriorBeta}), and let \((\hat{\alpha}_R, \hat{\beta}_R)\) be the MLE estimates obtained by maximizing the average log-likelihood over the parameter space $\Theta$
\[
(\hat{\alpha}_R, \hat{\beta}_R) = \arg\max_{(\alpha, \beta)\in \Theta} \frac{1}{R} \sum_{i=1}^R \log q_{\alpha,\beta}(\tilde{Y}_i).
\]
Define the optimal shape parameters \((\alpha^*, \beta^*)\) as
\[
(\alpha^*, \beta^*) = \arg\min_{(\alpha, \beta)\in \Theta} D_{\mathrm{KL}}(g \| q_{\alpha,\beta}),
\]
where \(D_{\mathrm{KL}}\) represents the Kullback-Leibler (KL) divergence.

Then, \((\hat{\alpha}_R, \hat{\beta}_R)\) converges almost surely to \((\alpha^*, \beta^*)\) as \(R \to \infty\). 
\end{theorem}

Note that, although the beta prior fitted in Algorithm~\ref{appendix_alg:PreparePriorBeta} is the best approximation within the beta family, it is not guaranteed to be close to the exact prior. In principle, it is possible that the beta family is not expressive enough, and that even the optimal beta approximation cannot closely represent the exact prior. However, this phenomenon is unlikely to occur in practice for BIQAE: Informally, for any smooth and unimodal distribution on \([0,1]\), there exists a beta distribution that can approximate it closely. We leave the rigorous theoretical justification of this claim to future work. Fig.~\ref{appendix_fig:beta_approx} shows the quality of this approximation with an example.
\begin{figure}[ht] 
    \centering
    \includegraphics[width=1.\textwidth]{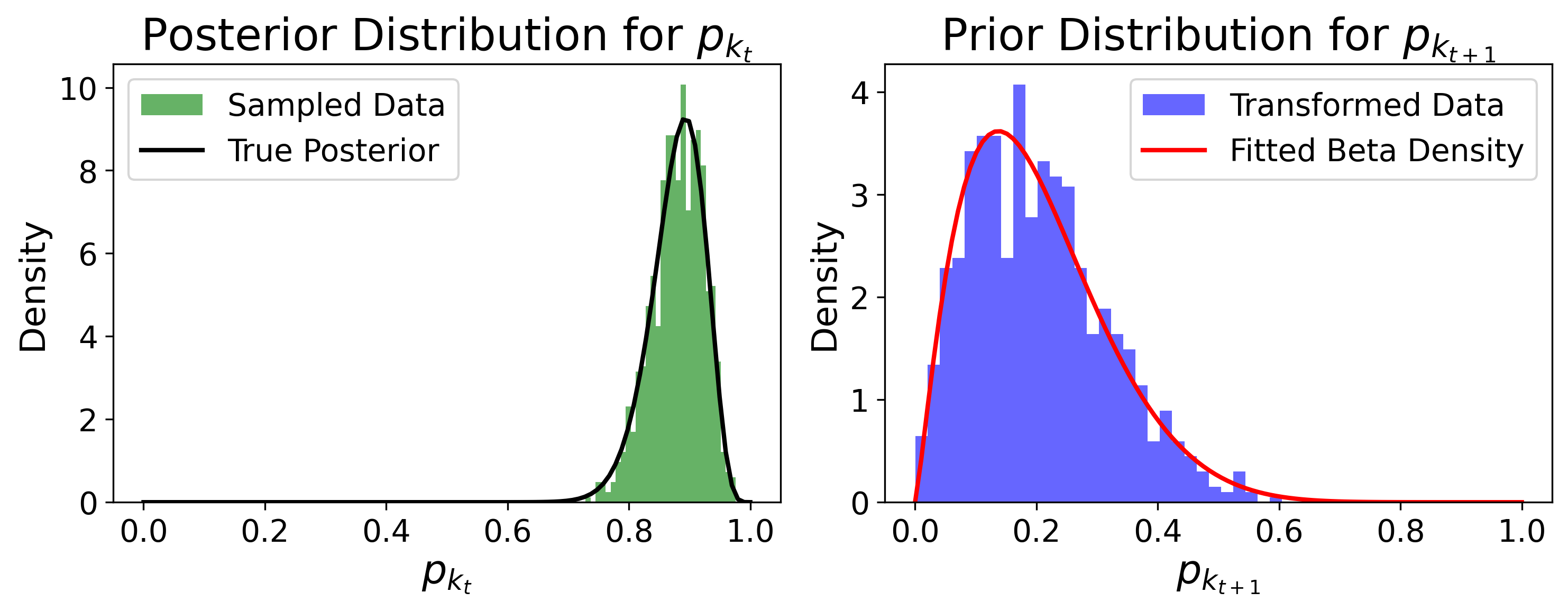}
    \caption{Beta approximation to the exact prior for $p_{k_{t+1}}$. We first generate (left) sampled data (green bars) from the true posterior (black line) and then transform the sampled data as $\sin^2{(K_{t+1}\cdot f_t(\cdot))}$ to obtain (right) the transformed data (blue bars). Finally, we fit a beta distribution (red line) to the transformed data. The close alignment between the red line and the blue bars demonstrates the accuracy of the approximation.}
    \label{appendix_fig:beta_approx}
\end{figure}

\subsection{Equivalence of Interval Estimates}
\label{appendix:eq_intv}

To demonstrate that IQAE can be derived as an implementation of BIQAE with a noninformative prior, we establish the similitude of the interval estimates in IQAE-CH and Normal-BIQAE and in IQAE-CP and Beta-BIQAE, as follows:

Consider Normal-BIQAE with a noninformative prior \(\Delta_t^2 \to \infty\). A Bayesian update according to Algorithm~\ref{appendix_alg:BayesianUpdate} yields a posterior distribution characterized by
\begin{align*}
    \mu_{\text{post},t} \to \bar{X}_t, \quad
    \sigma_{\text{post},t}  \to \sqrt{\frac{\bar{X}_t(1-\bar{X}_t)}{N_t}}
\end{align*}
regardless of \(\mu_{0,t}\). Approximation of the z-value as
\begin{align*}
    z_{\alpha_t/2} \approx \sqrt{2 \log\left(\frac{2}{\alpha_t}\right)}
\end{align*}
then yields the credible interval
\begin{align}
    \label{appendix_ci:pkt_non_info}
    \mu_{\text{post},t} \pm z_{\alpha_t/2} \cdot \sigma_{\text{post},t} \approx \bar{X}_t \pm \sqrt{\frac{2\bar{X}_t(1-\bar{X}_t)}{N_t}\log\left(\frac{2}{\alpha_t}\right)},
\end{align}
which is directly comparable to the Chernoff-Hoeffding confidence interval for $p_{k_t}$
\begin{align}
    \label{appendix_ci:pkt_iqae_ch}
    \bar{X}_t \pm \sqrt{\frac{1}{2N_t}\log\left(\frac{2}{\alpha_t}\right)},
\end{align}
where the truncation function $\mathrm{trunc}_{[0,1]}(\cdot)$ is omitted for simplicity.
The credible interval of Normal-BIQAE Eq.~\eqref{appendix_ci:pkt_non_info} and the confidence interval of IQAE-CP Eq.~\eqref{appendix_ci:pkt_iqae_ch} are thus structurally similar, possessing identical centers with a discrepancy only in first factor under the square root that enables the Chernoff-Hoeffding confidence interval to be slightly more conservative since \(\frac{1}{4}\) is always no less than \(\bar{X}_t(1-\bar{X}_t)\), such that construction of interval estimates with the two methods is highly similar; we thus do not distinguish linguistically between the Normal-IQAE method corresponding to the approach entailed by Eq.~\eqref{appendix_ci:pkt_non_info} and the IQAE-CH method corresponding to the approach entailed by Eq.~\eqref{appendix_ci:pkt_iqae_ch} in this paper. 

Likewise, the credible interval of Beta-BIQAE according to Algorithm~\ref{appendix_alg:ComputeCRIBeta} 
\begin{equation}
\left[p_{k_t}^l, p_{k_t}^u\right] = \left[\text{BetaCDF}^{-1}\left(\frac{\alpha_t}{2}, a_{\text{post}, t}, b_{\text{post}, t}\right), \text{BetaCDF}^{-1}\left(1 - \frac{\alpha_t}{2}, a_{\text{post}, t}, b_{\text{post}, t}\right)\right]
\end{equation}
is directly comparable to the Clopper-Pearson confidence interval used in IQAE-CP 
\begin{equation}
\left[p_{k_t}^l, p_{k_t}^u\right] = \left[\text{BetaCDF}^{-1}\left(\frac{\alpha_t}{2}, N_t\bar{X}_{t}, 1 + N_t\left(1-\bar{X}_{t}\right)\right), \text{BetaCDF}^{-1}\left(1 - \frac{\alpha_t}{2}, 1 + N_t\bar{X}_{t},N_t\left(1-\bar{X}_{t}\right)\right)\right]
\end{equation}
for a noninformative prior $a_{0,t}=b_{0,t}=0.5$, in particular for relatively large sample sizes $N_t$ (namely, for  \textit{Jeffreys prior} with corresponding credible interval \textit{Jeffreys interval}). Given the absence of  significant differences between practical applications of the two approaches, we do not distinguish linguistically between the IQAE-CP approach using Clopper-Pearson confidence intervals and the Beta-IQAE approach using Jeffreys intervals in this paper.

\subsection{Settings for Quantum Sample Complexity Scaling Analysis}
\label{appendix:circ_math_simul}

As a proof-of-concept, numerical simulations of the quantum sample complexity for arbitrary quantum states are performed using a simplified circuit design with mock \(\mathcal{Q}\)-operators that eliminate the impact of circuit depth on classical simulation runtime---a simplification that becomes essential as the target accuracy increases, implying increasingly prohibitively time-consuming numbers of $\mathcal{Q}$-operators. The circuit, shown in Fig.~\ref{appendix_fig:mock_Q}, employs a single $R_Y(2k\theta)$ gate to emulate $\mathcal{Q}^k$, in which the behavior of the circuit is constructed with knowledge $\theta$, but the algorithm is executed as if $\theta$ were unknown.

\begin{figure}
\begin{center}
\begin{quantikz}
\lstick{$q_0$} & \gate{R_Y(\theta)} 
\gategroup[1,steps=1,style={dashed, rounded corners, inner sep=3pt}]{$A$} & \qw 
& \gate{R_Y(2k\theta)} 
\gategroup[1,steps=1,style={dashed,rounded corners, inner sep=3pt}]{$Q^k$} 
& \qw
\end{quantikz}
\end{center}
\caption{Circuit diagram for the quantum sample complexity scaling analysis.}
    \label{appendix_fig:mock_Q}
\end{figure}
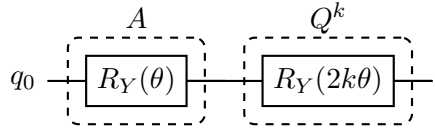 
Note that although this mock $\mathcal{Q}$ operator is used to facilitate classical simulations, all presented circuit depths and quantum sample complexities accurately account for the internal structure of the $\mathcal{Q}$ operator. Specifically, the total circuit depth is calculated as
\begin{align*}
    \text{total circuit depth} = \text{circuit depth of }\mathcal{A} + k \cdot \text{circuit depth of }\mathcal{Q},
\end{align*}
and, since each $\mathcal{Q}$ operator contains two $\mathcal{A}$ operators, the quantum sample complexity weights each set of $N_k$ shots of the quantum circuit by either $k$ (by the IQAE convention) or $K=2k+1$ (by the rigorous count of oracle calls).

\subsection{Settings for Molecular Ground-State Energy Estimation}
\label{appendix:vqe_exps}
To identify the molecular ground-state wavefunction, the electronic Hamiltonian for each molecule and bond length considered is generated in Qiskit Nature \cite{qiskit2024} via the PySCF driver \cite{sun2018pyscf} in parity mapping \cite{seeley2012bravyi}. Frozen core orbitals and qubit tapering \cite{bravyi2017tapering} are then employed to reduce the required number of qubits to represent the Hamiltonian.\footnote{Note Jordan-Wigner mapping \cite{jordan1928uber} without qubit tapering is employed for $\text{H}_2$.}
The ground-state wavefunction is then determined via the Variational Quantum Eigensolver (VQE) \cite{peruzzo2014variational,cao2019quantum} for the efficient SU(2) Ansatz with full entanglement, as depicted in Fig.~\ref{fig:efficientsu2}; the parameters of the Ansatz are optimized classically with Sequential Least Squares Programming (SLSQP); and the expectation value corresponding to each parameter set is computed classically via the Qiskit Runtime Estimator primitive service with zero shots (see additional details in Table \ref{tab:size_reduction}). In all molecular ground-state energy calculations, the electronic energy is supplemented by the classically calculated nuclear repulsion energy and core electron energy as determined with PySCF.      

    \begin{figure}
        \begin{center}
\begin{quantikz}
\lstick{$q_0$} & \gate{R_Y\left(\theta_0\right)} & \gate{R_Z\left(\theta_3\right)} & \qw  
    \gategroup[3,steps=7,style={dashed,rounded corners, inner sep=3pt}]{Repeated Layer $i$} & \ctrl{1} & \qw      & \ctrl{2} &
    \gate{R_Y\left(\theta_{6i}\right)} & \gate{R_Z\left(\theta_{6i+3}\right)} & \qw \\

\lstick{$q_1$} & \gate{R_Y\left(\theta_1\right)} & \gate{R_Z\left(\theta_4\right)} & \qw & \targ{}  & \ctrl{1}   & \qw  
    & \gate{R_Y\left(\theta_{6i+1}\right)} & \gate{R_Z\left(\theta_{6i+4}\right)} & \qw \\

\lstick{$q_2$} & \gate{R_Y\left(\theta_2\right)} & \gate{R_Z\left(\theta_5\right)} & \qw & \qw      & \targ{}  & \targ{}    
    & \gate{R_Y\left(\theta_{6i+2}\right)} & \gate{R_Z\left(\theta_{6i+5}\right)} & \qw
\end{quantikz}
\end{center}
        \caption{Efficient SU(2) Ansatz, as defined by a single base layer and repeated layers $i=1,2,\ldots,L$. The Ansatz for a single repeated layer $L=1$ is shown.}
        \label{fig:efficientsu2}
    \end{figure}
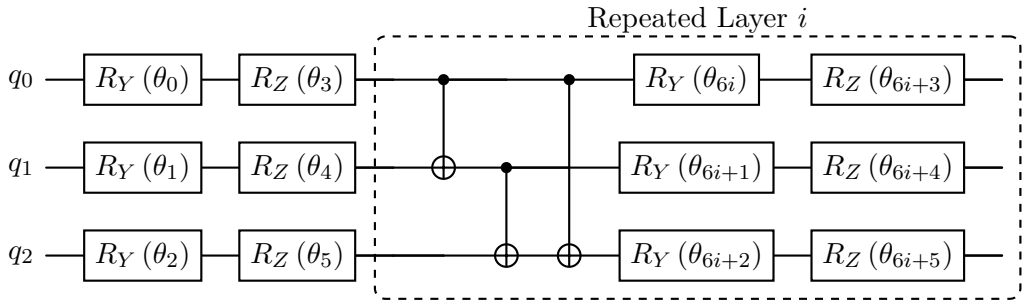

\begin{table*}
    \centering
\begin{tabular}{ccccc}
\toprule 
\multirow{2}{*}{\textbf{Molecule}} & \multicolumn{2}{c}{\textbf{Qubit Number}} & \multirow{2}{*}{\textbf{Ansatz Layer Number}} & \multirow{2}{*}{\textbf{Max. Iter. Number}}\tabularnewline
\cmidrule{2-3} 
 & \textbf{Before Taper.} & \textbf{After Taper.} &  & \tabularnewline
\midrule
\midrule 
$\text{H}_{2}$ & 4 & N/A & 12 & 300\tabularnewline
\midrule 
LiH & 10 & 6 & 12 & 300\tabularnewline
\midrule 
HF & 10 & 6 & 12 & 300\tabularnewline
\midrule 
$\text{BeH}_{2}$ & 12 & 7 & 15 & 500\tabularnewline
\bottomrule
\end{tabular}
    \caption{Number of qubits (before and after qubit tapering), layers, and maximum iterations used to determine the electronic ground-state wavefunction for each molecule considered. The number of qubits is fixed according to the form of the mapped Hamiltonian, and the number of Ansatz layers and maximum iterations is chosen to converge the ground-state wavefunction.}
    \label{tab:size_reduction}
\end{table*}

\section{Extended Results}
\label{appendix:ext_res}

\subsection{Linear Regression Analysis} 
\label{ext_res:scaling_analysis}

As detailed in Table~\ref{tab:loglog_regression}, regression of the quantum sample complexity as a function of the median absolute error in Fig.~\ref{fig:compare_all} suggests Beta-BIQAE is robust, improves upon Classical QAE quadratically, and features the lowest quantum sample complexity of all QAE methods considered. Linear regression of the log-log data to produce slope, intercept, and $R^2$ values suggests all of these QAE methods (with the possible exception of FQAE) reliably obey the expected power law scaling, with $R^2$ values exceeding $0.99$. The slopes indicate that methods based on amplitude amplification offer a quadratic advantage over Classical QAE, as, for example, Beta-BIQAE's slope of $-1.0$ and Classical QAE's slope of $-2.0$ correspond to quantum sample complexities of $N_\text{oracle}\in \mathcal{O}(\varepsilon^{-1})$ and $\mathcal{O}(\varepsilon^{-2})$, respectively. Similarly, the intercept, which represents the logarithm of the constant overhead, implies that, of all of the QAE methods considered that achieve the quadratic advantage, Beta-BIQAE has the lowest constant factor of the quantum sample complexity.

\begin{table}
\centering
\caption{Linear regression results for the log-log relationship between the quantum sample complexity and the estimation error ($\log_{10} N_\text{oracles}\sim \log_{10} \varepsilon$).}
\label{tab:loglog_regression}
\begin{tabular}{lccc}
\toprule
\textbf{Method} & \textbf{Slope} & \textbf{Intercept} & \textbf{$R^2$} \\
\midrule
Beta-BIQAE      & $-1.0088$  & $0.0211$    & $0.9999$ \\
IQAE-CP        & $-1.0044$  & $0.1253$    & $0.9999$ \\
BAE             & $-1.0137$  & $0.1089$    & $0.9996$ \\
Canonical QAE   & $-1.0058$  & $0.4847$    & $1.0000$ \\
Classical QAE   & $-1.9821$  & $-1.1693$   & $0.9996$ \\
IQAE\_CH        & $-1.0056$  & $0.3930$    & $0.9999$ \\
MLQAE\_EIS      & $-0.9820$  & $0.5626$    & $0.9999$ \\
MLQAE\_LIS      & $-1.3375$  & $-0.3395$   & $0.9979$ \\
FQAE            & $-1.2091$  & $3.6295$    & $0.9677$ \\
QAES            & $-0.9973$  & $6.8670$    & $0.9999$ \\
\bottomrule
\end{tabular}
\end{table}

\subsection{Interval Radii and Coverage Rates}
\label{appendix:biqae_vs_bae}

Beta-BIQAE is also found to provide more accurate and reliable quantum amplitude estimates than IQAE-CP and BAE from an interval radius and coverage rate perspective, where the radius of the interval estimate for BAE is approximated as the product of the critical value $v$ and standard deviation and where the \textit{coverage rate} is defined as the proportion of repeated independent experiments in which the interval contains the true parameter value. As shown in Fig.~\ref{appendix_fig:compare_intervals_median} (left panel), for the critical value $v=1.96$, which corresponds to 95\% coverage under the assumption that the sampling distribution of the point estimate is normal, both Beta-BIQAE and IQAE-CP have a higher coverage rate for all target accuracies considered than BAE, which provides the narrowest intervals (and thus the highest accuracy) but exhibits a coverage rate of only $\le87\%$. Likewise, as shown in Fig.~\ref{appendix_fig:compare_intervals_median} (right panel), for $v=10$, which corresponds to coverage of all but $1.523971\times10^{-21}\%$ for a normal distribution, Beta-BIQAE and IQAE-CP have higher coverage rates than BAE, which features coverage rates as low as 87.5\% and the largest intervals. Note in the $v=10$ case Beta-BIQAE yields the smallest intervals and frequently the highest coverage rates for all cases considered, which is indicative of an advantageous balance between a narrow interval radius and a high coverage rate ({\em i.e.}, accuracy and reliability). The observed performance difference between BIQAE and BAE is also consistent with BIQAE and BAE's distinct methodological approaches. Namely, although BIQAE and BAE both employ Bayesian updates within an iterative QAE framework to extract information from measurement outcomes, (i) BIQAE adopts the $K$-scheduling strategy of IQAE to provide interval estimates for the amplitude in place of BAE's Bayesian experimental design strategy that schedules $K$ by maximizing the utility function to obtain a point estimate with its standard deviation and (ii) BIQAE utilizes conjugate priors to perform Bayesian updates (in the case of the beta distribution, with the advantage of small-sample efficiency enhancement) in place of BAE's Sequential Monte Carlo.

\begin{figure}[ht]
    \centering   
    \includegraphics[width=\textwidth]{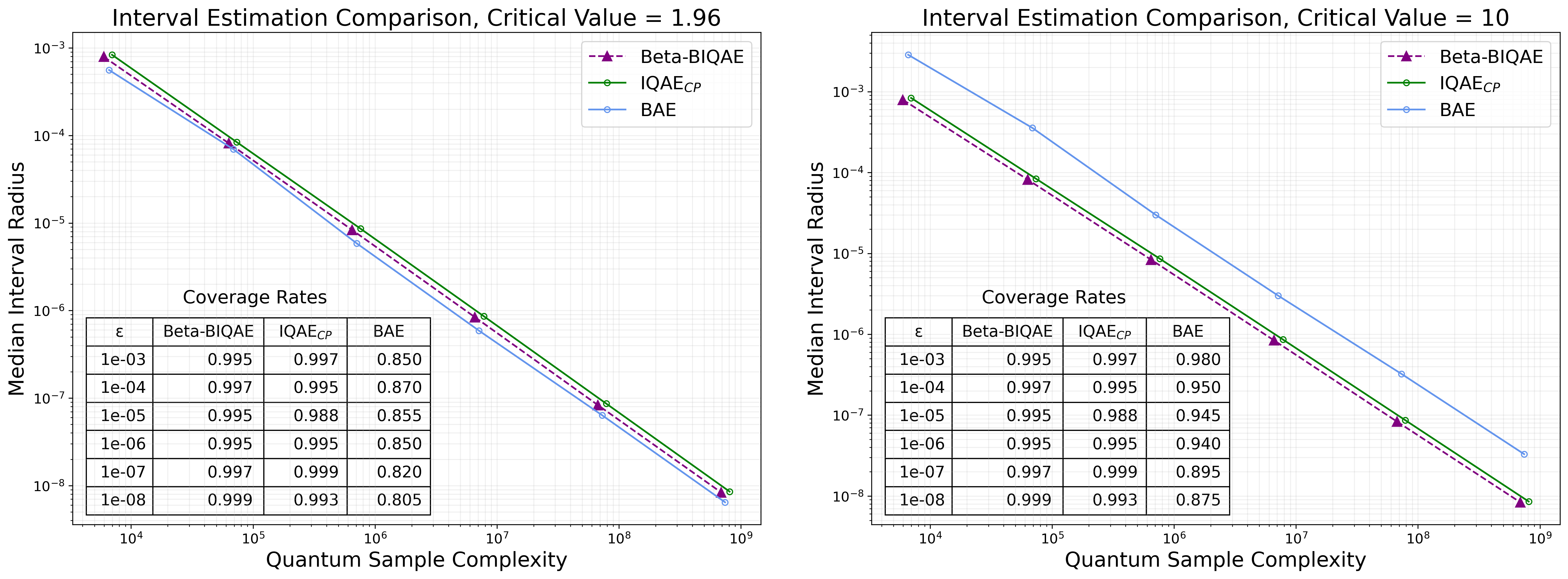}
    \caption{Comparison of the median radius and the coverage rate for BAE and IQAE-based estimates for $a=0.5$ from an interval estimation perspective. The critical value for BAE is selected as (left) $1.96$ or (right) $10$.}\label{appendix_fig:compare_intervals_median}
\end{figure}

\subsection{Robustness of Beta-BIQAE Advantage Over IQAE-CP}\label{appendix:biqae_vs_iqae}

Fig.~\ref{fig:angle_plots} shows that the near-constant-factor improvement in the sample complexity of Beta-BIQAE over IQAE-CP for all quantum amplitude values $a$ shown for $\varepsilon=10^{-8}$ in Fig.~\ref{fig:all_angle_plot_epsilon_1e-08} persists for target accuracies $\varepsilon\in\{10^{-2},10^{-3},10^{-4},10^{-5},10^{-6},10^{-7}\}$. The cost reduction of Beta-BIQAE over IQAE-CP is consistently 10.5--12.8\% and increases for higher target accuracies. Note deviations from the scaling laws for both Beta-BIQAE and IQAE-CP often appear at the same $a$ values across  target accuracies, an interesting area for further study.

\begin{figure}[htbp]
    \centering
    \includegraphics[width=\linewidth]{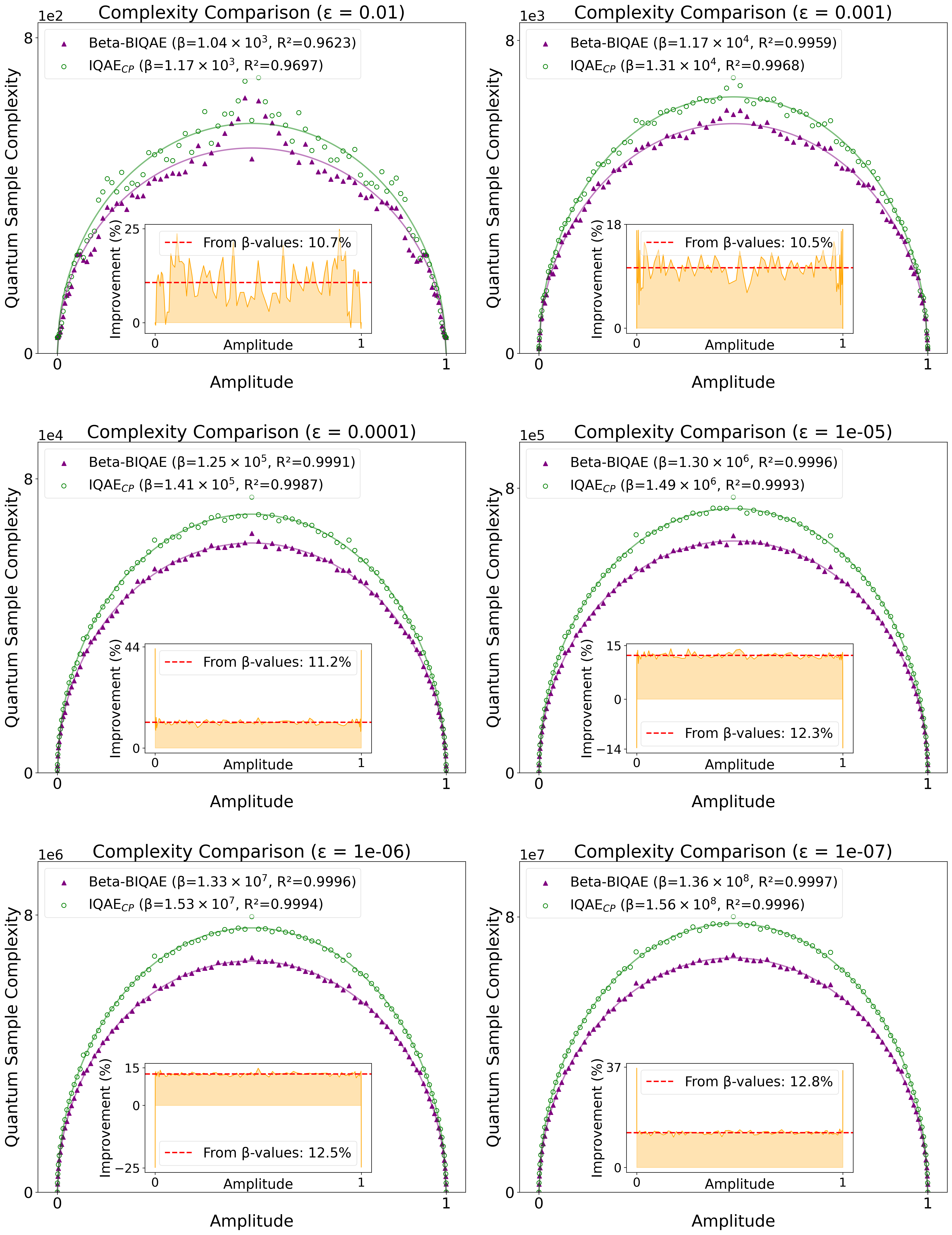}
    \caption{Comparison of the quantum sample complexities of Beta-BIQAE and IQAE-CP as a function of the quantum amplitude for target accuracies $\varepsilon\in\{10^{-2},10^{-3},10^{-4},10^{-5},10^{-6},10^{-7}\}$, as in Fig.~\ref{fig:all_angle_plot_epsilon_1e-08}.}
    \label{fig:angle_plots}
\end{figure}

\subsection{Circuit Depth Comparison}
\label{ext_res:vqe_circ_depth}

Consideration of the quantum circuits required for molecular ground-state energy estimation suggests that the average number of Grover $\mathcal{Q}$ operators needed to estimate the expectation value of Pauli terms is similar for Beta-BIQAE and IQAE-CP for all molecules considered. Given that the number of such operators is the major contributor to the quantum circuit depth in both methods, these results suggest Beta-BIQAE's near-constant-factor improvement over IQAE-CP is largely accomplished without growth of the cost of quantum computation.

\begin{figure*}
    \centering
    \begin{tabular}{cc}  
        \includegraphics[width=0.5\textwidth]{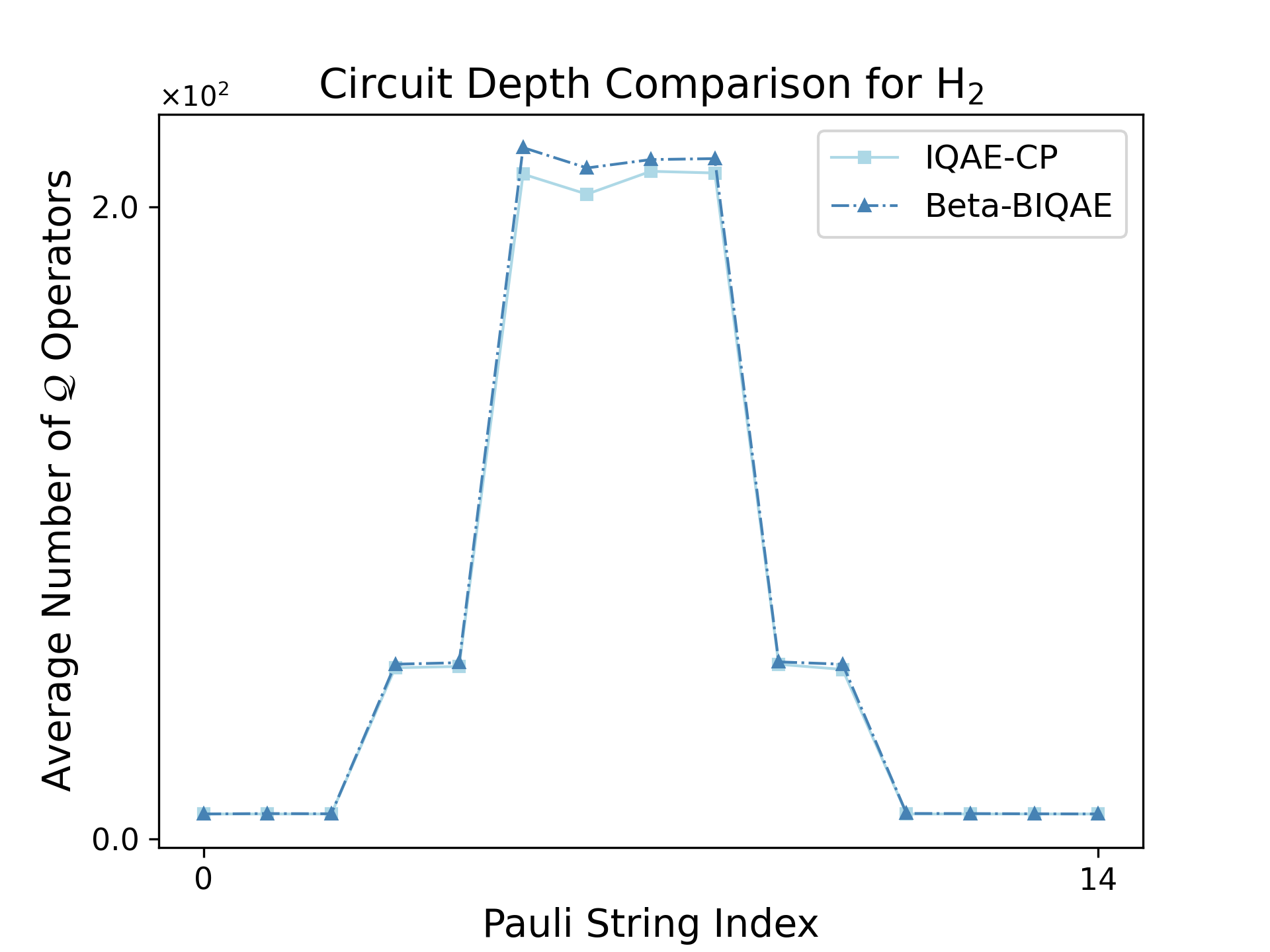} &  
        \includegraphics[width=0.5\textwidth]{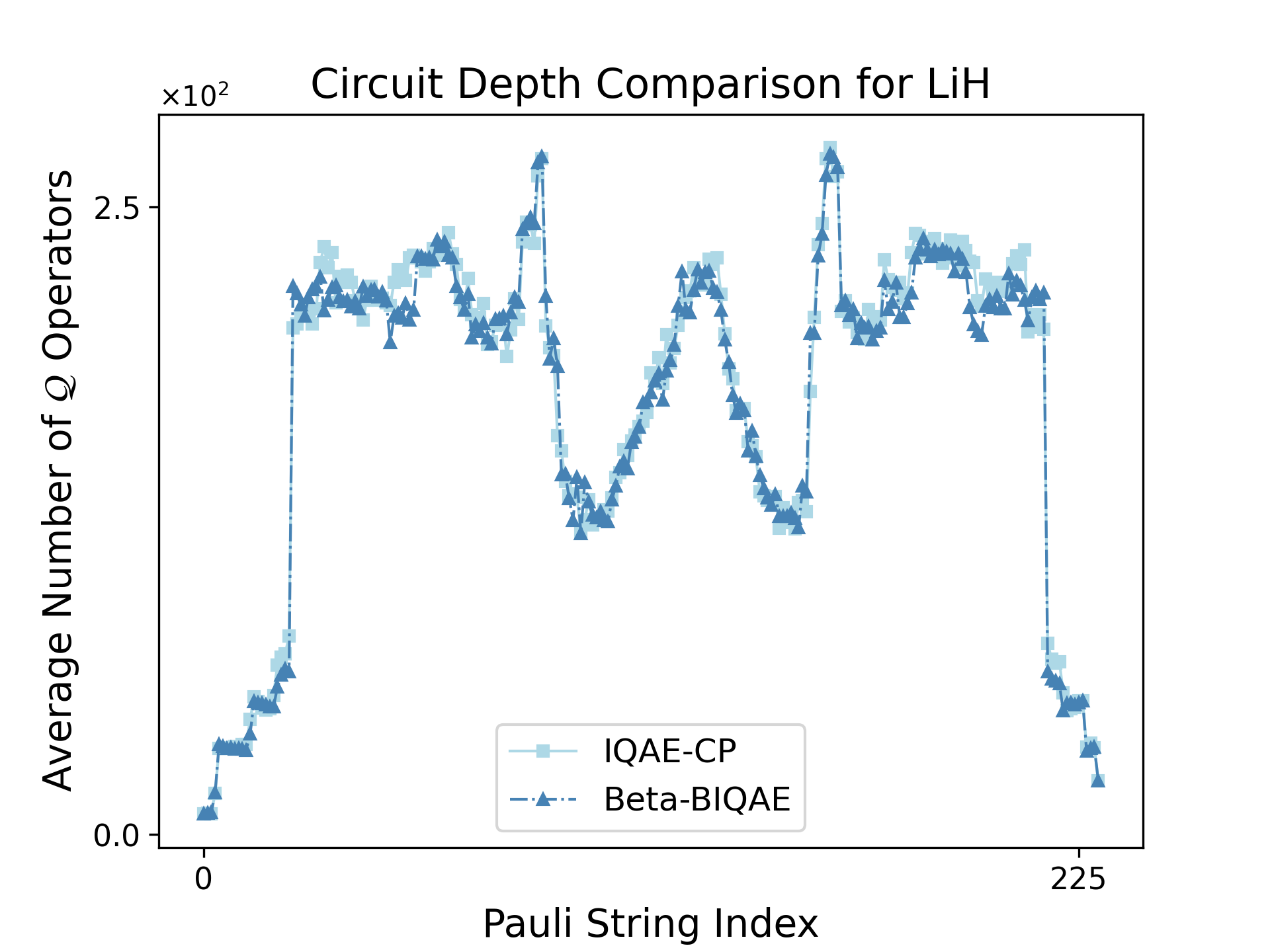} \\

        \includegraphics[width=0.5\textwidth]{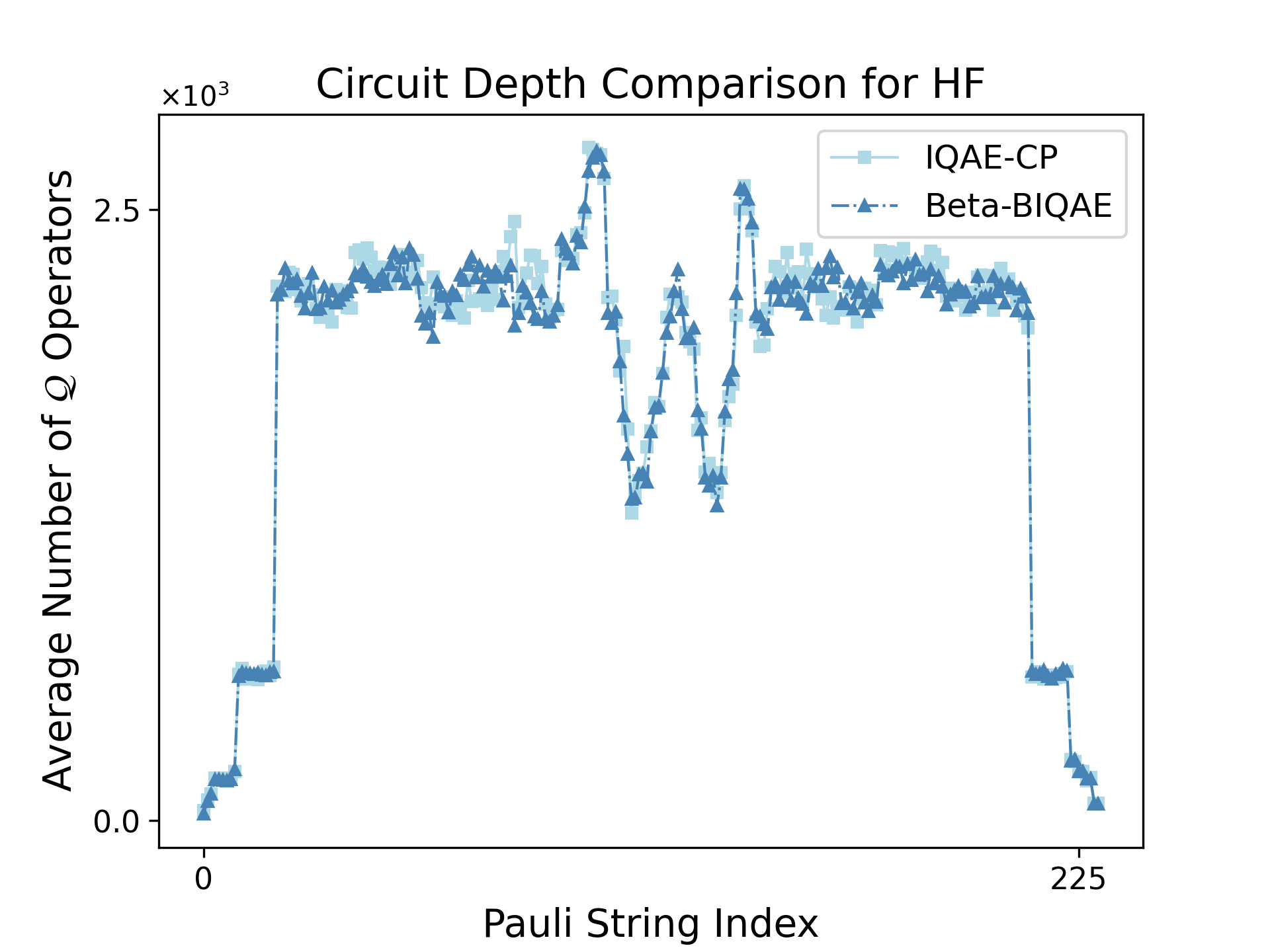} &  
        \includegraphics[width=0.5\textwidth]{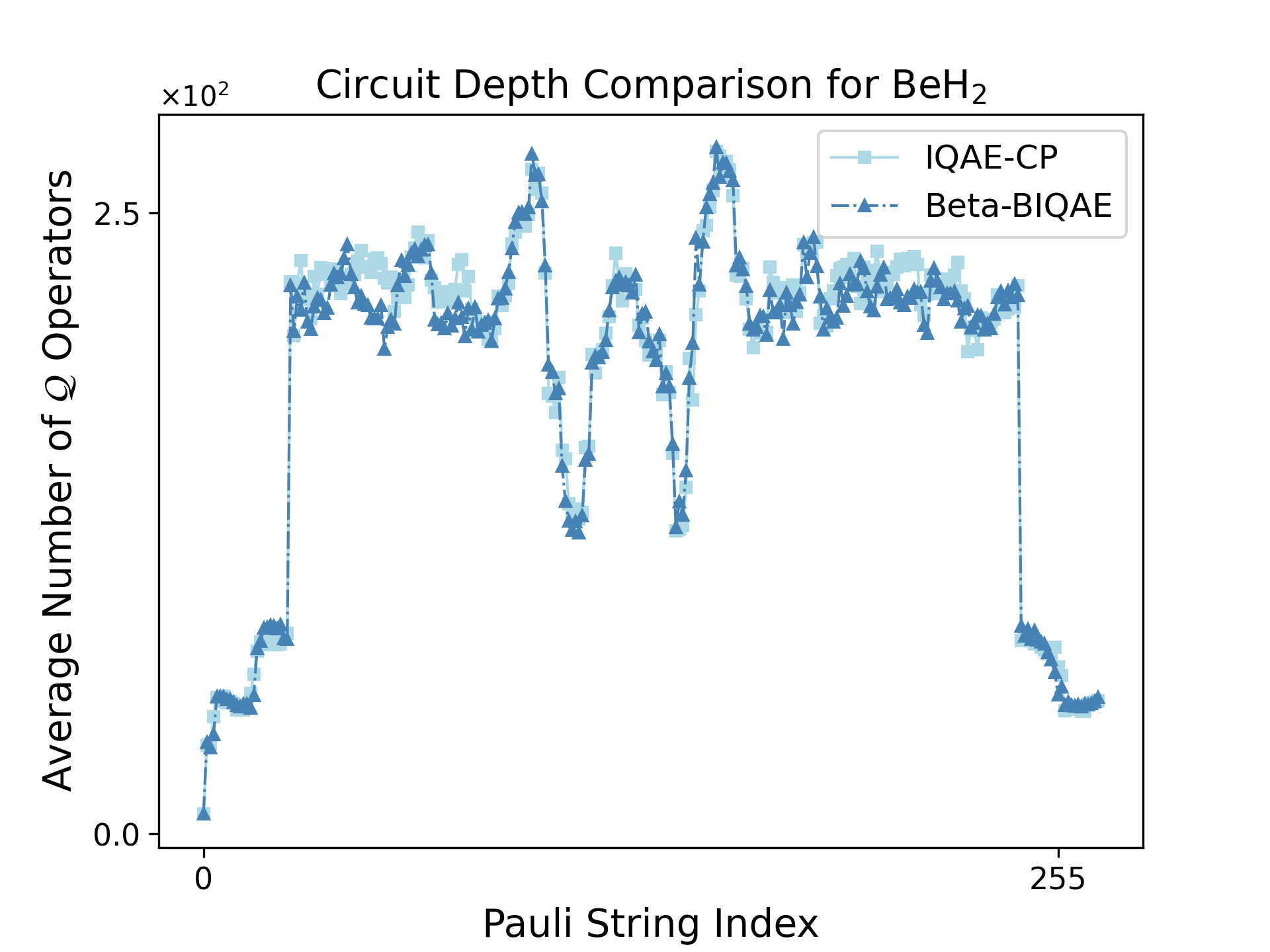} 
    \end{tabular}
    
    \caption{Comparison of the number of Grover operators in the last stage required to estimate Pauli string expectation values with Beta-BIQAE (medium blue triangles) and IQAE-CP (light blue squares). The number is averaged over 200 repetitions.}
    \label{fig:comparison_dep}
\end{figure*}

\subsection{Comparison of BIQAE and IQAE with beta and normal priors}\label{appendix:beta_vs_normal}

\begin{figure}[htbp]
    \centering
    \includegraphics[width=0.8\linewidth]{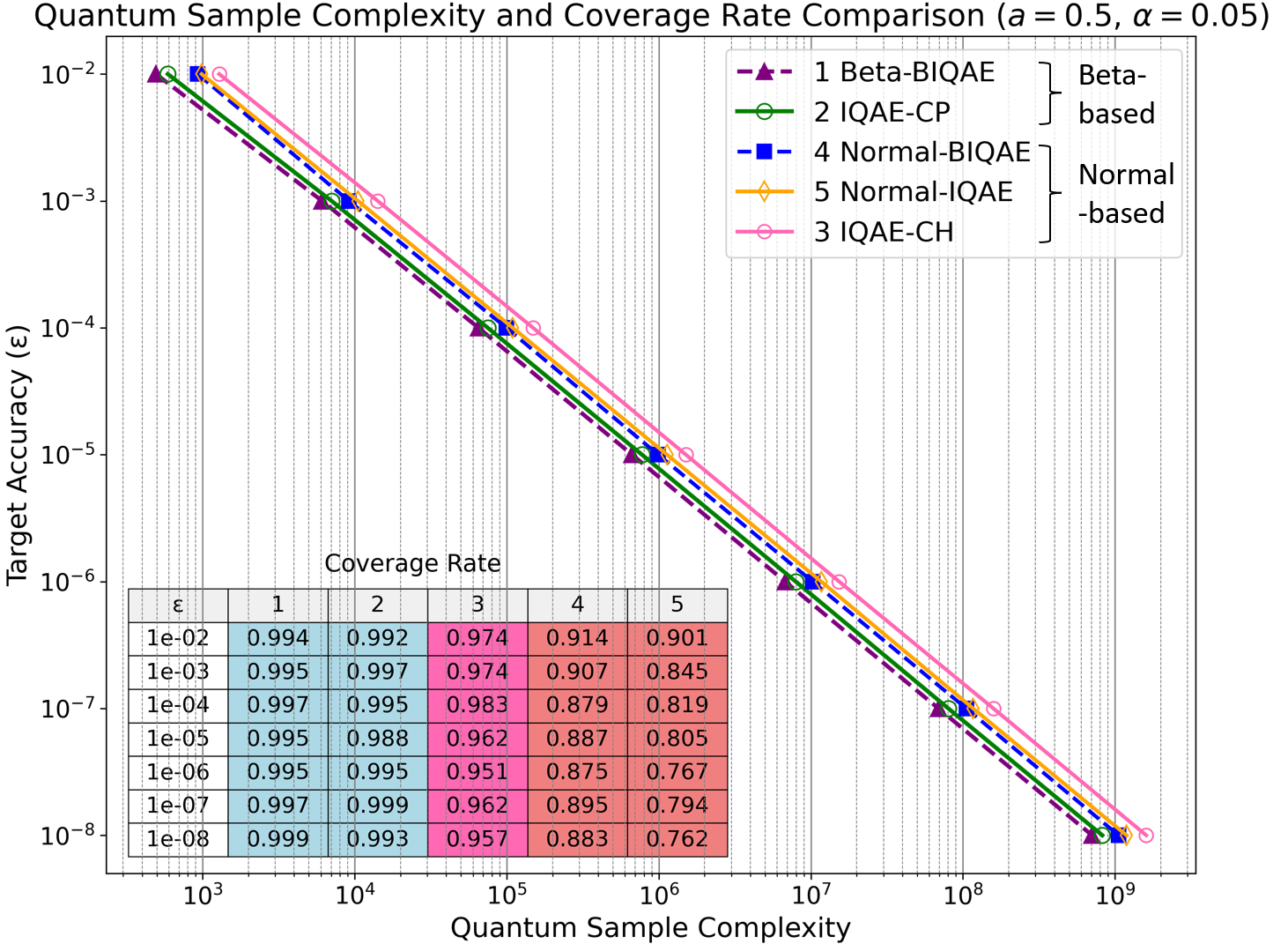}
    \caption{Comparison of quantum sample complexity and coverage rate between beta-based methods (Beta-BIQAE and IQAE-CP) and normal-based methods (Normal-BIQAE, Normal-IQAE, and IQAE-CH). Legends are ordered from the lowest to the highest quantum sample complexity as a function of target accuracy. Inset table cells are shaded to indicate low coverage associated with failure to achieve the target coverage rate of 0.95 (red), intermediate coverage between 0.95 and 0.99 (pink), and high coverage above 0.99 (blue).
}
    \label{fig:beta_vs_normal}
\end{figure}

Fig.~\ref{fig:beta_vs_normal} supports the predicted superiority of the quantum sample complexity and coverage rate of beta-based methods (Beta-BIQAE and IQAE-CP, which employ informative and non-informative beta priors, respectively) over normal-based methods (Normal-BIQAE, Normal-IQAE, and IQAE-CH, which employ an informative normal prior, a non-informative normal prior, and the Chernoff--Hoeffding bound, respectively, where the Chernoff--Hoeffding bound is asymptotically equivalent to the noninformative normal prior up to a more conservative constant factor, as shown in Appendix~\ref{appendix:eq_intv}) in the small per-stage sample size limit. For all target accuracy values considered, the beta-based methods exhibit lower complexity than the normal-based methods. The coverage rate for beta-based methods is likewise superior to that of normal-based methods, notwithstanding the imposition in the Normal-IQAE and Normal-BIQAE implementations of a minimum sample size of at least 30 shots per stage to avoid zero-variance estimates associated with diminished coverage rates and algorithmic failure. For all target accuracies considered, beta-based Beta-BIQAE and IQAE-CP---but not normal-based Normal-BIQAE and Normal-IQAE---achieve the target coverage rate of 0.95. This distinction between the two classes of methods can be attributed to the normal-based methods' variance estimator $\hat{\sigma}_t^2 = \bar{X}_t(1-\bar{X}_t)$, which requires a substantially larger per-stage sample size to have a comparable accuracy to that of the beta-based IQAE or BIQAE settings. IQAE-CH mitigates this issue by replacing the variance estimator with the upper bound $\bar{X}_t(1-\bar{X}_t)\le \tfrac{1}{4}$, which stabilizes the coverage rate at the cost of increased complexity due to its conservativeness, in keeping with its coverage intermediate between the remaining normal- and beta-based methods considered. By contrast, the use of beta priors in Beta-BIQAE and IQAE-CP eliminates the small-sample instability entirely, which enables both low quantum sample complexities and high coverage rates. 

Note a further observation from Fig.~\ref{fig:beta_vs_normal} is that Normal-BIQAE consistently improves upon Normal-IQAE, as expected.

\FloatBarrier

\section{Proofs}
\label{appendix:proofs}
\subsection{MLE and AMSE of QAE}
\label{appendix:proof_thm123}

{\em Proofs of Theorem~\ref{appendix_thm:knownbase} and Appendix Theorems~\ref{appendix_thm:ClassicalQAE} and \ref{appendix_thm:knownbase}.} Given that $\bar{X}_k$ is the MLE of $p_{k}$ and the quantum amplitude is
$a = \sin^2{f_k(p_k)}$ given by Eq.~\eqref{eq:mul_sol},
where $f_k(\cdot)$ is defined as
\begin{align}\label{eq:fk}
    f_k(x) \coloneqq
    \begin{cases} 
        \frac{1}{2k+1}\left(\arcsin{(\sqrt{x})} + l(k) \cdot \frac{\pi}{2}\right) & \text{if \(l(k)\) is even}, \\
        \frac{1}{2k+1}\left(\arccos{(\sqrt{x})} + l(k) \cdot \frac{\pi}{2}\right) & \text{if \(l(k)\) is odd},
    \end{cases}
\end{align}
the MLE of $a$ is $\hat{a}=\sin^2{f_k(\bar{X}_k)}$ according to the invariance principle.

Likewise, $\bar{X}_k$ is an unbiased estimator of $p_{k}$, and its MSE is readily calculated as  
\begin{align*}
    \text{MSE}_{\bar{X}_k}(p_k)=\frac{1}{N_k}p_k(1-p_k).
\end{align*}
The AMSE of $\hat{a}$ can be derived by the delta method as follows:
\begin{align*}
    \text{AMSE}_{\hat{a}}(a) & =\frac{1}{N_k}p_k(1-p_k)[(\sin^2f_k(p_k))']^2 \\
    & = \frac{1}{N_k}p_k(1-p_k)\cdot4\sin^2{f_k(p_k)}\cos^2{f_k(p_k)}[f_k'(p_k)]^2 \\
    & = \frac{1}{N_k}p_k(1-p_k)\cdot 4a(1-a)[f_k'(p_k)]^2.
\end{align*}
Since the derivative of $f_k(\cdot)$ is given by 
\begin{align}\label{eq:deriv_fk}
    f_k'(x) =
    \begin{cases} 
        \frac{1}{2(2k+1)}\frac{1}{\sqrt{x(1-x)}} & \text{if \(l(k)\) is even}, \\
        -\frac{1}{2(2k+1)}\frac{1}{\sqrt{x(1-x)}} & \text{if \(l(k)\) is odd}, 
    \end{cases}
\end{align}
we obtain the AMSE expression given in Theorem~\ref{appendix_thm:knownbase}, namely,   
\begin{align*}
    \text{AMSE}_{\hat{a}}(a) = \frac{1}{(2k+1)^2}\frac{1}{N_k}a(1-a).
\end{align*}
The quantum sample complexities of Classical QAE (Appendix Theorem~\ref{appendix_thm:ClassicalQAE}) and Amplified Amplitude Estimation (Appendix Theorem~\ref{appendix_thm:knownbase}) then follow as special cases.

\subsection{Schedule Number \texorpdfstring{$T$}{T} of Base-3 Hybrid Scheduled IQAE}

\label{appendix:proof_T}

{\em Proof of Eq.~\eqref{eq:T}.} Let $\varepsilon_T$ denote the radius of the confidence interval for $\theta$ at the end of the algorithm, where the confidence interval must be fully contained within one of the reference intervals at Stage $T-1$. Since each reference interval at Stage $T-1$ has length $\frac{1}{3^T}\frac{\pi}{2}$, we have
\begin{align*}
    2\varepsilon_t \leq \frac{1}{3^T}\frac{\pi}{2}.
\end{align*}
The Taylor expansion and definition $a=\sin^2\theta$ imply that the radius $\varepsilon$ of the final confidence interval for $a$ satisfies
\begin{align*}
    \varepsilon = \left(2\sqrt{a(1-a)}+o(1)\right)\varepsilon_t.
\end{align*}
The combination of the inequality for $\varepsilon_t$ and the equation for $\varepsilon$ then yields
\begin{align*}
    T \leq \log_3\left( \frac{\pi}{2}\frac{1}{\varepsilon}\left( \sqrt{a(1-a)}+o(1) \right) \right).
\end{align*}
Thus, the choice
\begin{align*}
    T = \left\lceil \log_3 \left( \frac{\pi}{2} \frac{1}{\varepsilon} \sqrt{a(1-a)}\right) \right\rceil
\end{align*}
is sufficient to achieve the target accuracy $\varepsilon$.

\subsection{Accumulated Quantum Sample Complexity of Base-3 Hybrid Scheduled IQAE}\label{appendix:proof_schedule3}

\begin{lemma}
    \label{appendix_lemma:knownBaseCI}
    Assume the same conditions for the quantum sample complexity of QAE as in Theorem~\ref{appendix_thm:knownbase}. We can construct the following confidence intervals for $\theta$ and $a$:
    \begin{enumerate}
        \item An asymptotic level-\(\alpha\) confidence interval for \(\theta\) can be constructed as
        \begin{align*}
            [\theta_l, \theta_u] = f_k\left({\rm{trunc}}_{[0,1]}\left(\bar{X}_k \pm \frac{1}{\sqrt{N_k}} \sqrt{\bar{X}_k(1-\bar{X}_k)} z_{\frac{\alpha}{2}}\right)\right),
        \end{align*}
        where $f_k(\cdot)$ is defined in Eq.~\eqref{eq:fk} and $z_{\frac{\alpha}{2}}$ is the upper $\frac{\alpha}{2}$ quantile of the standard normal distribution. The approximate number of accesses to the oracle required for the radius to be \(\varepsilon\) is
        \begin{align*}
            N_\text{oracle} = \frac{z^2_{\frac{\alpha}{2}}}{4(2k+1)\varepsilon^2}.
        \end{align*}

        \item An asymptotic level-\(\alpha\) confidence interval for \(a\) can be constructed as
        \begin{align*}
            [a_l, a_u] = \sin^2{([\theta_l, \theta_u])}.
        \end{align*}
        The approximate number of accesses to the oracle required for the radius to be \(\varepsilon\) is 
        \begin{align*}
            N_\text{oracle} = \frac{z^2_{\frac{\alpha}{2}}}{(2k+1)\varepsilon^2}a(1-a).
        \end{align*}
    \end{enumerate}
    
\end{lemma}

{\em Proof of Lemma 1.} Note that the distribution of the sample average $\bar{X}_k$ is
\begin{align*}
    \bar{X}_k \sim \frac{1}{N_{k}}{\rm Bin}\left(N_{k}, p_{k}\right).
\end{align*}    
When $N$ is large, this distribution can be approximated by the following normal distribution: 
\begin{align*}
    \bar{X}_k\sim \mathcal{N}\left(p_{k}, \frac{1}{N_{k}}p_{k}(1-p_{k})\right), 
\end{align*}
and the normal approximation allows us to construct an asymptotic level-\(\alpha\) confidence interval for \(p_k\) as follows: 
\begin{align*}
    {\rm{trunc}}_{[0,1]}\left(\bar{X}_k \pm \frac{1}{\sqrt{N_k}} \sqrt{\bar{X}_k(1-\bar{X}_k)} z_{\frac{\alpha}{2}}\right).
\end{align*}
Consequently, using $\theta = f_k(p_k)$ given by Eq.~\eqref{eq:mul_sol}, we obtain an asymptotic level-\(\alpha\) confidence interval for \(\theta\)  
\begin{align*}
    [\theta_l, \theta_u] = f_k\left({\rm{trunc}}_{[0,1]}\left(\bar{X}_k \pm \frac{1}{\sqrt{N_k}} \sqrt{\bar{X}_k(1-\bar{X}_k)} z_{\frac{\alpha}{2}}\right)\right).
\end{align*}
An application of the Taylor expansion leads to the following radius of the confidence interval for $\theta$: 
\begin{align*}
    \frac{\theta_u-\theta_l}{2} & \approx \frac{1}{\sqrt{N_k}} \sqrt{\bar{X}_k(1-\bar{X}_k)} z_{\frac{\alpha}{2}}\cdot|f_k'(\bar{X}_k)| \\
    & = \frac{z_{\frac{\alpha}{2}}}{2(2k+1)\sqrt{N_k}}.
\end{align*}
Hence, when the radius is equal to $\varepsilon$, the corresponding quantum sample complexity is given by 
\begin{align*}
    N_\text{oracle}&=(2k+1)N_k = (2k+1)\frac{z^2_{\frac{\alpha}{2}}}{4(2k+1)^2\varepsilon^2} = \frac{z^2_{\frac{\alpha}{2}}}{4(2k+1)\varepsilon^2}.
\end{align*} 

Given the constructed confidence interval $[\theta_l, \theta_u]$ for $\theta$, we have a natural asymptotic level-\(\alpha\) confidence interval for $a=\sin^2{\theta}$ 
\begin{align*}
    [a_l, a_u] = \sin^2{([\theta_l, \theta_u])},
\end{align*}
and its radius 
\begin{align*}
    \frac{a_u-a_l}{2} & \approx \frac{\theta_u-\theta_l}{2} \cdot2\sin\theta\cos\theta \\
    & = \frac{z_{\frac{\alpha}{2}}}{(2k+1)\sqrt{N_k}}\sqrt{a(1-a)},
\end{align*}
where the Taylor expansion is used again. If we make the radius equal to $\varepsilon$, the quantum sample complexity is given by 
\begin{align*}
    N_\text{oracle}&=(2k+1)N_k = (2k+1)\frac{z^2_{\frac{\alpha}{2}}}{(2k+1)^2\varepsilon^2}a(1-a) = \frac{z^2_{\frac{\alpha}{2}}}{(2k+1)\varepsilon^2}a(1-a) .
\end{align*}

{\em Proof of Theorem~\ref{thm:schedule3}.} To specify the required confidence at each stage, we set all the confidence levels \(\alpha_t\) equal to \(\frac{\alpha}{T}\), which ensures that the overall confidence level is \(\alpha\) by the union bound. By Lemma~\ref{appendix_lemma:knownBaseCI}, the accumulated sample complexity is then given by 
\begin{align}
    \label{appendix_eq:acc_complexity_decomp}
    \sum_{t=0}^{T-1} N_{\text{oracle},t} & = \sum_{t=0}^{T-1} \frac{z^2_{\alpha_t/2}}{4 K_t \varepsilon_t^2} = \frac{1}{4} z^2_{\frac{\alpha}{2T}} \left(\sum_{t=0}^{T-2} \frac{1}{K_t} \frac{1}{\varepsilon_t^2} + \frac{1}{K_{T-1}} \frac{1}{\varepsilon_{T-1}^2} \right),
\end{align}
where $\varepsilon_t$ is the radius of the interval estimate for $\theta$ obtained at stage $t$. 
 
To bound the first term in the parentheses of (\ref{appendix_eq:acc_complexity_decomp}), we assume that the radius of the confidence interval is sufficiently large compared to the reference intervals (Assumption~\ref{assump:schedule3}), which indicates 
\begin{align*}
    \sum_{t=0}^{T-2} \frac{1}{K_t} \frac{1}{\varepsilon_t^2}
    & \leq \sum_{t=0}^{T-2} \frac{1}{3^t} \frac{16}{\gamma^2\pi^2}3^{2t+2} \\
    & = \frac{144}{\gamma^2\pi^2} \sum_{t=0}^{T-2} 3^t \\
    & = \frac{72}{\gamma^2\pi^2} (3^{T-1} - 1) \\
    & \leq \frac{36}{\gamma^2\pi} \frac{1}{\varepsilon}\sqrt{a(1-a)}.
\end{align*}
The definition of $T$ thus yields
\begin{align*}
    T = \left\lceil \log_3 \left( \frac{\pi}{2} \frac{1}{\varepsilon} \sqrt{a(1-a)}\right) \right\rceil & \Rightarrow T \leq \log_3 \left( \frac{\pi}{2} \frac{1}{\varepsilon} \sqrt{a(1-a)}\right) + 1 \\ & \Rightarrow 3^{T-1} \leq \frac{\pi}{2} \cdot \frac{1}{\varepsilon}\sqrt{a(1-a)}.
\end{align*}

To bound the second term in the parentheses of (\ref{appendix_eq:acc_complexity_decomp}), we derive
\begin{align*}
    \frac{1}{K_{T-1}} \frac{1}{\varepsilon_{T-1}^2} & = \frac{1}{3^{T-1}} \frac{1}{\varepsilon_{T-1}^2} \\
    & \leq \frac{6}{\pi\sqrt{a(1-a)}} \frac{\varepsilon}{\varepsilon_{T-1}^2} \\
    & = \frac{24}{\pi} \frac{1}{\varepsilon} \sqrt{a(1-a)},
\end{align*}
where, again, the second to the last inequality follows from the definition of $T$ as follows: 
\begin{align*}
    T = \left\lceil \log_3 \left( \frac{\pi}{2} \frac{1}{\varepsilon} \sqrt{a(1-a)}\right) \right\rceil \Rightarrow 3^{T} \geq \frac{\pi}{2} \frac{1}{\varepsilon} \sqrt{a(1-a)} \Rightarrow \frac{1}{3^{T-1}}\leq \frac{6\varepsilon}{\pi}\frac{1}{\sqrt{a(1-a)}}.
\end{align*}
An application of the radius of interval estimates for $\theta$ and $a$ in Lemma~\ref{appendix_lemma:knownBaseCI} leads to 
\begin{align*}
    \varepsilon_{T-1} = \frac{\varepsilon}{2\sqrt{a(1-a)}}.
\end{align*}
Hence, we have 
\begin{align*}
    \sum_{t=0}^{T-1} N_{\text{oracle},t} 
    & \leq \frac{1}{4} z^2_{\frac{\alpha}{2T}} \left( \frac{36}{\gamma^2\pi} + \frac{24}{\pi} \right) \frac{1}{\varepsilon}\sqrt{a(1-a)} \\
    & = \frac{3}{\pi} \left( \frac{3}{\gamma^2} + 2 \right) z^2_{\frac{\alpha}{2T}} \frac{1}{\varepsilon}\sqrt{a(1-a)}.
\end{align*}

To evaluate the prefactor, by the upper and lower bounds for the Mills ratio, we approximate \(z_{\alpha}\) as follows: 
\[z_{\alpha} \approx \sqrt{-2 \log\alpha},\] 
or, more rigorously,
\begin{align*}
    z_{\alpha} =\sqrt{-(2+o(1)) \log\alpha},
\end{align*}
where \( o(1) \) represents a term that vanishes as \( \alpha \) approaches zero.  In fact, for $\alpha \leq 0.1$, this approximation is an upper bound
\begin{align*}
    z_{\alpha} <\sqrt{-2 \log\alpha}.
\end{align*}
Hence, we obtain
\begin{align*}
    z^2_{\frac{\alpha}{2T}} < -2 \log\frac{\alpha}{2T} \leq 2 \log\left(\frac{2}{\alpha}\left\lceil \log_3 \left( \frac{\pi}{4\varepsilon} \right) \right\rceil\right).
\end{align*}
Finally, we conclude that the accumulated sample complexity is
\begin{align*}
    \sum_{t=0}^{T-1} N_{\text{oracle},t} \leq \frac{6}{\pi} \left( \frac{3}{\gamma^2} + 2 \right) \frac{1}{\varepsilon} \log\left(\frac{2}{\alpha} \left\lceil \log_3 \left( \frac{\pi}{4\varepsilon} \right) \right\rceil\right)\sqrt{a(1-a)}.
\end{align*}

\subsection{Accumulated Quantum Sample Complexity of Base-3/Base-5 Hybrid Scheduled IQAE}
\label{appendix:proof_schedule35}

{\em Proof of Appendix Theorem~\ref{appendix_thm:schedule35}.} Given the decomposition of the accumulated complexity Eq.~(\ref{appendix_eq:acc_complexity_decomp}), we bound the first term in the parentheses by replacing the assumption that the confidence interval is sufficiently large relative to the base-3 reference intervals (Assumption~\ref{assump:schedule3})  with the requirement that the confidence interval is contained in the base-3 or base-5 reference intervals [Eq.~\eqref{eq:epsilont35}]. This replacement yields the inequality 
\begin{align*}
     \varepsilon_t \geq \frac{1}{2}\frac{1}{K_{t}}\frac{1}{15}\frac{\pi}{2}, \quad\text{for $t = 0, 1, \ldots, T-2$} 
\end{align*}
and thus the bound 
\begin{align*}
    \sum_{t=0}^{T-2} \frac{1}{K_t} \frac{1}{\varepsilon_t^2} \leq \frac{3600}{\pi^2}\sum_{t=0}^{T-2} K_t.
\end{align*}
Note that $\{K_t\}$ is no longer a geometric sequence as in base-3 hybrid schedule IQAE, and recognize that the ratio $\frac{K_{t+1}}{K_{t}}$ can only be 3 or 5. We have 
\begin{align*}
    K_{t}\leq \frac{1}{3} K_{t+1},\quad \text{for all $t=0,\ldots,T-2$},
\end{align*}
\begin{align*}
    K_t \leq \frac{1}{3^{T-1-t}}K_{T-1},\quad \text{for all $t=0,1,\ldots,T-1$}.
\end{align*}
and 
\begin{align*}
    \sum_{t=0}^{T-2}K_t & \leq \sum_{t=0}^{T-2}\frac{1}{3^{T-1-t}}K_{T-1} =K_{T-1}\sum_{t=1}^{T-1}\frac{1}{3^{t}} \leq \frac{1}{2}K_{T-1}. 
\end{align*}
Therefore, we conclude 
\begin{align*}
    \sum_{t=0}^{T-1} N_{\text{oracle},t} 
    & \leq \frac{1}{4} z^2_{\frac{\alpha}{2T}} \left( \frac{1800}{\pi^2} K_{T-1} + \frac{1}{K_{T-1}} \frac{1}{\varepsilon_{T-1}^2}\right) \\
    & = \frac{1}{4} z^2_{\frac{\alpha}{2T}} \left( \frac{1800}{\pi^2} K_{T-1}\varepsilon_{T-1} + \frac{1}{K_{T-1}\varepsilon_{T-1}}\right)\frac{1}{\varepsilon_{T-1}}.
\end{align*}

To express the above accumulated quantum sample complexity in terms of the quantum amplitude and determine $K_{T-1}\varepsilon_{T-1}$, we consider that (i) since the interval estimate for the amplified angle $K_{T-1}\theta$ must lie within a single quadrant, such that $K_{T-1}\varepsilon_{T-1} \leq \frac{1}{2}\frac{\pi}{2}$ and (ii) at the last stage, the interval estimate for the amplified angle cannot be fully contained within any reference interval, otherwise the algorithm would proceed to stage $T$ to achieve a more accurate interval estimate. This argument yields the bound
\begin{align*}
    K_{T-1}\varepsilon_{T-1} \geq \frac{1}{2}\frac{1}{15}\frac{\pi}{2}.
\end{align*}
Due to the monotonicity of $\frac{1800}{\pi^2} K_{T-1}\varepsilon_{T-1} + \frac{1}{K_{T-1}\varepsilon_{T-1}}$ with respect to $K_{T-1}\varepsilon_{T-1}$, the accumulated quantum sample complexity is upper bounded by 
\begin{align*}
    \sum_{t=0}^{T-1}N_{\text{oracle},t} \leq\frac{227}{2\pi}z^2_{\frac{\alpha}{2T}}\frac{1}{\varepsilon_{T-1}}.
\end{align*}
Replacing $z^2_{\frac{\alpha}{2T}}$ by its upper bound and $\varepsilon_{T-1}$ by $\varepsilon$ as in Appendix~\ref{appendix:proof_schedule3}, we derive the accumulated sample complexity\footnote{Note in base-3/base-5 hybrid scheduled IQAE, $T$ is upper bounded by $\left\lceil\log_3 \left( \frac{\pi}{2\varepsilon}\sqrt{a(1-a)} \right)\right\rceil$ because the slowest schedule always selects three as the multiplier.} 
\begin{align*}
    \sum_{t=0}^{T-1}N_{\text{oracle},t}\leq \frac{227}{\pi}\frac{1}{\varepsilon}\log\left(\frac{2}{\alpha}\left\lceil \log_3 \left( \frac{\pi}{4\varepsilon} \right) \right\rceil \right) \sqrt{a(1-a)}.
\end{align*}

\subsection{Preparation of the Prior in Normal-BIQAE}
\label{appendix:proof_approx_prior}

{\em Proof of Appendix Theorem~\ref{appendix_thm:approx_prior}.}
Note the relationship between $p_{k_{t+1}}$ and $p_{k_t}$ in Eq.~\eqref{eq:p2p} 
\[
p_{k_{t+1}} = \varphi(p_{k_t}),
\]
where \(\varphi\) is defined as
\begin{align}\label{eq:map}
\varphi(x) \coloneqq \sin^2\left(K_{t+1} \cdot f_t(x)\right).
\end{align}
The posterior distribution of \(p_{k_t}\) is approximately normal with mean \(\mu_{\text{post}, t}\) and variance \(\sigma_{\text{post},t}^2\), and, for \(x \in (0,1)\),
\begin{align}
    \label{appendix_eq:f_derivative}
    |f'_t(x)| = \frac{1}{2K_t\sqrt{x(1-x)}}.
\end{align} 
An application of the delta method yields that $p_{k_{t+1}}$ approximately follows a normal distribution with mean
\[
\mu_{0, t+1} = \varphi(\mu_{\text{post}, t}) = \sin^2\left(K_{t+1} \cdot f_t(\mu_{\text{post}, t})\right)
\]
and variance
\begin{align*}
    \Delta_{t+1}^2 &= \left[\varphi'(\mu_{\text{post}, t})\right]^2 \cdot \sigma_{\text{post},t}^2 \\
    & = \left\{2 \sin\left[K_{t+1} f_t(\mu_{\text{post}, t})\right] \cos\left[K_{t+1} f_t(\mu_{\text{post}, t})\right] \cdot K_{t+1} f'_t(\mu_{\text{post}, t})\right\}^2 \cdot \sigma_{\text{post},t}^2 \\
    & = \left(\frac{K_{t+1}}{K_t}\right)^2 \sin^2\left[K_{t+1} f_t(\mu_{\text{post}, t})\right] \cos^2\left[K_{t+1} f_t(\mu_{\text{post}, t})\right] \cdot \frac{1}{\mu_{\text{post}, t}(1-\mu_{\text{post}, t})} \cdot \sigma_{\text{post},t}^2 \\
    & = \left(\frac{K_{t+1}}{K_t}\right)^2 \frac{\mu_{0, t+1}(1-\mu_{0, t+1})}{\mu_{\text{post}, t}(1-\mu_{\text{post}, t})} \cdot \sigma_{\text{post},t}^2.
\end{align*}
Thus, the prior distribution of \(p_{k_{t+1}}\) is approximately normal with mean \(\mu_{0, t+1}\) and variance \(\Delta_{t+1}^2\).

\subsection{Quantum Sample Complexity of Stage \texorpdfstring{$t$}{t} in Normal-BIQAE}
\label{appendix:proof_sym2}

{\em Proof of Appendix Theorem~\ref{appendix_thm:sym2}.} Without loss of generality, we assume that the interval estimate for $p_{k_t}$ in Normal-BIQAE as presented in Algorithm~\ref{appendix_alg:ComputeCRI} is contained in $[0,1]$. We will derive an expression for the required sample size $N_t$ to achieve the stopping criterion at stage $t$. 

At stage $t$, when the radius of the interval estimate for $\theta$  reaches $\varepsilon_t$, we have
\begin{align*}
    z_{\alpha_t/2}\cdot\sigma_{{\rm post},t}\cdot \left|f_t'(\mu_{{\rm post},t})\right| = \varepsilon_t,
\end{align*}
which upon substitution of $\left|f_t'(\mu_{{\rm post},t})\right|$ given by Eq.~\eqref{appendix_eq:f_derivative} 
becomes
\begin{align}
    \label{appendix_eq:from_stopping_criterion}
     \frac{z_{\alpha_t/2}^2 }{4K_t^2\varepsilon_t^2}\frac{1}{\mu_{{\rm post},t}(1-\mu_{{\rm post},t})}= \sigma_{{\rm post},t}^{-2}.
\end{align}
Substituting $\sigma_{{\rm post},t}^2$ by the expression as detailed in Algorithm~\ref{appendix_alg:BayesianUpdate} yields
\begin{align*}
     \frac{z_{\alpha_t/2}^2 }{4K_t^2\varepsilon_t^2}\frac{1}{\mu_{{\rm post},t}(1-\mu_{{\rm post},t})}= \frac{1}{\Delta_t^2} + \frac{N_t}{\bar{X}_t\left(1 - \bar{X}_t\right)},
\end{align*}
and a rearrangement of terms leads to the number of shots $N_t$ as follows: 
\begin{align*}
    N_t = \frac{\bar{X}_t\left(1 - \bar{X}_t\right)}{\mu_{{\rm post},t}(1-\mu_{{\rm post},t})}\left(\frac{z_{\alpha_t/2}^2 }{4K_t^2\varepsilon_t^2} - \frac{\mu_{{\rm post},t}(1-\mu_{{\rm post},t})}{\Delta_t^2}\right).
\end{align*}

We derive the sample complexity at stage $t$ by considering the corresponding number of accesses to the oracle $\mathcal{A}$, namely, 
\begin{align*}
    N_{\text{oracle},t} = K_t N_t = \frac{\bar{X}_t\left(1 - \bar{X}_t\right)}{\mu_{{\rm post},t}(1-\mu_{{\rm post},t})}\left(\frac{z_{\alpha_t/2}^2 }{4K_t\varepsilon_t^2} - \frac{\mu_{{\rm post},t}(1-\mu_{{\rm post},t})}{\Delta_t^2}K_t\right). 
\end{align*}
Note in the above expression (i) the first term inside the parentheses is the sample complexity of Normal-IQAE according to Lemma~\ref{appendix_lemma:knownBaseCI}, (ii) the second term inside the parentheses is a positive quantity that reduces the sample complexity, and (iii) the multiplicative factor outside the parentheses accounts for the impact on the interval length when the center is shifted from the sample mean to the posterior mean. 

To show that the sample complexity of Normal-BIQAE is lower than that of IQAE with the standard-schedule, we prove that the added computational cost of the multiplicative factor is outweighed by the reduction in sample complexity brought about by the second term in the parentheses. Specifically, the separation of the first term inside the parentheses and comparison of the multiplicative factor to one lead to the reformulation as follows: 
\begin{align}
    \label{appendix_eq:rewriteN}
    N_{\text{oracle}, t} = \frac{z^2_{\alpha_t/2}}{4K_t}\frac{1}{\varepsilon_t^2} + 
    \left(\frac{\bar{X}_t\left(1-\bar{X}_t\right)}{\mu_{{\rm post},t}\left(1-\mu_{{\rm post},t}\right)} -1\right)
     \frac{z^2_{\alpha_t/2}}{4K_t}\frac{1}{\varepsilon_t^2} -\frac{\bar{X}_t\left(1-\bar{X}_t\right)}{\Delta_t^2}K_t.
\end{align}
To express this quantum sample complexity in terms of the target accuracy, we first factorize the difference between the additional multiplicative factor and one to give
\begin{align*}
    \frac{\bar{X}_t(1-\bar{X}_t)}{\mu_{{\rm post},t}(1-\mu_{{\rm post},t})} -1 & = \frac{(\bar{X}_t-\mu_{{\rm post},t})(1-\bar{X}_t-\mu_{{\rm post},t})}{\mu_{{\rm post},t}(1-\mu_{{\rm post},t})} \\
    & = \frac{1}{\mu_{{\rm post},t}(1-\mu_{{\rm post},t})}  \left(\bar{X}_t- \frac{\frac{\mu_{0, t}}{\Delta_t^2} + \frac{\bar{X}_t}{\hat{\sigma}_t^2 / N_t}}{\frac{1}{\Delta_t^2} + \frac{1}{\hat{\sigma}_t^2 / N_t}}\right) (1-\bar{X}_t-\mu_{{\rm post},t}) \\
    & = \frac{1}{\mu_{{\rm post},t}(1-\mu_{{\rm post},t})}  \frac{\frac{\bar{X}_t- \mu_{0, t}}{\Delta_t^2}}{\frac{1}{\Delta_t^2} + \frac{1}{\hat{\sigma}_t^2 / N_t}} (1-\bar{X}_t-\mu_{{\rm post},t}) \\
    & = \underbrace{\frac{1}{\mu_{{\rm post},t}(1-\mu_{{\rm post},t})}\frac{\sigma^2_{{\rm post},t}}{\Delta_t^2}}_{\text{First term} }
    \underbrace{(\bar{X}_t-\mu_{0,t})\cdot (1-\bar{X}_t-\mu_{{\rm post},t})}_{\text{Second term} },
\end{align*}
where we have used in the first step the algebraic identity 
\begin{align}
    \label{eq:simple_algebra}
    a(1-a) - b(1-b) = (a-b)(1-a-b)
\end{align}
and in the second and fourth steps the Bayesian updates $\mu_{{\rm post},t}$ and $\sigma^2_{{\rm post},t}$ as detailed in Algorithm~\ref{appendix_alg:BayesianUpdate}. By the expression for $\sigma_{\text{post},t}$ derived from the stopping criterion Eq.~(\ref{appendix_eq:from_stopping_criterion}), we obtain the first term
\begin{align*}
    \frac{\sigma^2_{{\rm post},t}}{\mu_{{\rm post},t}(1-\mu_{{\rm post},t})\Delta_t^2}= \frac{4K^2_t\varepsilon^2_t}{z^2_{\alpha_t/2}\Delta_t^2}.
\end{align*}
Since \(\mu_{\text{post},t}\) is a weighted average of \(\bar{X}_t\) and \(\mu_{0,t}\),
the second term is bounded 
by\footnote{Note the inequality in the first step holds regardless of the relationship between \(\bar{X}_t\) and \(\mu_{0,t}\), as no matter whether \(\bar{X}_t \geq \mu_{0,t}\) 
({\em i.e.}, \(\mu_{\text{post},t}\) lies between \(\mu_{0,t}\) and \(\bar{X}_t\)) or \(\bar{X}_t < \mu_{0,t}\) ({\em i.e.}, \(\mu_{\text{post},t}\) lies between \(\bar{X}_t\) and \(\mu_{0,t}\)), the second term attains a maximum value  
    \(
    (\bar{X}_t - \mu_{0,t}) \cdot (1 - \bar{X}_t - \mu_{0,t}) 
    \)
at \(\mu_{\text{post},t} = \mu_{0,t}\).} 
\begin{align*}
    (\bar{X}_t - \mu_{0,t}) \cdot (1 - \bar{X}_t - \mu_{0,t}) 
    & = \bar{X}_t(1 - \bar{X}_t) - \mu_{0,t}(1 - \mu_{0,t}),
\end{align*}
where the algebraic identity Eq.~(\ref{eq:simple_algebra}) is used again. Combining the two terms together yields
\begin{align*}
    \frac{\bar{X}_t(1-\bar{X}_t)}{\mu_{{\rm post},t}(1-\mu_{{\rm post},t})} -1 \leq \left[\bar{X}_t(1-\bar{X}_t)-\mu_{0,t}(1-\mu_{0,t})\right]\cdot \frac{1}{\Delta_t^2} \frac{4K^2_t\varepsilon^2_t}{z^2_{\alpha_t/2}}.
\end{align*}
Substituting these into the sample complexity at stage $t$ in Eq.~(\ref{appendix_eq:rewriteN}), we obtain 
\begin{align}
    N_{\text{oracle}, t} & \leq \frac{z^2_{\alpha_t/2}}{4K_t}\frac{1}{\varepsilon_t^2} + 
    \left[\bar{X}_t(1-\bar{X}_t)-\mu_{0,t}(1-\mu_{0,t})\right]\cdot \frac{K_t}{\Delta_t^2} 
    -\frac{\bar{X}_t\left(1-\bar{X}_t\right)}{\Delta_t^2}K_t \notag \\
    & = \frac{z^2_{\alpha_t/2}}{4K_t}\frac{1}{\varepsilon_t^2} - \mu_{0,t}(1-\mu_{0,t})\cdot \frac{K_t}{\Delta_t^2}. \label{appendix_ineq:conclusion_part2}
\end{align}
Inequality (\ref{appendix_ineq:conclusion_part2}) demonstrates that the complexity of Normal-BIQAE is lower than that of the standard-scheduled IQAE by a positive term, 
which is in agreement with the hypothesis that injection of a Bayesian approach achieves a complexity reduction that outweighs the extra cost due to retargeting.

To further quantify the magnitude of this improvement, we consider the prior distribution preparation process as outlined in Algorithm~\ref{appendix_alg:PreparePrior}. Using the mean and variance of the posterior distribution of stage $t-1$ according to Appendix Theorem~\ref{appendix_thm:approx_prior}, we show  
\begin{align*}
    \Delta_t^2 =\left(\frac{K_{t}}{K_{t-1}}\right)^2\frac{\mu_{0,t}(1-\mu_{0,t})}{\mu_{post,t-1}(1-\mu_{post,t-1})}\cdot \sigma^2_{\text{post},t-1}.
\end{align*}
As in the derivation of $\sigma_{\text{post},t}$ in Eq.~(\ref{appendix_eq:from_stopping_criterion}), we have 
\begin{align*}
    \frac{\sigma^2_{{\rm post},t-1}}{\mu_{{\rm post},t-1}(1-\mu_{{\rm post},t-1})}= \frac{4K^2_{t-1}\varepsilon^2_{t-1}}{z^2_{\alpha_{t-1}/2}}. 
\end{align*}
Thus, we conclude 
\begin{align*}
    \Delta_t^2 =\left(\frac{K_{t}}{K_{t-1}}\right)^2\mu_{0,t}(1-\mu_{0,t})\cdot \frac{4K^2_{t-1}\varepsilon^2_{t-1}}{z^2_{\alpha_{t-1}/2}}.
\end{align*}
Choosing $\alpha_t=\alpha_{t-1}=\frac{\alpha}{T}$ and applying the union bound according to the convention of IQAE \cite{grinko2021iterative}, we derive the prior variance 
\begin{align*}
    \Delta_t^2 =\mu_{0,t}(1-\mu_{0,t})\cdot \frac{4K^2_t\varepsilon^2_{t-1}}{z^2_{\alpha_{t}/2}}. 
\end{align*}
Finally, via substitution of the expression for $\Delta_{t}^2$ into the quantum sample complexity inequality in Eq.~(\ref{appendix_ineq:conclusion_part2}),  we establish the desired upper bound result 
\begin{align*}
    N_{\text{oracle},t} \leq \frac{z^2_{\alpha_t/2}}{4K_t}\left(\frac{1}{\varepsilon_t^2} - \frac{1}{\varepsilon_{t-1}^2} \right).
\end{align*}

\subsection{Preparation of the Prior in Beta-BIQAE}
\label{appendix:proof_approx_prior_beta}
{\em Proof of Appendix Theorem~\ref{appendix_thm:MLE_convergence_beta_prior}.} 
We prove the consistency of the estimator obtained in PreparePriorBeta Algorithm~\ref{appendix_alg:PreparePriorBeta} by Theorem~{2.2} of ref.~\cite{white1982maximum}, which we reproduce below in the context of Beta-BIQAE for ease of reference: 
\begin{manualtheorem}{2.2}\label{{thm:mismodel}}
    Given the following Assumptions~\ref{assump:A1}--\ref{assump:A3b}, the MLE estimator derived in Algorithm~\ref{appendix_alg:PreparePriorBeta}, $(\hat{\alpha}_R, \hat{\beta}_R)$, almost surely converges to $ (\alpha^*, \beta^*)$ as $R\to \infty$.
\end{manualtheorem}

\begin{manualassumption}{A1}\label{assump:A1}
    The independent random variables $\tilde{Y}_i,\ i=1,\ldots, R$, have a common distribution with a measurable density function $g$.
\end{manualassumption}

\begin{manualassumption}{A2}\label{assump:A2}
    The considered beta family has density functions $q_{\alpha,\beta}(\tilde{y})$ that are measurable in $\tilde{y}$ for every $(\alpha,\beta)$ in $\Theta$, a compact subset of $\mathbb{R}^2$, and continuous in $(\alpha,\beta)$ for every $\tilde{y}$.
\end{manualassumption}

\begin{manualassumption}{A3a}\label{assump:A3a}
      $\mathbb{E}\left[\log g\left( \tilde{Y}_i\right)\right]$ exists and $\lvert \log q_{\alpha, \beta} \left( \tilde{y}\right) \rvert\leq m(\tilde{y})$ for all $(\alpha,\beta)$ in $\Theta$, where $m(\tilde{Y}_i)$ is integrable. 
\end{manualassumption}

\begin{manualassumption}{A3b}\label{assump:A3b}
      $D_{\mathrm{KL}}(g \| q_{\alpha,\beta})$ has a unique minimum at $(\alpha^*, \beta^*)$ in $\Theta$.
\end{manualassumption}

Thus, to employ Theorem~{2.2} of ref.~\cite{white1982maximum} we first verify Assumptions A1--A3 hold for Beta-BIQAE. According to Algorithm~\ref{appendix_alg:PreparePriorBeta}, in Beta-BIQAE, the sample drawn from the posterior distribution obtained at the end of Stage $t$ is
\begin{align*}
    Y_1, Y_2, \cdots, Y_R  \overset{\text{i.i.d.}}{\sim} \operatorname{Beta}\left(a_{\text {post}, t}, b_{\text {post}, t}\right), 
\end{align*}
and the sample $\tilde{Y}_1, \tilde{Y}_2, \cdots, \tilde{Y}_R$ obtained from the mapping in the algorithm is related to the original sample $Y_1, \cdots, Y_R$ as 
\begin{align}
    \label{eq:inv_map}
    Y_i = \sin ^2\left(K_t \cdot f_{t+1}(\tilde{Y}_i)\right), \quad \forall i \in \{1, \dots, R\}.
\end{align}

Assumption~\ref{assump:A1} that the transformed samples are i.i.d.~and possess a density function is therefore satisfied since the original samples ($Y_i$'s) are i.i.d.~and follow a beta distribution, the transformation Eq.~\eqref{eq:inv_map} is differentiable, and the density function of the transformed sample follows as Lemma~\ref{lemma:density_g}: 

\begin{lemma}\label{lemma:density_g}
    The density function of $\tilde{Y}_i$ is
    \begin{align*}
        g(\tilde{y})= & \frac{\left[\sin ^2\left(K_t \cdot f_{t+1}(\tilde{y})\right)\right]^{a_{\text{post}, t}-1} \cdot\left[\cos ^2\left(K_t \cdot f_{t+1}(\tilde{y})\right)\right]^{b_{\text{post}, t}-1}}{B\left(a_{\text{post}, t}, b_{\text{post}, t}\right)} \cdot \frac{K_t}{K_{t+1}} \cdot \frac{1}{\sqrt{\tilde{y}(1-\tilde{y})}} \\
        & \cdot\left|\sin \left(K_t \cdot f_{t+1}(\tilde{y})\right) \cos \left(K_t \cdot f_{t+1}(\tilde{y})\right)\right|.
    \end{align*}
\end{lemma}

{\em Proof of Lemma~\ref{lemma:density_g}.}
    Denote the mapping from $Y_i$ to $\tilde{Y}_i$ by $\varphi$ [given in Eq.~\eqref{eq:map}] such that the density function is expressed as
    $$
    g(\tilde{y})=p_Y\left(\varphi^{-1}(\tilde{y})\right) \cdot\left|\frac{d}{d \tilde{y}} \varphi^{-1}(\tilde{y})\right|,
    $$
    where $p_Y(\cdot)$ is the density function of $Y_i$, and, according to Eq.~\eqref{eq:inv_map}, $\varphi^{-1}$ is the inverse mapping
    $$
    \varphi^{-1}(\tilde{y})=\sin ^2\left(K_t \cdot f_{t+1}(\tilde{y})\right).
    $$
    Direct calculations then lead to 
    $$
    \frac{d}{d \tilde{y}}\left[\varphi^{-1}(\tilde{y})\right]=2 \sin \left(K_t \cdot f_{t+1}(\tilde{y})\right) \cdot \cos \left(K_t \cdot f_{t+1}(\tilde{y})\right) \cdot K_t \cdot f_{t+1}^{\prime}(\tilde{y}),
    $$
    which expressed in terms of the derivative of $f_k$ in Eq.~\eqref{eq:deriv_fk} is
    $$
    \begin{aligned}
    & \quad f'_{t+1}(\tilde{y})=f'_{k_{t+1}}(\tilde{y})=(-1)^{l_{t+1}} \cdot \frac{1}{2 K_{t+1}} \cdot \frac{1}{\sqrt{\tilde{y}(1-\tilde{y})}} \\
    & \Rightarrow\left|\frac{d}{d \tilde{y}}\left[\varphi^{-1}(\tilde{y})\right]\right|=\frac{K_t}{K_{t+1}} \cdot \frac{1}{\sqrt{\tilde{y}(1-\tilde{y})}} \cdot\left|\sin \left(K_t \cdot f_{t+1}(\tilde{y})\right) \cdot \cos \left(K_t \cdot f_{t+1}(\tilde{y})\right)\right|.
    \end{aligned}
    $$
    Moreover, since $Y_i$ follows the beta distribution $\operatorname{Beta}\left(a_{\text {post}, t} , b_{\text {post}, t}\right)$,
    $$
    p_Y\left(\varphi^{-1}\left(\tilde{y}\right)\right)=\frac{\left[\varphi^{-1}(\tilde{y})\right]^{a_{\text {post},t}-1} \cdot\left[1-\varphi^{-1}(\tilde{y})\right]^{b_{\text{post}, t}-1}}{B\left(a_{\text {post} ,t}, b_{\text {post}, t}\right)}. 
    $$
    Combination of the expressions for the gradient of the inverse mapping $\varphi^{-1}(\tilde{y})$ and $p_Y\left(\varphi^{-1}(\tilde{y})\right)$ thus yields the expression for the density function of $\tilde{Y}_i$
    $$
    \begin{aligned}
    g(\tilde{y})= & \frac{\left[\sin ^2\left(K_t \cdot f_{t+1}(\tilde{y})\right)\right]^{a_{\text{post}, t}-1} \cdot\left[\cos ^2\left(K_t \cdot f_{t+1}(\tilde{y})\right)\right]^{b_{\text{post}, t}-1}}{B\left(a_{\text{post}, t}, b_{\text{post}, t}\right)} \cdot \frac{K_t}{K_{t+1}} \cdot \frac{1}{\sqrt{\tilde{y}(1-\tilde{y})}} \\
    & \cdot\left|\sin \left(K_t \cdot f_{t+1}(\tilde{y})\right) \cos \left(K_t \cdot f_{t+1}(\tilde{y})\right)\right|.
    \end{aligned}
    $$

Assumption~\ref{assump:A2} is likewise satisfied due to the chosen beta family and parameter space: Specifically, the parameter space $\Theta$ in Theorem~\ref{appendix_thm:MLE_convergence_beta_prior} is a finite rectangle and thus a compact set, and the density function of the beta distribution 
\begin{align*}
    q_{\alpha, \beta}(\tilde{y}) = \frac{1}{B(\alpha, \beta)} \, \tilde{y}^{\alpha - 1} (1 - \tilde{y})^{\beta - 1} 
\end{align*}
is continuous in $(\alpha,\beta) \in (0, \infty)^2$ for every $\tilde{y}$.

To confirm Assumption~\ref{assump:A3a} is satisfied, we first verify the existence of $\mathbb{E}[\log (g(\tilde{Y}_i))]$. According to Lemma~\ref{lemma:density_g}, $g(\tilde{Y_i})$ can be expressed as
\[
\begin{aligned}
& \quad g(\tilde{Y}_i)
= \frac{Y_i^{a_{\text{post}, t}} (1 - Y_i)^{b_{\text{post}, t}}}{B\left(a_{\text{post}, t}, b_{\text{post}, t}\right)} 
\cdot \sqrt{\frac{Y_i(1 - Y_i)}{\tilde{Y}_i(1 - \tilde{Y}_i)}} 
\cdot \frac{K_t}{K_{t+1}} \\
& \Rightarrow \log g(\tilde{Y}_i)
= \left(a_{\text{post}, t} + \frac{1}{2} \right) \log Y_i 
+ \left(b_{\text{post}, t} + \frac{1}{2} \right) \log(1 - Y_i) \\
& \quad\quad\quad\quad\quad\quad
 - \frac{1}{2} \log \tilde{Y}_i 
- \frac{1}{2} \log(1 - \tilde{Y}_i) \\
& \quad\quad\quad\quad\quad\quad
- \log B\left(a_{\text{post}, t}, b_{\text{post}, t}\right)
+ \log \left( \frac{K_t}{K_{t+1}} \right).
\end{aligned}
\]
Direct calculation then shows
$$
\mathbb{E}[\log Y_i]=\psi\left(a_{\text{post}, t}\right)-\psi\left(a_{\text{post},t}+b_{\text{post},t}\right),
$$
where $\psi$ is the digamma function; and since $1-Y_i \sim \operatorname{Beta}\left(b_{\text{post}, t}, a_{\text{post}, t}\right)$,
$$
\mathbb{E}[\log \left(1-Y_i\right)]=\psi\left(b_{\text{post}, t}\right)-\psi\left(a_{\text{post},t}+b_{\text{post},t}\right). 
$$
It is then sufficient to show that $\mathbb{E}[\log \tilde{Y}_i]$ and $\mathbb{E}[\log (1-\tilde{Y}_i)]$ are finite, which is proven in Lemma 2:

\begin{lemma}\label{lemma:A3(a)1}
    $\mathbb{E}[\log \tilde{Y}_i]$ and $\mathbb{E}[\log (1-\tilde{Y}_i)]$ are finite.
\end{lemma}

{\em Proof of Lemma~\ref{lemma:A3(a)1}.}
    By Lemma~\ref{lemma:density_g},
    $$
    \begin{aligned}
    \mathbb{E}[|\log \tilde{Y}_i|] & =\int_0^1|\log \tilde{y}| \cdot g(\tilde{y}) d \tilde{y} \\
    & \leq \frac{-1}{B\left(a_{\text{post}, t}, b_{\text{post}, t}\right)} \int_0^1 \frac{\log \tilde{y}}{\sqrt{\tilde{y}(1-\tilde{y})}} d \tilde{y} \\
    & =\frac{2 \pi \log 2}{B\left(a_{\text {post},t}, b_{\text {post}, t}\right)},
    \end{aligned}
    $$
   and 
    $$
    \mathbb{E}[|\log (1-\tilde{Y}_i)|] \leqslant \frac{-1}{B\left(a_{\text {post}, t}, b_{\text{post}, t}\right)} \cdot \int_0^1 \frac{\log (1-\tilde{y})}{\sqrt{\tilde{y}(1-\tilde{y})}} d \tilde{y} = \frac{2 \pi \log 2}{B\left(a_{\text {post},t}, b_{\text{post},t}\right)} .   
    $$

We then confirm the remainder of Assumption~\ref{assump:A3a} holds: 
Given that 
$$
|\log q_{\alpha,\beta}(\tilde{y})| \leqslant|\alpha-1| \cdot|\log \tilde{y}|+|\beta-1| \cdot|\log (1-\tilde{y})|+|\log B(\alpha, \beta)|,
$$
we define the control function 
$$
\begin{aligned}
m(\tilde{y})\coloneqq & \max \left\{\left|\alpha_{\min }-1\right|,\left|\alpha_{\max }-1\right|\right\} \cdot|\log \tilde{y}|+ \\
& \max \left\{\left|\beta_{\min }-1\right|,\left|\beta_{\max }-1\right|\right\} \cdot|\log (1-\tilde{y})|+ \\
& \max_{(\alpha, \beta)\in \Theta} \quad    |\log B(\alpha, \beta)|.
\end{aligned}
$$
Note that $m(\tilde{y})$ is well-defined because $B(\alpha, \beta)$ is continuous on $(0,+\infty)^2$ and $\Theta$ is a rectangle in $(0,+\infty)^2$. By Lemma~\ref{lemma:A3(a)1}, $m(\tilde{Y}_i)$ is integrable, and thus $m(\tilde{y})$ is a valid control function.

Finally, we verify Assumption~{A3(b)}. The Kullback–Leibler divergence of $q_{\alpha,\beta}$ from $g$ can be expressed as
\begin{align*}
    D_{\mathrm{KL}}(g \| q_{\alpha,\beta})=\int g(\tilde{y}) \log{g(\tilde{y})} \, d\tilde{y} - \int g(\tilde{y}) \log{q_{\alpha,\beta}(\tilde{y})} \, d\tilde{y}, 
\end{align*}
where $g$ does not depend on $(\alpha,\beta)$. Thus, Assumption~{A3(b)} is equivalent to the assumption that
$$h(\alpha,\beta) \coloneqq -\mathbb{E}[\log q_{\alpha,\beta}(\tilde{Y}_i)]$$ 
has a unique minimum in $\Theta$, which we prove as follows: 

Since
$$
\begin{aligned}
h(\alpha,\beta)=\log B(\alpha, \beta)-\mathbb{E}[\log \tilde{Y}_i] \cdot(\alpha-1)-\mathbb{E}[\log (1-\tilde{Y}_i)](\beta-1),
\end{aligned}
$$
the finiteness of $h(\alpha, \beta)$ is guaranteed by Lemma~\ref{lemma:A3(a)1}; according to Lemma~\ref{lemma:strict_log_convex} below, $\log B(\alpha, \beta)$ is strictly convex on $(0,+\infty)^2$; and since the remaining terms are affine, the function $h(\alpha,\beta)$ is strictly convex on $(0,+\infty)^2$. The unique minimizer of $h(\alpha, \beta)$ over $\Theta$ is thus implied by the fact that $\Theta$ is a closed compact set. 

\begin{lemma}\label{lemma:strict_log_convex}
    $B(\alpha, \beta)$ is strictly log-convex on $(0,+\infty)^2$.
\end{lemma}
\begin{proof}
    Let $\alpha_1, \beta_1, \alpha_2, \beta_2 \in(0,+\infty), u, v \geqslant 0, u+v=1$.
    We have 
    $$
    \begin{aligned}
    B\left(u\left(\alpha_1, \beta_1\right)+v\left(\alpha_2, \beta_2\right)\right) & =  B\left(u \alpha_1+v \alpha_2, u \beta_1+v \beta_2\right) \\
    &=  \int_0^1 t^{u \alpha_1+v \alpha_2-1}(1-t)^{u \beta_1+v \beta_2-1} d t \\
    & =\int_0^1 t^{u\left(\alpha_1-1\right)+v\left(\alpha_2-1\right)}(1-t)^{u\left(\beta_1-1\right)+v\left(\beta_2-1\right)} d t \\
    & =\int_0^1\underbrace{\left[t^{\alpha_1-1}(1-t)^{\beta_1-1}\right]^u}_{f_1(t)} \cdot \underbrace{\left[t^{\alpha_2-1}(1-t)^{\beta_2-1}\right]^v}_{f_2(t)} d t \\
    & =\int_0^1 f_1(t) f_2(t) d t \\
    & \leq \left[\int_0^1 f_1^{\frac{1}{u}}(t) d t\right]^{u} \cdot\left[\int_0^1 f_2^{\frac{1}{v}}(t) d t\right]^{v} \\
    & = \left[\int_0^1 t^{\alpha_1-1}(1-t)^{\beta_1-1} d t\right]^u \cdot\left[\int_0^1 t^{\alpha_2-1}(1-t)^{\beta_2-1} d t\right]^v \\
    & = B^u\left(\alpha_1, \beta_1\right) \cdot B^v\left(\alpha_2, \beta_2\right),
    \end{aligned}
    $$
    where the sixth step follows from Hölder’s inequality. Therefore, $B(\alpha, \beta)$ is log-convex on $(0,+\infty)^2$.
    
    To show that $B(\alpha, \beta)$ is strictly log-convex, we further note that 
    the equality in the sixth step holds if and only if $\lambda |f_1|^{\frac{1}{u}} = \mu |f_2|^{\frac{1}{v}}$ almost everywhere for some constants $\lambda, \mu\geq 0$ not simultaneously zero. With substitution of the expressions for $f_1$ and $f_2$, this condition becomes, for some constants $\lambda, \mu\geq 0$, not both zero, almost everywhere
    \[
    \begin{aligned}
    & \lambda\, t^{\alpha_1 - 1}(1 - t)^{\beta_1 - 1} = \mu\, t^{\alpha_2 - 1}(1 - t)^{\beta_2 - 1} \\
    \Leftrightarrow\quad & t^{\alpha_1 - \alpha_2}(1 - t)^{\beta_1 - \beta_2} = \text{const} \\
    \Leftrightarrow\quad & \alpha_1 = \alpha_2,\quad \beta_1 = \beta_2.
    \end{aligned}
    \]
    That is, the equality holds if and only if $(\alpha_1, \beta_1) = (\alpha_2, \beta_2)$, which implies that the log-convexity of $B(\alpha, \beta)$ is strict.

\end{proof}

\end{document}